\PassOptionsToPackage{usenames,svgnames}{xcolor}

\documentclass[nonacm]{acmart}

\acmJournal{PACMPL}
\acmVolume{1}
\acmNumber{CONF} 
\acmArticle{1}
\acmYear{2018}
\acmMonth{1}
\acmDOI{} 
\startPage{1}

\setcopyright{none}

\usepackage{prftree}
\usepackage{natbib}
\usepackage{savesym}    
\usepackage{enumerate}
\usepackage[inline]{enumitem} 
\usepackage{tikz-qtree}
\usepackage{amsmath}
\usepackage{thmtools}
\usepackage[normalem]{ulem}
\usepackage{stmaryrd}
\usepackage{wrapfig}
\usepackage{listings}
\usepackage{framed}
\usepackage{mathtools}
\usepackage{url}
\usepackage{thm-restate}
\usepackage{galois}
\usepackage{tikz}
\usetikzlibrary{chains, arrows, positioning, quotes, shapes, automata, backgrounds, fit}
\makeatletter
\tikzset{
  off chain/.code={\def\tikz@lib@on@chain{}}%
}
\usepackage{stackengine}
\usepackage{tabularx}
\usepackage{graphicx}

\DeclareFontEncoding{LS1}{}{}
\DeclareFontSubstitution{LS1}{stix}{m}{n}
\DeclareSymbolFont{symbols2}{LS1}{stixfrak}{m}{n}
\DeclareMathSymbol{\typecolon}{\mathbin}{symbols2}{"25}

\newcommand\sem[1]{\llbracket#1\rrbracket}
\newcommand\set[1]{\{#1\}}

\newcommand{\dom}[1]{\mathsf{dom}( {#1} )}
\newcommand\defeq{:=}

\renewcommand\vv[1]{\overline{ {#1} }} 

\newcommand\OWRes{\ruleName{OW$_\text{R}$}} 
\newcommand\OWRef{\ruleName{OW$_\top$}} 
\newcommand\tto{\twoheadrightarrow} 
\newcommand\univRel{\mathsf{Univ}}
\newcommand\existSort[1]{\mathsf{sorts}_\exists(#1)}
\newcommand\existGoals[1]{\mathsf{goals}_\exists(#1)}
\newcommand\Dhochc{D_{\mathit{Th}}}
\newcommand\Ghochc{G_{\mathit{Th}}}
\newcommand{\trSem}[1]{\langle\!|#1|\!\rangle} 
\newcommand{\rtSem}[1]{\trSem{#1}^{-1}}
\newcommand\HO{{\mathsf{HO}}}
\newcommand\FO{{\mathsf{FO}}}
\newcommand\var{{\mathsf{var}}}
\newcommand\yfo{\yy_\FO}
\newcommand\xho{\xx_\HO}
\newcommand\gphi{G_\varphi}
\newcommand\gnphi{G_{\neg \varphi}}
\newcommand\ga{G_A}
\newcommand\gv{G_\var}

\newcommand{\depth}[1]{\mathsf{d}({#1})}

\newcommand\lang[1]{\mathcal{L}( {#1} )}

\newcommand\tree{\mathcal{T}}

\newcommand{\abs}[2]{\lambda #1.\,#2}
\newcommand{\types}{\vdash}
\newcommand{\ruleName}[1]{\textsf{\small{(#1)}}}

\newcommand{\fv}{\mathsf{FV}}

\renewcommand{\implies}{\Rightarrow}
\newcommand{\truetm}{\mathsf{true}}

\newcommand{\apprx}{\mathcal{V}}

\newcommand{\toClause}[2]{\mathsf{toCl}({#1})({#2})}
\newcommand{\fromClause}[1]{\mathsf{fromCl}({#1})}
\newcommand\basetypes[0]{Q}
\newcommand\hasType[2]{{#1} :: {#2}}
\newcommand\type{\mathsf{type}}
\newcommand\sigmaInt{{\sigma_{\mathsf{t}}}}
\newcommand\preorder{\Theta}

\newcommand\prfAb{\Delta}
\newcommand{\leqa}{\mathrel{\leq_{\mathsf{a}}}}

\newcommand\boolType{\top_o}

\newcommand\typeAppr[1]{\Gamma^\infty(#1)}
\newcommand\upclos[1]{{\uparrow_{\mathsf{a}}}(#1)}

\newcommand\consig{\Sigma_\mathsf{con}}
\newcommand\boundvars{\mathsf{Vars}}
\newcommand\zero{\mathsf{zero}} 
\newcommand\err{\mathsf{err}}

\newcommand\divi{\mathsf{div}}

\DeclarePairedDelimiter{\size}{\lvert}{\rvert}

\newcommand{\order}{\mathsf{ord}}

\stackMath

\newcommand{\HOne}{$\mathcal{H}1$}

\newcommand{\HOMSL}{\textsf{HOMSL}}
\newcommand{\MSL}{\textsf{MSL}}

\newcommand\rew[0]{\triangleright}
\newcommand\res[2]{#1_{|{#2}}}

\newcommand{\uu}{\vv{u}}
\newcommand{\xx}{\vv{x}}
\newcommand{\yy}{\vv{y}}
\newcommand{\zz}{\vv{z}}
\renewcommand{\ss}{\vv{s}}
\newcommand{\ww}{\vv{w}}

\newcommand{\proves}{\vdash}

\renewcommand{\:}{\mathord{:}}

\newcommand{\Cex}{\mathsf{V}}
\newcommand{\IsClosed}{\mathsf{Closed}}
\newcommand{\IsOpen}{\mathsf{Open}}

\newcommand{\openS}{\mathsf{open}}
\newcommand{\closeS}{\mathsf{close}}
\newcommand{\readS}{\mathsf{read}}
\newcommand{\putStrS}{\mathsf{putStrLn}}
\newcommand{\withFileS}{\mathsf{withFile}}
\newcommand{\hGetContentsS}{\mathsf{hGetContents}}
\newcommand{\openhdl}{\mathsf{o}}

\newcommand{\closedhdl}{\mathsf{c}}
\newcommand{\ReadModeS}{\mathsf{ReadMode}}

\newcommand{\Exists}{\mathsf{Ex}}
\newcommand{\PredS}{\mathsf{Pred}}
\newcommand{\withFileSA}{\mathsf{withFile}_1}
\newcommand{\withFileSB}{\mathsf{withFile}_2}
\newcommand{\withFileSC}{\mathsf{withFile}_3}

\newcommand{\actSB}{\mathsf{act}_2}

\newcommand{\putContS}{\mathsf{putCont}}
\newcommand{\idS}{\mathsf{id}}

\savesymbol{comment}
\usepackage[final]{commenting}
\declareauthor{sr}{Steven}{blue}
\authorcommand{sr}{comment}
\declareauthor{jj}{Jerome}{cyan}
\authorcommand{jj}{comment}
\declareauthor{ej}{Eddie}{red}
\authorcommand{ej}{comment}
\declareauthor{ch}{Change}{blue}
\authorcommand{ch}{comment}

\theoremstyle{acmdefinition}
\newtheorem*{remark}{Remark}
\newtheorem*{notation}{Notation}

\begin{document}

\title{Higher-Order MSL Horn Constraints}         


\author{Jerome Jochems}
\affiliation{
  \department{Department of Computer Science}              
  \institution{University of Bristol}            
  \city{Bristol}
  \country{UK}                    
}
\email{jerome.jochems@bristol.ac.uk}          

\author{Eddie Jones}
\affiliation{
  \department{Department of Computer Science}              
  \institution{University of Bristol}            
  \city{Bristol}
  \country{UK}                    
}
\email{ej16147@bristol.ac.uk}         

\author{Steven Ramsay}
\affiliation{
  \department{Department of Computer Science}              
  \institution{University of Bristol}            
  \city{Bristol}
  \country{UK}                    
}
\email{steven.ramsay@bristol.ac.uk}         

\begin{abstract}
The monadic shallow linear (MSL) class is a decidable fragment of first-order Horn clauses that was discovered and rediscovered around the turn of the century, with applications in static analysis and verification.
We propose a new class of higher-order Horn constraints which extend MSL to higher-order logic and develop a resolution-based decision procedure.
Higher-order MSL Horn constraints can quite naturally capture the complex patterns of call and return that are possible in higher-order programs, which make them well suited to higher-order program verification.
In fact, we show that the higher-order MSL satisfiability problem and the HORS model checking problem are interreducible, so that higher-order MSL can be seen as a constraint-based approach to higher-order model checking.
Finally, we describe an implementation of our decision procedure and its application to verified socket programming.
\end{abstract}



\maketitle



\section{Introduction}

Constraints of various kinds form the basis of many program analyses and type inference algorithms.  
Specifying an analysis as a combination of generating and resolving constraints is very appealing: 
as \citet{aiken1999sets} remarks in his overview paper, \emph{constraints help to separate specification from implementation}, they can \emph{yield natural specifications} and their often rich theory (typically independent from the problem at hand) can \emph{enable sophisticated optimisations} that may not be apparent if stating the analysis algorithm directly.

To realise these advantages, we want classes of constraints that can naturally express important features of the problem domain, that draw upon a well-understood theory and yet have a decidable satisfiability problem.

In this work we propose a new class of constraints that are designed to capture the complex, higher-order behaviours of programs with first-class procedures.
We develop a part of the theory of these constraints and situate them in relation to other  higher-order program analyses.
Finally, we obtain an efficient decision procedure and we describe an application of the constraints to automatic verification of socket programming in a functional programming language.

\subsection{MSL Horn constraints}

Our constraints can be framed as an extension of the well-known \emph{Monadic Shallow Linear} (MSL) Horn constraints to higher-order logic.
Like many natural ideas, MSL constraints were discovered independently in two different communities. 
At CADE'99, \citet{Weidenbach1999} proposed MSL Horn constraints as a natural setting in which to ``combine the benefits of the finite state analysis and the inductive method'', he showed that satisfiability was decidable and described how this class of constraints could be used in an automatic analysis of security protocols.  
Independently, at SAS'02, \citet*{nielson2002normalizable} proposed the \HOne{} class of Horn constraints, and it was later observed by \cite{goubault-larrecq2005} to be an equivalent reformulation of MSL.
The \HOne{} class was originally used for the control flow analysis of the Spi language, but has also been applied to e.g. the verification of cryptosystems written in C.

MSL is the fragment of first-order Horn clauses obtained by restricting predicates to be monadic and restricting the subject of positive literals to be shallow and linear -- that is, the single argument of a predicate in the head of a clause must be either a variable $x$ or a function symbol applied to distinct variables $f(x_1,\ldots,x_n)$.
In practice, because we can view the function symbol $f$ as constructing a tuple $(x_1,\ldots,x_n)$, it is convenient to also allow non-monadic predicates, so long as they are only applied to variables when used positively: $P(x_1,\ldots,x_n)$.
With this concession, all of the following are MSL Horn clauses (we omit the prefix of universal quantifiers in each case):
\[
\begin{array}{c}
\mathsf{Zero}(x) \implies \mathsf{Leq}(x,z) \qquad\!\! \mathsf{Leq}(\mathsf{s}(x),\mathsf{s}(y)) \implies \mathsf{Leq}(x,y) \qquad\!\! \mathsf{Leq}(x,y) \wedge \mathsf{Leq}(y,z) \implies \mathsf{Leq}(x,z) \\[1.5mm]
\mathsf{Black}(\mathsf{leaf}) \qquad \mathsf{Black}(\mathsf{branch}(x,d,z)) \qquad  \mathsf{Black}(x) \wedge \mathsf{Black}(z) \implies \mathsf{Red}(\mathsf{branch}(x,d,z)) \\[1.5mm]
\mathsf{M}(\mathsf{sent}(y,\mathsf{b},\mathsf{pr}(\mathsf{encr}(\mathsf{tr}(y,x,\mathsf{tb}(z)),\mathsf{bt}),\mathsf{encr}(\mathsf{nb}(z),x)))) \wedge \mathsf{Sb}(\mathsf{pr}(y,z)) \implies \mathsf{Bk}(\mathsf{key}(x,y))
\end{array}
\]
Note: there are no syntactical restrictions on the body of clauses.  
As can be seen in the last example, which is taken from Weidenbach's security protocol analysis, atoms in the body may contain terms with arbitrary nesting.

Sets of MSL clauses were shown by Weidenbach to have decidable satisfiability.  
The procedure is an instance of ordered resolution, with a carefully crafted ordering that guarantees terminating saturation.
Essentially the same procedure was rediscovered independently by Goubault-Larrecq, as he sought to construct a more standard procedure than Nielsen, Nielsen and Seidl's original, which was a bespoke kind of constraint normalisation. 

In each case, the authors identify a key, solved form for constraints.  Clauses in this solved form have shape:
$
  Q_1(y_1) \wedge \cdots{} \wedge Q_k(y_k) \implies P(f(x_1,\ldots{},x_m))
$
with $\{y_1,\ldots,y_k\} \subseteq \{x_1,\ldots,x_m\}$.
A set of clauses of this form can straightforwardly be viewed as an alternating tree automaton
and so such clauses are called \emph{automaton clauses}.  
Given as input a set of MSL Horn constraints $\mathcal{C}$, each of the above decision procedures can be viewed as constructing a set of automaton clauses $\mathcal{A}$ equisatisfiable with $\mathcal{C}$, and since $\mathcal{A}$ is essentially a tree automaton, its satisfiability can be effectively determined.

\subsection{Contributions of this paper}

In this work, we propose three different higher-order extensions of MSL constraints: (i) the class \HOMSL($1$) obtained by allowing predicates of higher types but maintaining the limitation of first-order function symbols, (ii) the class \MSL($\omega$) obtained by allowing for function symbols of higher-types but maintaining the limitation of first-order (monadic) predicates and (iii) the class \HOMSL($\omega$) obtained by allowing for both predicates and function symbols of higher type.


\paragraph{I. Reduction to existential-free \MSL($\omega$)}
Our first contribution is to show that the satisfiability problem for all of the above classes reduces to the satisfiability problem for a fragment of $\MSL(\omega)$ -- i.e. first-order predicates only -- in which formulas contain no existential quantification.  
We show that existential quantifiers are, in a sense, already definable using higher-order predicates and that higher-order predicates in general can be represented as higher-order functions, whose truth is internalised as a new first-order predicate.

\paragraph{II. Decidability of \MSL($\omega$)}
Our second contribution is to give a resolution-based algorithm for deciding the satisfiability problem of $\MSL(\omega)$.  A key difficulty in generalising the resolution-based decision procedure for the first-order fragment is what to do about negative literals whose subject is headed by a variable $P(x\ s_1\ \cdots{} s_n)$.  Literals of this form simply cannot occur in the first-order case, and allowing (higher-order) resolution on such literals creates clauses that violate one of the cornerstones of the decidability result at first-order, namely that clause heads are shallow.

We introduce a novel kind of higher-order resolution which avoids this phenomenon, but it necessitates a
significantly different notion of automaton formula (solved form), which nevertheless specialises to the existing definition at first-order.
A simple type system for automaton clauses ensures that the level of nesting and the binding structure of variables is in tight correspondence with the type theoretic order of the function symbols involved.
Consequently, as in the first-order case, there can be only finitely many automaton clauses associated with a given 
problem instance, and this forms the backbone of the decidability proof.

\paragraph{III. Interreducibility of \MSL($\omega$), HORS model checking, and intersection (refinement) typeability.}
Although they look superficially complex, it is easy to see that our higher-order automaton clauses, viewed as constraints, are in 1-1 correspondence with a simple kind of intersection types used in higher-order model checking.
Since it is known that this class of intersection types define regular tree languages \cite{broadbent2013saturation}, the name \emph{automaton clauses} is still justified.
Our third contribution is to use this correspondence to moreover give two problem reductions: (i) from \MSL($\omega$) satisfiability to HORS model checking and (ii) from intersection typeability to \MSL($\omega$) satisfiability.
The reduction from HORS model checking to intersection typeability is already well known \cite{kobayashi2013model}, and thus completes the cycle.
We obtain from these reductions that \MSL($\omega$) satisfiability is $(n-1)$-EXPTIME hard for $n \geq 2$ (here $n$ refers to the type theoretic order of the function symbols).  This is the class of problems that can be solved in time bounded by a tower of exponentials of height-$n$.

\paragraph{IV. Application to verified socket programming.}

As proof of concept, we have implemented our prototype decision procedure in Haskell and applied it to the higher-order verification problem of socket usage in Haskell programs.
Interestingly, our extraction of clauses from a Haskell program does not analyse the source code directly.
Instead, the domain-specific language is represented as a typeclass that can be instantiated with a specific ``analysis'' instance that extracts a representation of the program to be passed to our decision procedure in addition to the usual \lstinline{IO} implementation. 
The principle advantage of this approach is that the source code need not be present, making way for library functions to appear freely.
Furthermore, it is easier to implement and maintain as orthogonal concerns in the source code need not be explicitly handled that have no natural encoding with clauses.
As far as we are aware, this technique is novel and could be fruitfully applied to other domains.

\changed[ch]{
\subsection{Wider significance}

In first-order automated program verification there is a consensus around first-order logic as foundation, to the benefit of the field.  On the one hand, ideas from first-order logic, such as interpolants, the Horn fragment, abduction, resolution and so on have found an important place in automated verification.  On the other, first-order logic provides a common vocabulary with which to understand the automated verification landscape \emph{conceptually}, and locate different technologies.  

However, in higher-order automated program verification there are only a disparate collection of formalisms: refinement types \cite{rondon-et-al-pldi2008,zhu-jagannathan-vmcai2013,vazou-et-al-esop2013,vazou-et-al-icfp2015,terauchi-popl2010,unno-kobayashi-PPDP2009}, higher-order grammars/automata \cite{kobayashi2013model,hague-et-al-lics2008,kobayashi-ong-lics2009,ramsay-et-al-popl2014,sal-wal-mscs2016}, fixpoint logic \cite{visvis-concur2004,kobayashi2021overview,bll-fi2021}, and many others.  Moreover, the procedures involved are often bespoke and difficult to relate to techniques that are standard in first-order verification.

This paper is part of a larger effort to establish an analogous foundation for higher-order automated program verification based on higher-order logic \cite{burn2018hochc,OngWagner2019,JochemsThesis,burn2021datalog}.  We have shown that a standard technique from first-order automated reasoning, namely saturation under resolution, is effective at higher-order, and even forms a decision procedure for the higher-order extension of MSL.  Furthermore, in combination with our correspondence between intersection types and higher-order automaton clauses, this sets up a pattern for understanding type-based approaches to verification more generally.  

Our interreducibility results allow us to situate higher-order model checking (also known as HORS model checking) conceptually, within the HOL landscape.  This is beneficial because HORS model checking, although influential, is a set of techniques for solving a somewhat exotic problem.  The safety version of the problem asks to decide if the value tree determined by a certain kind of higher-order grammar is accepted by a B\"uchi tree automaton with a trivial acceptance condition \cite{kobayashi2013model}.  Thanks to the results of this paper, this form of higher-order model checking can be located more simply as ``a group of decision procedures for the MSL fragment''.  We hope this will make this topic more accessible to the wider verification community, since monadic, shallow, and linear restrictions are well understood even by non-experts on higher-order program verification.

Our results also make an interesting connection with work on constructive logic and logic programming.  Our higher-order automaton formulas are a special form of higher-order hereditary Harrop clauses (HOHH), lying in the intersection of HOHH goal and definite formulas.  The fact that we have shown them to play an essential role in characterising satisfiability for this class of higher-order \emph{Horn} clauses sheds a novel light on the relationship between these two classes of formulas, which have been instrumental in the work of Miller and his collaborators \cite{miller_nadathur_2012,miller-etal1991uniform}.  Furthermore, our intersection type and higher-order automaton clause correspondence suggests an alternative to the ``Horn Clauses as Types'' interpretation pioneered by Fu, Komendantskaya, and coauthors \cite{fu-Komendantskaya-lopstr2015,fu-etal-flops2016,fu-Komendantskaya-facs2017,farkaThesis}.  Instead of identifying Horn clauses and types, resolution and proof term construction, our work casts types as monadic predicates, represented as HOHH clauses with a single free variable, and saturation-under-resolution as type inference.
}

\subsection{Outline}
The paper is structured as follows.  In Section~\ref{sec:lang} we introduce higher-order MSL Horn constraints and their proof system, and we give an example of their use in verifying lazy IO computations.  In Section~\ref{sec:reductions} we reduce the provability problem for \HOMSL($\omega$) to the same problem for existential-free \MSL($\omega$).  In Section~\ref{sec:rewriting} we introduce higher-order automaton clauses and use them in deciding satisfiability (via goal-formula entailment).  In Section~\ref{sec:automaton_formulas}, we show that the \MSL($\omega$) satisfiability and HORS model checking are interreducible.  In Section~\ref{sec:application} we describe our implementation and its application.  Finally, in Section~\ref{sec:related} we conclude with a description of related and future work.
\changed[jj]{All proofs are available in the appendix of this paper.}

\section{Higher-Order MSL Horn Formulas}\label{sec:lang}

The logics we consider will make a type-level distinction between predicates and the subjects that they classify.  In the first-order case, the subjects will be terms built from a certain signature of function symbols and constants.

\begin{definition}[Syntax of types]
We consider a subset of simply typed applicative terms.
Starting from the atomic \emph{type of individuals} $\iota$ and the atomic \emph{type of propositions} $o$, the types are given by:
\[
\begin{array}{rrcl}
	\ruleName{Constructor Types} & \gamma &\Coloneqq &\iota \mid \kappa \to \gamma \\
	\ruleName{Constructor Arg Types} & \kappa &\Coloneqq &\gamma \\
\end{array}\quad
\begin{array}{rrcl}
	\ruleName{Predicate Types} & \rho &\Coloneqq &o \mid \sigma \to \rho \\
	\ruleName{Predicate Arg Types} & \sigma &\Coloneqq &\kappa \mid \rho \\
\end{array}
\]
We use $\tau$ as a metavariable for types in general (i.e. that may belong to any of the above classes).
Thus, the constructor types are just the simple types built over a single base type $\iota$, and the predicate types are those simple types with an $o$ in tail position (i.e. that are ultimately constructing a proposition) and whose arguments are either other predicates or constructors.
The introduction of the metavariable $\kappa$ seems unmotivated, but we will later place restrictions on $\kappa$ that differ from those on $\gamma$ more generally.
We define the \emph{order of a type} as follows:
\[
	\order(\iota) = \order(o) = 0 \qquad \order(\tau_1 \to \tau_2) = \max(\order(\tau_1) + 1,\, \order(\tau_2))
\]
After introducing terms below, we will say that the \emph{order of a typed term} is the order of its type.
\end{definition}

\begin{figure}
  \[
		\begin{array}{c}
    	\prftree[r,l]{$x:\tau \in \Delta$}{\ruleName{Var}}{\Delta \types x : \tau}
		\qquad
		\prftree[r,l]{$c:\gamma \in \Sigma$}{\ruleName{Cst}}{
			\Delta \types c : \gamma
		}
		\qquad
		\prftree[r,l]{$P:\rho \in \Pi$}{\ruleName{Pred}}{
			\Delta \types P : \rho
		}
		\\[4mm]
		\prftree[l]{\ruleName{True}}{
			\Delta \types \truetm : o
		}
		\qquad
			\prftree[l]{\ruleName{App}}{
				\Delta \types s : \tau_1 \to \tau_2
			}
			{
					\Delta \types t : \tau_1
			}
			{
			\Delta \types s \, t : \tau_2
			}
		\qquad
		\prftree[l]{\ruleName{And}}{
			\Delta \types G : o
		}
		{
			\Delta \types H : o
		}
		{
			\Delta \types G \wedge H : o
		}
		\\[4mm]
		\prftree[l,r]{$x \notin \dom{\Delta}$}{\ruleName{Ex}}{
			\Delta,\, x:\sigma \types G : o
		}
		{
			\Delta \types \exists x\mathord{:}\sigma.\, G : o
		}
		\qquad
		\prftree[l]{\ruleName{Cl}}{
			\Delta,\,\vv{y\mathord{:}\sigma} \types G : o
		}
		{
			\Delta,\,\vv{y\mathord{:}\sigma} \types A : o
		}
		{
			\Delta \types \forall \vv{y\mathord{:}\sigma}.\, G \implies A
		}
		\end{array}
	\]
	\caption{Typing of terms, goals and clauses}\label{fig:syntax}
	\end{figure}

\begin{definition}[Terms, Clauses and Formulas]\label{def:syntax}
	In all that follows, we assume a finite signature $\Sigma$ of typed function symbols and a finite signature $\Pi$ of typed predicate symbols.  We use $a,b,c$ and other lowercase letters from the beginning of the Roman alphabet to stand for function symbols and $P,Q,R$ and so on to stand for predicates.  Function symbols have types of shape $\gamma$ and predicate symbols have type of shape $\rho$.

	\paragraph{Terms}
	We assume a countably infinite set of variables. We use lowercase letters $x,y,z$ and so on to stand for variables.
  A \emph{type context}, typically $\Delta$, is a finite, partial function from variables to types.  
	Then \emph{terms}, typically $s,t,u$, are given by the following grammar:
  \[
		\begin{array}{rrcl}
		\ruleName{Term} & s,\,t,\,u &\Coloneqq& x \mid c \mid P \mid s\,t
		\end{array}
	\]
	We will only consider those terms that are well typed according to the system specified in Figure~\ref{fig:syntax}, \changed[jj]{where \ruleName{Ex} instantiates one variable at a time for convenience.}
	
	Since we work in higher-order logic, the syntactic category of terms includes both objects that we think of as predicates (which have type $\rho$) and those that we think of as strictly the subjects of predicates (which have type $\gamma$).  
	To make the discussion easier, we will typically refer to terms of the former type as predicates and terms of the latter type $\gamma$ as trees, or tree constructors.

  The \emph{depth} of a symbol $x$, $c$ or $P$ is $0$ and the depth of an application $s\,t$ is the maximum of the depth of $s$ and the depth of $t$ with 1 added.

	\paragraph{Formulas}
  The \emph{atomic formulas}, typically $A$ and $B$, are just those terms of type $o$.
	We define the \emph{goal formulas} and \emph{definite formulas} by mutual induction using the following grammar:
  \[
		\begin{array}{rrcl}
		\ruleName{Goal Formula} & G,\,H &\Coloneqq& A \mid G \wedge H \mid \exists x\mathord{:}\sigma.\,G \mid \truetm\\
		\ruleName{Definite Formula} & C,\,D &\Coloneqq& \forall \vv{y\mathord{:}\sigma}.\,G \implies P\,\vv{y} \mid \forall \vv{y\mathord{:}\sigma}.\,G \implies P\,(c\,\vv{y}) \mid C \wedge D \mid \truetm
    \end{array}
  \]
  Wherever possible we will omit the explicit type annotation on binders and we write $\fv(X)$ to denote the typed free variables of some term, clause or formula $X$.  We identify formulas up to renaming of bound variables.

  The first two alternatives of the syntactic class of definite formulas are the two kinds of \emph{definite clause} that we consider in this work.  They differ in the shape of the head of the clause ($P\,\vv{y}$ vs $P\,(c\,\vv{y})$), but in both cases the variables $\yy$ are required to be distinct from each other (i.e.~the arguments are \emph{shallow} and \emph{linear}).  This is how we express the MSL restriction in our higher-order setting.

  We will often think of definite formulas rather as sets of definite \emph{clauses} (the conjuncts of the formula), thus legitimising notation such as $(\forall \vv{y}.\,G \implies P\,\vv{y}) \in D$, and this is greatly smoothed by adopting the following conventions: we identify formulas up to the commutativity and associativity of conjunction and the use of $\truetm$ as unit.  

	In the following, we will only consider those goal formulas and clauses that are well typed according to the judgement defined inductively in Figure \ref{fig:syntax}.

\end{definition}

\begin{figure}
  \[
    \begin{array}{c}
			\prftree[l]{\ruleName{T}}{
				D \types \truetm
			}
		\hspace{5mm}
      \prftree[l]{\ruleName{And}}{
        D \types G
      }
      {
        D \types H
      }
      {
        D \types G \wedge H
      }
		\hspace{5mm}
      \prftree[l]{\ruleName{Ex}}{
        D \types G[t/x]
      }
      {
        D \types \exists x.\,G
      }
		\hspace{5mm}
			\prftree[l,r]{$(\forall \vv{y}.\, G \!\Rightarrow\! A) \!\in\! D$}{\ruleName{Res}}{
        D \types G[\vv{s}/\vv{y}]
      }
      {
        D \types A[\vv{s}/\vv{y}]
      }
    \end{array}
  \]
  \caption{Proof system for goal formulas}\label{fig:prf-goal}
\end{figure}

\subsection{Proof system and decision problems}

In this paper we will mostly work with respect to a certain proof system for judgements of the form $D \types G$, that is: from a given set of definite clauses $D$ a given goal formula $G$ follows.

\begin{definition}[Proof System]
  The rules defining the system are given in Figure~\ref{fig:prf-goal}, \changed[jj]{with $[\vv{s}/\vv{y}]$ denoting the substitution of $s_i$ for each $y_i$ in $\yy$, and $[t/x]$ of $t$ for $x$.
  Substitution terms may be higher-order and are assumed to be well typed}.
  Note that the substitution in \ruleName{Res} may be vacuous.
\end{definition}

Most of the proof rules are standard but note that \ruleName{Res} is given its name because it simulates (part of) the role of the resolution rule in clausal presentations of higher-order logic and we will sometimes refer to this rule as performing a ``resolution step''.  We note that the mechanism by which a given definite clause and a corresponding atomic formula interact in the \ruleName{Res} rule is by matching rather than unification, and the reason that this is possible is because we have an explicit existential introduction rule \ruleName{Ex}.

The above remarks are substantiated by the following result which, leveraging results on higher-order constrained Horn clauses due to \citet{OngWagner2019}, shows that this system characterises truth for this fragment of higher-order logic under the standard semantics.  
The result requires that we assume (a) the predicate signature contains a universal relation $\univRel_\rho$ for each type $\rho$ occurring in $D$ or $G$, which is axiomatised by a corresponding universal clause $(\forall \yy.\, \truetm \Rightarrow \univRel_\rho\,\yy)$ in $D$, and (b) that every tree constructor type $\gamma$ that occurs in $D$ or $G$ is inhabited by some closed term.  Clearly, both of these can be satisfied by suitable extensions of the signatures if they are not satisfied already.  We shall call signatures that satisfy these requirements \emph{adequate}.

\begin{restatable}{theorem}{proofSystemCompleteness}
  \label{thm:proof_system}
  Assuming adequate signatures, the proof system is sound and complete.
\end{restatable}

For this reason, we will develop most of our results with respect to the proof system and 
implicitly adhere to the usual definition of the semantics of formulas.

As a corollary of the correspondence between the systems, we also obtain that higher-type existentials do not add any power to the system, and so can be ignored:
\begin{corollary}\label{prop:rho_existential_elim}
	$D \vdash \exists x\mathord{:}\rho.\, G$ if, and only if, $D \vdash G[\univRel_\rho/x]$
\end{corollary}


The proof system is very straightforward to use since, apart from choosing which clause to apply in a \ruleName{Res} step, the rules are syntax directed according to the shape of the goal.  
An example of the use of the proof system is given following the next subsection, in \ref{sec:file-example}.

\changed[jj]{As usual in higher-order logic, we write $D \models G$ to mean ``$G$ holds for every model of $D$''.} 

\paragraph{Decision problems}
Although we have been discussing satisfiability in the introduction, a conjunction of Horn constraints $D \wedge (G \implies \mathsf{false})$ is satisfiable iff $D \not\models G$.  Hence, we can equivalently consider the problem of deciding the entailment problem $D \models G$, and this is more natural in our setup.
Moreover, due to completeness, we can equivalently consider the provability of $D \proves G$.
We state all kinds of problems because we will prefer one or the other in reductions according to the shape of yes-instances and the form of the input.

\begin{definition}[Entailment, Satisfiability and Provability Problems]
	Given a definite formula $D$ and a goal formula $G$ over signatures $\Pi$ and $\Sigma$, the \emph{entailment problem} is to determine $D \models G$, the \emph{satisfiability problem} is to determine $D \not\models G$ and the \emph{provability problem} is to determine $D \proves G$.
\end{definition}

\subsection{Stratification by type-theoretic order}
\label{sec:fragments}

In the forgoing subsection we have given a very liberal account of what it means to extend MSL to higher types, but one can imagine at least two other possible definitions which are just as natural.  First, we may consider a higher-order extension in which tree constructors are allowed to be of higher-type, but predicates are not -- in other words, as in first-order MSL, the subject of a predicate is a term of type $\iota$, but now this term may be constructed internally using constants of higher type.  Second, we may consider a higher-order extension in which predicates are allowed to be higher-order, but tree constructors are not -- in other words, as in first-order MSL, the terms of type $\iota$ are essentially first-order, but now predicates may take other predicates as arguments.

These two fragments and the unrestricted logic of the previous subsection arise naturally from the following stratification according to type-theoretic order.
\begin{definition}[Higher-order fragments]
	The family of fragments \HOMSL($n$), for $n$ drawn from $\mathbb{N} \cup \{\omega\}$, denotes the restriction of the logic of the previous subsection to tree constructor argument types $\kappa$ of order at most $n-1$.  The family of fragments $\MSL($n$)$ additionally restricts predicate argument types $\sigma$ to $\iota$ only (i.e. all predicates are of type $\iota \to o$).
\end{definition}

Here we regard $\omega-1$ as $\omega$ and an index of $0$ as a prohibition on arguments (i.e. one may not construct function types).
Under this stratification, we can recognise the following fragments:

\begin{itemize}
  \item \MSL(0) is Datalog: predicates are first-order and their subjects are (nullary) constants.
  \item \MSL(1) is the first-order MSL fragment.
  \item \MSL($\omega$) is the second of the two higher-order extensions described above: we have first-order predicates whose subjects are trees constructed from an arbitrary higher-order signature.
  \item \HOMSL(0) is higher-order Datalog, as studied by, e.g. \citet{charalambidis2019hodatalog}. 
  \item \HOMSL(1) is the first of the two higher-order extensions described above: we have higher-order predicates over trees defined using only first-order tree constructors.
  \item \HOMSL($\omega$) is the full language described in the previous subsection.
\end{itemize}

\changed[jj]{The \MSL(0) and \MSL(1) fragments live within first-order Horn clauses and follow \citet{miller_nadathur_2012}'s presentation thereof (as \emph{fohc}), while our larger fragments fall within higher-order Horn clauses (as \emph{hohc}):}

\begin{center}
\newcommand{\rotsubseteq}{\rotatebox[origin=c]{180}{\ensuremath\subseteq}}
\tikzset{
  background/.style={rectangle,rounded corners,inner sep=3pt,rounded corners=2mm,opacity=0.7,fill=#1!30},
  inner/.style={rectangle,rounded corners,inner sep=3pt,rounded corners=2mm,opacity=0.7,fill=#1!30},
  outer/.style={rectangle,rounded corners,inner sep=3pt,rounded corners=2mm,opacity=0.7,fill=#1!30}
}
\begin{tikzpicture}[
	every node/.style={inner sep=1pt}
]
	\node(m0) at (0,2) {\MSL(0)};
	\node(m1) at (2,2) {\MSL(1)};
	\node(mdots) at (4,2) {$\dots$};
	\node(momega) at (6,2) {\MSL($\omega$)};
	\node(h0) at (0,1) {\HOMSL(0)};
	\node(h1) at (2,1) {\HOMSL(1)};
	\node(hdots) at (4,1) {$\dots$};
	\node(homega) at (6,1) {\HOMSL($\omega$)};
	\draw (0,2) edge[draw=none] node [sloped, auto=false] {$\subseteq$} (2,2);
	\draw (2,2) edge[draw=none] node [sloped, auto=false] {$\subseteq$} (4,2);
	\draw (4,2) edge[draw=none] node [sloped, auto=false] {$\subseteq$} (6,2);
	\draw (0,1) edge[draw=none] node [sloped, auto=false] {$\subseteq$} (2,1);
	\draw (2,1) edge[draw=none] node [sloped, auto=false] {$\subseteq$} (4,1);
	\draw (4,1) edge[draw=none] node [sloped, auto=false] {$\subseteq$} (6,1);
	\draw (0,1) edge[draw=none] node [sloped, auto=false] {$\rotsubseteq$} (0,2);
	\draw (2,1) edge[draw=none] node [sloped, auto=false] {$\rotsubseteq$} (2,2);
	\draw (6,1) edge[draw=none] node [sloped, auto=false] {$\rotsubseteq$} (6,2);

	\begin{pgfonlayer}{background}
		\node [background={blue},
							fit=(m0) (homega) (h0),
							label=below right:{\emph{hohc}}, 
							draw=blue] {};
		\node [background={red},
							fit=(m0) (m1),
							label=above:{\emph{fohc}}, 
							draw=red] {};
	\end{pgfonlayer}
\end{tikzpicture}
\end{center}


\subsection{Example: constraints for Lazy IO}\label{sec:file-example}

Our motivation is to use higher-order constraints to specify certain higher-order program verification problems, and especially the verification of safety properties for functional programs.
We describe a general approach to using higher-order MSL constraints for verified socket programming in Section~\ref{sec:application}, but let us here consider a different example: verifying the correctness of a lazy IO computation.
Consider the following Haskell expression, which is featured on the Haskell.org wiki as a prototypical example of a mistake due to improper use of lazy IO for any input \cite{haskellwiki}. 
The expression throws a runtime exception for attempting to read from a closed file handle. 
\begin{lstlisting}[label={ex:file-wrong}]
    do  contents <- withFile "test.txt" ReadMode hGetContents
        putStrLn contents
\end{lstlisting}
This code reads the file named ``test.txt'' (line 1) and prints the contents to stdout (line 2).

The problem comes from the interaction between the lines.
The reading of the file is done using the primitive $\mathsf{hGetContents}$, which returns the list of characters read from a handle lazily\footnote{In fact the handle is put into an intermediate \emph{semi-closed} state, but it is not important to this example so, in the interests of simplicity, we will not model it in what follows.}.  
The $\mathsf{hGetContents}$ action is wrapped in the $\withFileS$ combinator, which brackets the execution of $\mathsf{hGetContents}$ between a call to open the handle and a call to close it again.  
Hence, at the point at which the contents of the file are demanded, in line 2, the file handle has already been closed as a result of leaving $\withFileS$, and forcing the lazy list of characters results in attempting to read from this closed handle.

\begin{figure}
	\begin{minipage}{.37\linewidth}
  \begin{align}
      \setcounter{equation}{0}
      p\ [] &\implies \Exists\ p\\
      \exists y.\ p(y \mathord{:} \readS) &\implies \Exists\,p \\
			\Cex(x) &\implies \Cex(\idS\,x) \\
      \Cex(k\,x\,h) &\implies \PredS\,h\,k\,x \\
			\Cex(k\,\openhdl)  &\implies \Cex(\openS\,h\,k) \\
			\Cex(k\,\closedhdl)  &\implies \Cex(\closeS\,h\,k)
  \end{align}
	\vspace{2.7cm}
\end{minipage}
\begin{minipage}{.6\linewidth}
	\begin{align}
		\IsClosed(h)   &\implies \Cex(\readS\ h\ k) \\
		\IsOpen(h) \wedge \Exists\,(\PredS\,h\,k)  &\implies \Cex(\readS\ h\ k) \\
		\Cex(k\ ())) &\implies \Cex(\putContS\ k\ y\ h) \\
		\Cex(x\ h\ (\putContS\ k))  &\implies \Cex(\putStrS\ x\ h\ k) \\
		\Cex(\openS\ h\ (\withFileSA\ f\ k))  &\implies \Cex(\withFileS\, x\, m\, f\, h\, k) \\
		\Cex(f\ h_0\ (\withFileSB\ k)) &\implies \Cex(\withFileSA\ f\ k\ h_0) \\
		\Cex(\closeS\ h_1 (\withFileSC\ y\ k)) &\implies \Cex(\withFileSB\ k\ h_1\ y) \\
		\Cex(k\ y\ h_2) &\implies \Cex(\withFileSC\ y\ k\ h_2) \\
		\IsOpen(h) \wedge \Cex(k\ \openhdl\ \readS)  &\implies \Cex(\hGetContentsS\ h\ k) \\
		\Cex(\putStrS\ x\ h_3\ \idS) &\implies \Cex(\actSB\ x\ h_3) \\
		\truetm &\implies \IsOpen(\openhdl) \\
		\truetm &\implies \IsClosed(\closedhdl) 
\end{align}
\end{minipage}
\caption{Clauses corresponding to the verification of Example~\ref{ex:file-wrong}}\label{ex:file}
\end{figure}



\begin{figure}
	\[
		\prftree[r]{(11)}{
			\prftree[r]{(5)}{
				\prftree[r]{(12)}{
					\prftree[r]{(15)}{
						\prftree{
							\prftree[r]{(17)}{\IsOpen(\openhdl)}
						}
						{
							\prftree[r]{(13)}{
								\prftree[r]{(6)}{
										\prftree[r]{(14)}{
											\prftree[r]{(16)}{
												\prftree[r]{(10)}{
													\prftree[r]{(7)}{
														\prftree[r]{(18)}{\IsClosed(\closedhdl)} 
													}
													{
														\Cex(\readS\ \closedhdl\ (\putContS\ \idS))
													}
												}
												{
													\Cex(\putStrS\ \readS\ \closedhdl\ \idS)
												}
											}
											{
												\Cex(\actSB\ \readS\ \closedhdl)
											}
										}
										{
											\Cex(\withFileSC\ \readS\ \actSB\ \closedhdl)  
										}
								}
								{
									\Cex(\closeS\ \openhdl\ (\withFileSC\ \readS\ \actSB))
								}
							}
							{
								\Cex(\withFileSB\ \actSB\ \openhdl\ \readS)
							}
						}
						{
							\IsOpen(\openhdl) \wedge \Cex(\withFileSB\ \actSB\ \openhdl\ \readS)
						}
					}
					{
						\Cex(\hGetContentsS\ \openhdl\ (\withFileSB\ \actSB))
					}
				}
				{
					\Cex(\withFileSA\ \hGetContentsS\ \actSB \ \openhdl)
				}
			}
			{
				\Cex(\openS\ \closedhdl\ (\withFileSA\ \hGetContentsS))
			}
		}
		{
			\Cex(\withFileS\ \texttt{``test.txt''}\ \ReadModeS\ \hGetContentsS\ \closedhdl\ \actSB)
		}
	\]
	\caption{Proof in the environment given by Figure~\ref{ex:file}}\label{fig:file-proof}
	\end{figure}

An abstraction of the behaviour of this expression, and the primitives and combinators contained therein, can be expressed as a set of higher-order MSL clauses shown in Figure~\ref{ex:file}.
A systematic approach to verifying lazy IO is not a contribution of this work, so it is not essential to understand the way in which they model the situation, since we will use them as a kind of running example, it is worth looking at the encoding in little bit of detail.
\changed[jj]{Note that [] and : are the usual Haskell nullary and binary list constructors, resp., that denote the empty list and list composition.}

The clauses effectively model a version of the above expression in which both global state (the status of the file handle) and control flow (lazy evaluation) are represented explicitly, by threading a state parameter $h$ and passing continuations $k$ respectively.
In addition to the various functions that appear in the source code, but which now expect an additional pair of arguments $h$ and $k$, there are two constants $\openhdl$ and $\closedhdl$, representing the two possible states (open and closed -- recall that we omit semi-closed for simplicity) of the handle, and four predicates $\Cex$, $\Exists$, $\IsOpen$ and $\IsClosed$.  

The idea is that the predicate $\Cex$ (for `\textsf{V}'iolation) is true of its argument $s$ just if $s$ represents an expression that will attempt to read from a closed file handle.
%

The goal $\Cex(\withFileS\ \texttt{"test.txt"}\ \hGetContentsS\ \closedhdl\ (\actSB\ \idS))$ represents the verification problem: does the given expression crash with a closed file-handle violation?
The idea of the representation is as follows.  
Given that we think of every function as taking a file handle and a continuation, we can rephrase the expression as:
\[
  \withFileS\ \texttt{"test.txt"}\ \hGetContentsS\ \closedhdl\ (\lambda x\,h_3.\ \putStrS\ x\ h_3\ (\abs{y}{y}))
\]
This captures via continuation passing style that evaluation must proceed by executing the expression 
\[
	\withFileS\ \texttt{"test.txt"}\ \hGetContentsS
\]
in the initially closed handle state and with continuation $\lambda x\,h_3.\ \putStrS\ x\ h_3\ k$.  
This continuation takes the suspended lazy stream $x$ that is output by $\mathsf{hGetContents}$ and the state of the handle $h_3$ on exit from $\withFileS$, and attempts to print it to stdout before continuing with the remainder of the program, which just returns whichever value is output by $\putStrS$ (which is just unit).
However, since we don't allow for $\lambda$-abstractions in our constraints, we give a $\lambda$-lifted version of the above, with the innermost abstraction given instead by $\idS$ and the outer one given by $\actSB$.

Similarly, clauses (11)--(14) model the bracketing behaviour of $\withFileS$ described above.  An application of $\withFileS$ to a filename $x$ and an action on handle $f$ will cause a violation (when started in a state in which the handle is $h$ and the remaining program to compute is $k$), whenever $\openS\ h\ (\lambda h_0.\ f\ h_0\ (\lambda y\,h_1.\ \closeS\ h_1\ (\lambda h_2.\ k\ y\ h_2)))$ does.
That is, calling $\openS$ (with the same state $h$) to open the (implicit) file, then continuing by running the action $f$, then continuing by calling $\closeS$ and then finally continuing by executing the remainder of the program $k$ (supplying the output $y$ of the action $f$).
The abstractions are lifted to, from left to right, $\withFileSA$, $\withFileSB$ and $\withFileSC$.

Clauses (1) and (2) constrain $\Exists$ to act like a specialised kind of (higher-order) existential quantifier.  $\Exists$ takes a predicate $p$ as input and holds whenever there is some list, of a certain form, that satisfies $p$.  The form of the list models the thunking behaviour of the lazy stream resulting from $\mathsf{hGetContents}$ -- in particular the fact that the tail of the list comprises another call to $\readS$.

A proof of $\Cex(\withFileS\ \texttt{"test.txt"}\ \hGetContentsS\ \closedhdl\ (\actSB\ \idS))$, witnessing the fact that the expression does cause a violation, can be seen in Figure~\ref{fig:file-proof}.
Here, each use of \ruleName{Res} is annotated by the number of the clause as given in Figure~\ref{ex:file}.

\section{From \texorpdfstring{\HOMSL($\omega$)}{\HOMSL(omega)} to Existential-Free \texorpdfstring{\MSL($\omega$)}{\MSL(omega)}}\label{sec:reductions}

\newcommand{\tr}[1]{(\!|#1|\!)}
\newcommand{\rt}[1]{(\!|#1|\!)^{-1}}
\newcommand{\calD}{\mathcal{D}}

The full \HOMSL($\omega$) fragment is a remarkably expressive language with higher-order constructors, predicates, and existentials, allowing a wide range of higher-order verification problems to be expressed in a language that closely matches a functional source program.

In this section, we show that some of that power is illusory: existential quantification is definable using higher-order predicates (Theorem~\ref{thm:MSL(omega)_to_exf_MSL(omega)}) and higher-order predicates are, in a sense, definable already using higher-order function symbols (Theorem~\ref{thm:HOMSL(omega)_to_MSL(omega)}).
Hence, we are able to reduce the solvability problem from \HOMSL($\omega$) to the solvability problem in existential-free \MSL($\omega$).
These reductions are extremely helpful for developing the rest of the results in the paper.

\newcommand\exf[1]{{#1}_{\mathsf{\not \exists}}}
\newcommand\exfD{\exf{D}}
\newcommand\exfFun[1]{\mathsf{exFree}(#1)}
\newcommand\resSort{\mathsf{S}_{\mathsf{sor}}}
\DeclarePairedDelimiter{\exfMap}{\Lbag}{\Rbag}

\begin{figure}
  \begin{minipage}{.2\linewidth}
    \begin{align*}
      \mathsf{P}\ \mathsf{Q}\ x &\implies \mathsf{S}\ x \\
      \mathsf{R}\ y \wedge x\ y &\implies \mathsf{P}\ x\ y \\
      \mathsf{R}\ x &\implies \mathsf{Q}\ (\mathsf{a}\ x) \\
      \truetm &\implies \mathsf{R}\ (\mathsf{a}\ x)
      &\\
      &\\
    \end{align*}
  \end{minipage}
  \begin{minipage}{.4\linewidth}
    \begin{align*}
      \mathsf{T}(\mathsf{p}\ \mathsf{q}\ x) &\implies \mathsf{T}(\mathsf{s}\ x) \\
      \mathsf{T}(\mathsf{r}\ y) \wedge \mathsf{T}(x\ y) &\implies \mathsf{T}(\mathsf{p}\ x\ y) \\
      \mathsf{Q}(z) &\implies \mathsf{T}(\mathsf{q}\ z) \\
      \mathsf{R}(z) &\implies \mathsf{T}(\mathsf{r}\ z) \\
      \mathsf{T}(\mathsf{r}\ x) &\implies \mathsf{Q}(\mathsf{a}\ x) \\
      \truetm &\implies \mathsf{R}(\mathsf{a}\ x)
    \end{align*}
  \end{minipage}
  \begin{minipage}{.35\linewidth}
    \[
  \prftree{
    \prftree{
      \prftree{
        \prftree{\mathsf{R}(\mathsf{a}\ (\mathsf{a}\ \mathsf{c}))}
      }
      {
        \prftree{
          \prftree{
            \prftree{\mathsf{R}(\mathsf{a}\ \mathsf{c})}
          }
          {
            \mathsf{T}(\mathsf{r}\ (\mathsf{a} \mathsf{c}))
          }
        }
        {
          \mathsf{Q}(\mathsf{a}\ (\mathsf{a}\ \mathsf{c}))
        }
      }
      {
        \mathsf{T}(\mathsf{r}\ (\mathsf{a}\ (\mathsf{a}\ \mathsf{c}))) \wedge \mathsf{T}(\mathsf{q}\ (\mathsf{a}\ (\mathsf{a}\ \mathsf{c})))
      }
    }
    {
      \mathsf{T}(\mathsf{p}\ \mathsf{q}\ (\mathsf{a}\ (\mathsf{a}\ \mathsf{c})))
    }
  }
  {
    \mathsf{T}(\mathsf{s}\ (\mathsf{a}\ (\mathsf{a}\ \mathsf{c})))
  }
\]
  \end{minipage}
  \caption{Example of clauses (left) and their transform (center) and a proof (right)}\label{fig:ho-transform}
\end{figure}

\subsection{Elimination of higher-order predicates}

The idea of the first reduction is to simulate higher-order predicates using higher-order function symbols and a new, first-order ``truth'' predicate $\mathsf{T} : \iota \to o$.
Consider the set of \HOMSL($\omega$) clauses, over predicates $\mathsf{P}: (\iota \to o) \to \iota \to o$, $\mathsf{Q}:\iota \to o$, $\mathsf{R}: \iota \to o$ and $\mathsf{S}:\iota \to o$, and function symbols $\mathsf{a}:\iota \to \iota$ and $\mathsf{c}:\iota$ that are shown on the left of Figure~\ref{fig:ho-transform}.
We have that the goal $\mathsf{S}\ (\mathsf{a}\ (\mathsf{a}\ \mathsf{c}))$ is provable from these clauses.  

We will represent each of the predicates $\mathsf{P}$, $\mathsf{Q}$, $\mathsf{R}$ and $\mathsf{S}$ by new function symbols $\mathsf{p}:(\iota \to \iota) \to \iota \to \iota$, $\mathsf{q} : \iota \to \iota$, $\mathsf{r}:\iota \to \iota$ and $\mathsf{s}: \iota \to \iota$ respectively.
Since we have exchanged $o$ everywhere in these types for $\iota$, combinations that were possible involving $\mathsf{P}$, $\mathsf{Q}$ and $\mathsf{R}$ are still possible using their representatives, just with a different type. 
For example, $\mathsf{P}\ \mathsf{Q}\ x$ in the the body of the first clause can be represented as $\mathsf{p}\ \mathsf{q
}\ x$, but note that this is a term of type $\iota$ so, in a sense, we have lost the notion of when the proposition is true.

To recover truth, we install a new predicate $\mathsf{T}$ and formulate the set of \MSL($\omega$) clauses shown in the center of Figure~\ref{fig:ho-transform}.
\changed[jj]{Thus $\mathsf{T}(t)$ is true just if the proposition represented by the tree $t$ is true (according to the representation scheme above).  When $t$ is a first-order predicate application, then its truth may depend on pattern matching in clause heads, and so truth is deferred to the original predicate (e.g.~in the third and fourth clauses).}
For example, we can derive the goal $\mathsf{T}(\mathsf{s}\ (\mathsf{a}\ (\mathsf{a}\ \mathsf{c})))$, which encodes the higher-order goal $\mathsf{S}\ (\mathsf{a}\ (\mathsf{a}\ \mathsf{c}))$, as shown on the right of Figure~\ref{fig:ho-transform}.

\begin{theorem}
  \label{thm:HOMSL(omega)_to_MSL(omega)}
  Provability in \HOMSL($\omega$) reduces to provability in \MSL($\omega$).
\end{theorem}

Thanks to Corollary~\ref{prop:rho_existential_elim}, we assume WLOG that $D$ and $G$ contain no existentials of type $\rho$.

First, from a given \HOMSL($\omega$) constructor signature $\Sigma$ and predicate signature $\Pi$, we construct \MSL($\omega$) signatures $\Sigma^\#$ and $\Pi^\#$.  
Let us write $D \proves_\omega G$ to distinguish proof in the former and $D \proves_1 G$ in the latter.
Then we construct a section $\tr{-}$ that maps formulas of the former into formulas of the latter in such a way that $D \proves_\omega G$ iff $\tr{D} \proves_1 \tr{G}$.

\paragraph{The \MSL($\omega$) signature.}
Let $\Pi_1$ be the subsignature consisting only of the first-order monadic predicates from $\Pi$.
We start by transforming \HOMSL($\omega$) types $\tau$ to \MSL($\omega$) types $\tr{\tau}$.
\[
  \tr{\iota} \coloneqq \iota \qquad \tr{o} \coloneqq \iota \qquad \tr{\sigma \to \rho} \coloneqq \tr{\sigma} \to \tr{\rho} 
\]
We build \MSL($\omega$) signatures $\Sigma^\#$ and $\Pi^\#$ by introducing one additional first-order monadic predicate symbol $\mathsf{T}$ to denote ``truth'', and a new tree constructor $p^\#$ for each predicate symbol $P \in \Pi$:
\[
  \Sigma^\# \coloneqq \{ p^\# : \tr{\rho} \mid P : \rho \in \Pi \} \cup \Sigma
  \qquad \Pi^\# \coloneqq \{ \mathsf{T} : \iota \to o \} \cup \Pi_1
\]

\paragraph{The term transformation.}
Then define $\tr{t}$ by:
\[
  \tr{x} \coloneqq x 
  \qquad \tr{c} \coloneqq c
  \qquad \tr{P} \coloneqq p^\#
  \qquad \tr{s\,t} \coloneqq \tr{s}\,\tr{t}
\]
By some abuse we write $\tr{\vv{s}}$ to denote the pointwise transformation of a vector of terms $\vv{s}$.
We extend this to goal formulas $\tr{G}$ by:
\[
  \tr{\truetm} \coloneqq \truetm
  \qquad \tr{A} \coloneqq \mathsf{T}\,(\tr{A})
  \qquad \tr{G \wedge H} \coloneqq \tr{G} \wedge \tr{H} 
  \qquad \tr{\exists x.\,G} \coloneqq \exists x.\, \tr{G} 
\]
where, by some abuse, we refer to the term-level transformation on the right-hand side of the second equation.  We extend to definite formulas $\tr{C}$ by:
\begin{align*}
  \tr{\truetm} &\coloneqq \truetm \\
  \tr{C \wedge D} &\coloneqq \tr{C} \wedge \tr{D} \\
  \tr{\forall \vv{y}.\, G \implies P\,\vv{y}} &\coloneqq \forall \vv{y}.\, \tr{G} \implies \mathsf{T}\,(p^\#\,\vv{y}) \\
  \tr{\forall \vv{y}.\, G \implies P\,(c\,\vv{y})} &\coloneqq (\forall \vv{y}.\, \tr{G} \implies P\,(c\,\vv{y})) \wedge (\forall \vv{z}.\, P\,z \implies \mathsf{T}\,(p^\#\,z))
\end{align*}
We call the second conjunct of the last case of this definition the \emph{reflection clause}.
Note that the form of head in this case implies that $P : \iota \to o$ in $\Pi$.
There is some obvious redundancy in that the image of the transformation will typically contain many copies of the same reflection clause, but this could be easily avoided if considered undesirable.

Finally, $D \proves G$ iff $\tr{D} \vdash \tr{G}$ completes the reduction from Theorem~\ref{thm:HOMSL(omega)_to_MSL(omega)}.

\subsection{Elimination of existentials}
\label{sec:existentials}

The completeness of our proof system shows (as is usual for Horn logics) that our fragment has the existential witness property, that is: $D \models \exists x.\ G$ iff $D \models G[t/x]$ for some term $t$.
Consequently, we can define existential quantifiers of tree constructor types using higher-order predicates. We introduce a family of new predicate symbols $\exists_\gamma$ indexed by $\gamma$ and constrain them so that they hold of a given predicate $p$ on $\gamma$ whenever $p$ holds for \emph{some} term $t$.

For example, an existential quantifier $\exists_\iota : (\iota \to o) \to o$ on natural numbers, constructed using successor $\mathsf{s}:\iota \to \iota$ and zero $\mathsf{z}:\iota$, can be defined by the two clauses:
\[
  \forall p\:\iota \to o.\ p\ \mathsf{z} \implies \exists_\iota\ p 
  \qquad
  \forall p\:\iota \to o.\ \exists_\iota (\lambda x.\ p\ (\mathsf{s}\ x)) \implies \exists_\iota\ p
\]
However, since we do without $\lambda$-abstraction in our setting, we develop a version of the above with a kind of built-in lambda-lifting \changed[jj]{in which a predicate $\Lambda_{\gamma,G}$ models the $\lambda$-abstraction $\lambda x:\gamma.\,G$.
Thus, an atom $\exists_\iota \,\Lambda_{\iota, G}$ will represent the goal formula $\exists x:\iota.\,G$.}

\begin{theorem}
  \label{thm:MSL(omega)_to_exf_MSL(omega)}
  Provability in \MSL($\omega$) reduces to provability in existential-free \MSL($\omega$).
\end{theorem}

Because the elimination of higher-order predicates does not introduce existentials (Theorem~\ref{thm:HOMSL(omega)_to_MSL(omega)}), it suffices to reduce provability in \MSL($\omega$) to provability in existential-free \HOMSL($\omega$);
given \MSL($\omega$) definite formula $D$ and goal formula $G$ over constructor signature $\Sigma$ and predicate signature $\Pi$, we construct a \HOMSL($\omega$) definite formula $\exfD$ and goal formula $\exf{G}$ over constructor signature $\exf{\Sigma}$ and predicate signature $\exf{\Pi}$ such that $D \vdash G$ if, and only if, $\exfD \vdash \exf{G}$.
Thanks to Corollary~\ref{prop:rho_existential_elim}, we assume WLOG that $D$ and $G$ contain no existentials of type $\rho$.

\paragraph{The existential-free \HOMSL($\omega$) signature.}
From a given \MSL($\omega$) constructor signature $\Sigma$ and predicate signature $\Pi$, we construct \HOMSL($\omega$) signatures $\exf{\Sigma} \defeq \Sigma$ and $\exf{\Pi}$.

Let $\existSort{D \land G}$ contain the sorts of existential variables appearing in $D \land G$ and the arguments of any constructor from \( \Sigma \).
Furthermore, we define $\existGoals{D \land G}$ as all goal formulas \( G' \) such that \( \exists x.\, G' \) appears in \( D \) or \( G \).
For the purpose of the definition, we assume that there is some fixed ordering on variables.
We then define an extended predicate signature $\exf{\Pi}$ as follows:
\begin{align*}
  \exf{\Pi} & \defeq \Pi \cup \set{\exists_\gamma: (\gamma \to o) \to o \mid  \gamma \in \existSort{D \land G} }                                                                                                                      \\
            & \cup \set{\mathsf{Comp}^{i,n}_f : (\gamma \rightarrow o) \rightarrow \gamma_1 \rightarrow \cdots \rightarrow \gamma_i \rightarrow o \mid f : \gamma_1 \rightarrow \cdots \rightarrow \gamma_n \rightarrow \gamma \in \Sigma,\, 0 \leq i \leq n } \\
            & \cup \set{\Lambda_{\gamma,\, H} : \gamma_1 \to \cdots{} \to \gamma_k \to \gamma \rightarrow o \mid \gamma \in \Gamma,\, H \in \existGoals{D \land G},\,\fv(H) \setminus \{x\} = \{x_1,\ldots,x_k\} }
\end{align*}
where $\mathsf{Comp}^{0,n}_f : (\gamma \to o) \to o$, and, in the final summand, we require each $x_i:\gamma_i$.
\changed[jj]{Intuitively, $\mathsf{Comp}^{i,n}_f$ denotes an eventual application of constructor $f$ to $n$ arguments, $n$ being at most the arity of $f$, with $i \leq n$ already provided.}

\paragraph{Existential-free goals.}
We map an \MSL($\omega$) goal formula $G$ to an existential-free \HOMSL($\omega$) counterpart $\exfMap{G}$:
\[
  \exfMap{\truetm} \defeq \truetm
  \qquad \exfMap{A} \defeq A
  \qquad \exfMap{G \land H} \defeq \exfMap{G} \land \exfMap{H}
  \qquad \exfMap{\exists x:\gamma.\,G} \defeq \exists_\gamma\ (\Lambda_{\gamma,\, G}\,x_1\,\cdots{}\,x_k)
\]
where, in the last clause, $\fv(G) \setminus \{x\} = \{x_1,\ldots,x_k\}$, as sequenced by the assumed order.

\paragraph{Existential-free definite formulas.}
We map an \MSL($\omega)$ definite formula $C$ to an existential-free \HOMSL($\omega$) definite formula $\exfMap{C}$ over $(\exf{\Sigma},\exf{\Pi})$:
  \[
    \exfMap{\truetm} \defeq \truetm 
    \qquad \exfMap{C \land D} \defeq \exfMap{C} \land \exfMap{D}
    \qquad \exfMap{\forall \yy.\, G \Rightarrow P\,(c\,\yy)} \defeq (\forall \yy.\, \exfMap{G} \Rightarrow P\,(c\,\yy))
  \]

  \paragraph{The equi-provable existential-free \HOMSL($\omega$) instance.}
  For any \MSL($\omega$) definite formula $D$ and goal formula $G$, we define an equi-provable existential-free \HOMSL($\omega$) instance with definite formula $\exfD$ and goal formula $\exf{G}\defeq \exfMap{G}$:
  \begin{align*}
    \exfD \defeq
     & \left\{ \mathsf{Comp}^{0,n}_f\,v \Rightarrow \exists_\gamma \, v \;\middle|\; f : \gamma_1 \to \dots \to \gamma_n \to \gamma \in \Sigma, \exists_\gamma \in \exf{\Pi} \right\} \\
     & \cup\; \left\{ v\ (f\ x_1\ \cdots\ x_n) \Rightarrow \mathsf{Comp}^{n,n}_f\ v\ x_1\ \cdots\ x_n \;\middle|\; f : \gamma_1 \to \dots \to \gamma_n \to \gamma \in \Sigma \right\} \\
     & \cup\; \left\{ \exists_{\gamma_{i+1}}\,(\mathsf{Comp}^{i+1,n}_f\,v\,x_1\,\cdots\,x_i) \Rightarrow \mathsf{Comp}^{i,n}_f\,v\,x_1\,\cdots\,x_i \;\middle|\; \begin{array}{ll} f : \gamma_1 \to \dots \to \gamma_n \to \gamma \in \Sigma,\\ 0\leq i < n \end{array}\right\}\\
     & \cup\; \left\{ \exfMap{H} \Rightarrow \Lambda_{\gamma,\, H}\ x_1 \cdots{} x_k\, x \;\middle|\; \gamma \in \existSort{D \land G},\, H \in \existGoals{D \land G} \right\} \cup \exfMap{D}
  \end{align*}
  where, in the third summand, when $i=0$, the clause head is $\mathsf{Comp}^{0,n}_f\,v$.

\begin{wrapfigure}{R}{0.46\textwidth}
  \hfill
  \begin{minipage}{0.45\textwidth}
  \centering
  \vspace*{-10pt}
  \[
            \prftree
            { \prftree
              { \prftree
                { \prftree
                  { \prftree
                    { \prftree
                      { \prftree
                        { \prftree
                          { \prftree
                              { P\,(g\,a\,f\,b) \land Q\,(h\,(g\,a\,f)) }
                              { \Lambda_{\iota \to \iota, P\,(x\,b) \land Q\,(h\,x)}\,(g\,a\,f) }  
                            }
                            { \mathsf{Comp}_g^{2,2}\,(\Lambda_{\iota \to \iota, P\,(x\,b) \land Q\,(h\,x)})\,a\,f  }  
                        }
                        { \mathsf{Comp}_f^{0,0}\,(\mathsf{Comp}_g^{2,2}\,(\Lambda_{\iota \to \iota, P\,(x\,b) \land Q\,(h\,x)})\,a) } 
                      }
                      { \exists_{\iota \to \iota}\,(\mathsf{Comp}_g^{2,2}\,(\Lambda_{\iota \to \iota, P\,(x\,b) \land Q\,(h\,x)})\,a) }
                    } 
                    { \mathsf{Comp}_g^{1,2}\,(\Lambda_{\iota \to \iota, P\,(x\,b) \land Q\,(h\,x)})\,a }
                  }
                  { \mathsf{Comp}_a^0\,(\mathsf{Comp}_g^{1,2}\,(\Lambda_{\iota \to \iota, P\,(x\,b) \land Q\,(h\,x)})) }
                }
                { \exists_{\iota}\,(\mathsf{Comp}_g^{1,2}\,(\Lambda_{\iota \to \iota, P\,(x\,b) \land Q\,(h\,x)})) } 
              }
              { \mathsf{Comp}^{0,2}_g\,(\Lambda_{\iota \to \iota, P\,(x\,b) \land Q\,(h\,x)}) }
            }
            { \exists_{\iota \to \iota} \, (\Lambda_{\iota \to \iota, P\,(x\,b) \land Q\,(h\,x)}) }
    \]
    \caption{Example existential-free proof}
    \label{fig:exf_proof}
  \end{minipage}
\end{wrapfigure}
  
Partial instantiation of existential variables is key to eliminating them.
After all, there may be countably many instantiations but only finitely many clauses.
To this end, predicates $\mathsf{Comp}^{i,n}_f$ act as delayed applications of constructor $f$ that come into effect when all $n$ expected arguments to $f$ are fully instantiated (when $i=n$ and the clause headed by $\mathsf{Comp}^{n,n}_f$ applies).

For example, since $g: \iota \to (\iota \to \iota) \to \iota \to \iota$ has tail type $\iota \to \iota$ when given two arguments, we may instantiate $x:\iota \to \iota$ in $\exists x.\, P\,(x\,b) \land Q\,(h\,x)$ with $g\,a\,f$ (for $a:\iota$ and $f:\iota\to\iota$), resulting in the existential-free proof in Figure~\ref{fig:exf_proof}.

It follows that, for goal formulas $G$ that are subexpressions of the given instance, $D \vdash G$ iff $\exfD \vdash \exf{G}$.
Provability of the latter reduces to provability of an existential-free \MSL($\omega$) instance via the elimination of higher-order predicates (Theorem~\ref{thm:HOMSL(omega)_to_MSL(omega)}), which completes the reduction from Theorem~\ref{thm:MSL(omega)_to_exf_MSL(omega)}.

\section{Higher-Order Automaton Clauses and the Decision Procedure}
\label{sec:rewriting}

\newcommand\Exp[1]{\mathsf{exp}_{#1}}

Consider \MSL($\omega$) $\Delta \vdash D$ and $\Delta \vdash G$ over constructor signature $\Sigma$ and predicate signature $\Pi$. 
We aim to decide $D \vDash G$ by rewriting clauses to a solved form we call \emph{(higher-order) automaton formulas}, after their first-order counterparts in \citet{Goubault-Larrecq2002}.


\subsection{Higher-order resolution}

In Weidenbach's original work on MSL and Goubault-Larrecq's later work on \HOne{} \cite{Weidenbach1999,goubault-larrecq2005}, satisfiability is decided by a form of ordered resolution: the given set of MSL clauses is saturated under the ordered resolution rule and satisfiability is determined according to the presence or absence of the empty clause.

There is a higher-order analogue of the resolution rule which also forms the core of a refutationally complete calculus for higher-order (constrained) Horn clauses \cite{OngWagner2019}:
\[
  \prftree[r]{\ruleName{HO-Resolution}}
  {
  G \wedge R\ \vv{s} \implies A
  }
  {
  G' \implies R\ \yy
  }
  {
  G \wedge G'[\vv{s}/\yy] \implies A
  }
\]
This higher-order rule has exactly the same structure as the standard first-order rule (for Horn clauses).
However, as we shall describe below, this form of resolution on its own does not lend itself to a decision procedure for \MSL($\omega$).

In the first-order case, the key to ensuring termination of saturation under resolution is to identify a certain kind of solved form of constraints, which are called \emph{automaton}.
Such clauses have shape: $Q_1(x_{\pi(1)}) \wedge \cdots{} \wedge Q_k(x_{\pi(k)}) \implies P(f(x_1,\ldots{},x_m))$ with $1 \leq \pi(i) \leq m$ for all $i$.
As remarked in the introduction, such formulas are nothing but a clausal representation of alternating tree automata, but for our purposes, there are two features to take note of: (a) they have a depth-1 head and each atom in the body of the clause is depth 0, and (b) there are no existentially quantified variables (variables that occur in the body but not in the head).

The ordering of the first-order resolution calculus is carefully crafted to ensure that the side premise of each resolution inference is automaton.
It is easy to see that a resolution inference between an arbitrary MSL clause and an automaton clause will produce an MSL clause that is \emph{strictly closer to automaton form}, whenever the selected negative literal has depth at least 1\footnote{Negative literals with depth 0 are essentially already solved.}:
\[
  \prftree
  {
    G \wedge P(f(t_1,\ldots,t_n)) \implies A
  }
  {
    Q_1(x_{\pi(1)}) \wedge \cdots{} \wedge Q_k(x_{\pi(k)}) \implies P(f(x_1,\ldots{},x_m))
  }
  {
    G \wedge Q_1(t_{\pi(1)}) \wedge \cdots{} \wedge Q_k(t_{\pi(k)}) \implies A
  }
\]
Since the body of an automaton clauses is required to contain only atoms with depth 0, we can think of the clause in the conclusion as closer to automaton form than the main premise (on the left-hand side) since the new atoms $Q_i(t_{\pi(i)})$ in the body replace an atom $P(f(t_1,\ldots{},t_n))$ of strictly greater depth.

This is also the case for higher-order clauses in \MSL($\omega$) \emph{whenever the selected negative literal is headed by a function symbol}.
However, in higher-order clauses, the selected negative literal may be headed by a variable, and this spells trouble.
Consider, for example, the following resolution inference.
\changed[jj]{Recall that function application is convntionally left associative, so $h\,y_1\,y_2 = (h\,y_1)\,y_2$.}
\[
  \prftree{
  \mathsf{P}(x_1\ (\mathsf{f}\ \mathsf{a}\ x_2)) \implies \mathsf{Q}(\mathsf{g}\ x_1\ x_2)
  }
  {
  \mathsf{R}(y_2) \wedge \mathsf{S}(y_2) \implies \mathsf{P}(\mathsf{h}\ y_1\ y_2)
  }
  {
  \mathsf{R}(\mathsf{f}\ \mathsf{a}\ x_2) \wedge \mathsf{S}(\mathsf{f}\ \mathsf{a}\ x_2) \implies \mathsf{Q}(\mathsf{g}\ (\mathsf{h}\ y_1)\ x_2)
  }
\]
As before, the body of the clause in the conclusion can be viewed as closer to our automaton solved form, but the head of the clause is further away.
In fact, the clause has departed the MSL fragment completely since it no longer has a shallow head!
This is a significant problem because, by inspection, further resolution inferences with this non-MSL clause as the main premise can only produce clauses with a head of the same or even greater depth.

However, resolving on clauses where the selected negative literal is headed by a variable appears inescapable if we insist one of the premises of each resolution inference to be automaton:
\[
  \truetm                                             \implies \mathsf{P}(\mathsf{h}\ y_1\ y_2)
  \qquad
  \mathsf{P}(x_1\ (\mathsf{f}\ \mathsf{a}\ x_2))      \implies \mathsf{Q}(\mathsf{g}\ x_1\ x_2)
  \qquad
  \mathsf{Q}(\mathsf{g}\ (\mathsf{h}\ z)\ \mathsf{a}) \implies \mathsf{false}
\]
In this example, we can obtain a contradiction by resolution, but the only automaton clause is the first one, so there is no choice but to resolve the first and second, which leads to a deep head as above.

Our solution to this problem is to radically rethink the form of automaton clauses in the higher-order setting.  We observe that a clause with a deep head $\mathsf{R}(\mathsf{f}\ \mathsf{a}\ x_2) \wedge \mathsf{S}(\mathsf{f}\ \mathsf{a}\ x_2) \implies \mathsf{Q}(\mathsf{g}\ (\mathsf{h}\ y_1)\ x_2)$ can be thought of as a clause with a shallow head $\mathsf{R}(\mathsf{f}\ \mathsf{a}\ x_2) \wedge \mathsf{S}(\mathsf{f}\ \mathsf{a}\ x_2) \wedge x_1 = h\ y_1 \implies \mathsf{Q}(\mathsf{g}\ x_1\ x_2)$
that contains an additional constraint $x_1 = h\ y_1$ in the body.

Of course, allowing arbitrary equational constraints (and especially at higher type) in the body will lead us immediately outside of a decidable fragment, so we cannot state such constraints directly.
Rather, we ask only that the higher-order variable $x_1$ ``behave like'' $h\ y_1$.
Since $x_1$ and $h\ y_1$ are both functions, the most obvious route to making this precise is to ask that they behave similarly on similar inputs.
Moreover, there is a clear way to define similar, because we are only able to observe the behaviour of terms through the lens of our stock of predicate symbols\footnote{For example, if we only had a single predicate $P$ then all terms $s$ for which $P(s)$ is true are alike, we have no mechanism to write a constraint that distinguishes them.}.

Hence, to ask that $x_1$ behaves as $h\ y_1$ is to ask that $x_1$ satisfies $(\forall y_2.\ R(y_2) \wedge S(y_2) \implies P(x_1\ y_2))$.
Clearly, $h\ y_1$ is an instance of $x_1$ that satisfies this constraint and, we claim, \MSL($\omega$) cannot distinguish between $h\ y_1$ and any other $x_1$ that also satisfies it.

Incorporating this leads to a kind of abductive inference, in which we infer an additional premise:
\[
  \prftree{
  \mathsf{P}(x_1\ (\mathsf{f}\ \mathsf{a}\ x_2)) \implies \mathsf{Q}(\mathsf{g}\ x_1\ x_2)
  }
  {
    \mathsf{R}(y_1) \wedge \mathsf{S}(y_2) \implies \mathsf{P}(\mathsf{h}\ y_1\ y_2)
  }
  {
    \mathsf{R}(\mathsf{f}\ \mathsf{a}\ x_2) \wedge \mathsf{S}(\mathsf{f}\ \mathsf{a}\ x_2) \wedge (\forall y_2.\ \mathsf{R}(y_2) \wedge \mathsf{S}(y_2) \implies \mathsf{P}(x_1\ y_2)) \implies \mathsf{Q}(\mathsf{g}\ x_1\ x_2)
  }
\]

Of course, we have still ended up outside the MSL fragment, but the additional power required to state this form of constraint on $x_1$ seems much less dangerous.
This nested implication is none other than an MSL clause itself -- the head is shallow and linear -- and, moreover, its body is already in the correct form to be automaton.
In fact we show that this form of nested clause is exactly the generalisation of automaton clause that we need in the higher-order setting.

\subsection{Higher-Order automaton formulas}


\begin{figure}
  \[
    \begin{array}{c}
      \prftree[l,r]{$c\:\iota \in \Sigma$}{\ruleName{Fact$_1$}}{
       \types P\ c : o
      }\qquad
      \prftree[l]{\ruleName{Fact$_2$}}{
       x \: \iota \types P\ x : o
      }
      \\[4mm]
      \prftree[l,r]{$f \: \vv{\gamma} \to \iota \in \Sigma$}{\ruleName{ACl$_1$}}{
        \vv{y \: \gamma} \types U : o
      }
      {
        \types \forall \vv{y\mathord{:}\gamma}.\, U \Rightarrow P\,(f\,\vv{y}) : o
      }
      \qquad
     \prftree[l]{\ruleName{ACl$_2$}}{
        \vv{y \: \gamma} \types U : o
     }
     {
      x \: \vv{\gamma} \to \iota \types \forall \vv{y\:\gamma}.\, U \Rightarrow P\,(x\,\vv{y}) : o
     }
    \\[4mm]
    \prftree[l]{\ruleName{True}}
    {
      \Delta \types \truetm : o
    }\qquad
     \prftree[l,r]{$\Delta_1 = \emptyset$ iff $\Delta_2 = \emptyset$}{\ruleName{And}}{
       \Delta_1 \types U : o
     }
     {
       \Delta_2 \types V : o
     }
     {
       \Delta_1 \cup \Delta_2 \types U \wedge V : o
     }
    \end{array}
  \]
  \caption{Typing for automaton clauses and formulas}\label{fig:auto}
\end{figure}

Higher-order automaton formulas allow for the nesting of clauses inside the body of other clauses.  
This makes them more properly a fragment of higher-order hereditary Harrop formulas (HOHH).  
In fact it is easy to see that they are HOHH formulas of a special kind, since they live in the intersection of HOHH goal and definite formulas \cite[for HOHH see e.g.][]{miller-etal1991uniform}.

\begin{definition}[Automaton Formulas]
  \label{def:automaton_clauses}
  Define the \emph{automaton formulas}, typically $U$ and $V$, by the following grammar:
  \[
    \begin{array}{rrcl}
      \ruleName{Automaton Fm} &U,\,V &\Coloneqq& \truetm \mid U \wedge V \mid P\,x \mid P\,c \mid \forall \vv{y\:\gamma}.\,U \implies P\,(f\,\yy) \mid \forall \vv{y\:\gamma}.\,U \implies P\,(x\,\yy) \\
    \end{array}
  \]
  Note: the form $P\,x$ concerns such a particular free $x$ (i.e. it is \emph{not} $\forall x.\, P\,x$).
  We identify formulas and clauses up to renaming of bound variables and we identify up to the commutativity, associativity and idempotence of conjunction, with $\truetm$ as a unit, so that a formula $U$ will be thought of equally well as a set of conjuncts.
  We only consider formulas $\Delta \types U : o$ that are well typed according to the judgement of Figure~\ref{fig:auto}.
\end{definition}

Each clause in $U$ is essentially monadic, and by this we mean that it concerns either a single free variable $x$ or a single constant $c$ or $f$ from the signature.  
This can be seen in the two pairs of rules \ruleName{Fact$_1$}, \ruleName{ACl$_1$} and \ruleName{Fact$_2$}, \ruleName{ACl$_2$}, which have an empty and singleton typing context respectively.  

In automaton formulas, clauses can be nested inside the body of another, and one function of the type system is to ensure some stratification to the nesting.  
In particular, in a closed automaton formula (i.e. without free variables), top-level clauses can only concern constants from the signature and strictly nested clauses can only concern variables introduced by the clause that immediately contains them. To this end, note that the type context on the single premise of \ruleName{ACl$_1$} and \ruleName{ACl$_2$} contains exactly the variables $\vv{y\:\gamma}$ introduced by the universal quantifier prefixing the immediately enclosing clause. 
\changed[jj]{The side condition of \ruleName{And} guarantees that each conjunction of automaton clauses contains subjects from the same nesting level, rejecting ill-formed clauses like (5)-(8) below.}

Since there is no weakening in general in this type system, the context $\vv{y\:\gamma}$ introduced on the premise of \changed[jj]{\ruleName{ACl$_1$} and \ruleName{ACl$_2$}} implies that the body $U$ must contain a constraint concerning each of the variables in $\vv{y\:\gamma}$.  These variables can be distributed to the appropriate conjuncts of $U$ that contain the corresponding constraint by the \ruleName{And} rule.  However, note that, despite the lack of weakening, it is possible to leave a subset of the variables, say $\vv{y'\:\gamma'}$ unconstrained, but formally we must do that by introducing a $\truetm$ conjunct and discharge it by the judgement $\vv{y\:\gamma'} \types \truetm : o$.  Since, in practice, we will typically omit $\truetm$ conjuncts we will consider, for example, the clause $\forall xy.\, P\,x \Rightarrow R\,(a\,x\,y)$ to be well typed with $P:\iota \to o$ and $a: \iota \to \iota \to \iota$ by regarding the body $P\,x$ as sugar for $P\,x \wedge \truetm$.

For example, the following are closed automaton clauses over predicate symbols $P,Q,R : \iota \to o$ and tree constructor symbols $a:\iota$, $b:\iota \to \iota \to \iota$, $c:(\iota \to \iota) \to \iota \to \iota$ and $d:((\iota \to \iota) \to \iota) \to \iota$.
\begin{align}
  \setcounter{equation}{0}
  & \mathrel{\phantom{\Rightarrow}} P\,a  \\
  \forall xy.\,P\,x \wedge Q\,x &\Rightarrow R\,(b\,x\,y) \\
  \forall xy.\,Q\,y \wedge (\forall z.\, Q\,z \Rightarrow R\,(x\,z)) &\Rightarrow P\,(c\,x\,y) \\
  \forall x.\,(\forall y.\, (\forall z.\, P\,z \wedge Q\,z \Rightarrow R\,(y\,z)) \Rightarrow R\,(x\,y)) &\Rightarrow P\,(d\,x)
\end{align}
Note: strictly speaking clause (2) must be constructed as e.g. $\forall xy.\,P\,x \wedge Q\,x \wedge \truetm \Rightarrow R\,(b\,x\,y)$ (with $\truetm$ representing the constraint on $b$'s argument $y$).
However, the following are \emph{not} well formed as \emph{closed} automaton clauses:
\begin{align}
  & \mathrel{\phantom{\Rightarrow}} P\,x  \\
  \forall xy.\,P\,x \wedge Q\,a &\Rightarrow R\,(b\,x\,y) \\
  \forall xy.\,Q\,y \wedge (\forall z.\, Q\,z \wedge R\,y \Rightarrow R\,(x\,z)) &\Rightarrow P\,(c\,x\,y) \\
  \forall x.\,(\forall y.\, (\forall z.\, P\,z \wedge (\forall z'.\, Q\,z' \Rightarrow R\,(y\,z')) \Rightarrow R\,(y\,z)) \Rightarrow R\,(x\,y)) &\Rightarrow P\,(d\,x)
\end{align}
In (5) we have a free variable $x$, yet the clause is supposed to be closed.  In (6) we have an atom $Q\,a$ concerning a constant that appears in a strictly nested position and in (7) we have an atom $R\,y$ that concerns a variable from an outer scope -- predicates that appear in nested clauses can only concern variables, and only variables that are introduced by the clause that immediately contains them.  This is also the problem in clause (8), where the nested clause $(\forall z'.\, Q\,z' \Rightarrow R\,(y\,z'))$ concerns $y$, but $y$ is not a variable introduced by the immediately enclosing clause $(\forall z.\, P\,z \ldots \Rightarrow R\,(y\,z))$, which introduces only $z$.

\begin{notation}
  If $U$ \emph{is automaton wrt} $\vv{y}\,\vv{z}$ (i.e.~$\yy$ and $\zz$ are the free variables in $U$), we write $\res{U}{\vv{y}}$ and $\res{U}{\vv{z}}$ for the partition $U = \res{U}{\vv{y}} \land \res{U}{\vv{z}}$ according to whether the automaton clauses in $U$ contain a (free) variable from $\vv{y}$ or $\vv{z}$.
\end{notation}


Say that an automaton clause $T$ of shape $\forall \yy.\, U \implies P\,(\xi\ \yy)$ (for $\xi$ either a variable or a constant) is \emph{order $n$} just if the type of $\xi$ is order $n$.
Note that the type system ensures that clauses that are nested in the body of an order-$n$ clause are of strictly smaller order.  
In the following, $\exp_n(m)$ denotes a tower of exponentials of height \changed[jj]{$n+1$, with $\exp_1(m) = 2^m$, $\exp_2(m) = 2^{2^m}$, etc.}

\begin{theorem}
  \label{thm:no_of_automaton_clauses}
  Let $k$ denote the largest arity of any function symbol in $\Sigma$ and $\size{\Pi}$ denote the number of predicate symbols in $\Pi$. There are $O(\Exp{n}(k\size{\Pi}))$ clauses $T$ such that $x:\gamma \types T : o$, \changed[jj]{for $\order(\gamma)=n$.}
\end{theorem}
\begin{proof}
  Fix an order-1 variable $x$ and consider an order-1 clause $\forall \yy.\,U \implies P\,(x\, \yy)$.  Since every variable of $\yy$ is of type $\iota$, the body $U$ cannot contain any nested clauses, so every conjunct is of the form $P' y$.  Hence, order-1 clauses coincide with the automaton clauses of first-order MSL, and it is easy to see that there are at most $|\Pi| \cdot 2^{|\Pi|^{k}} = O(2^{k|\Pi|})$ different such, where $k$ is the maximum arity of any function symbol.
  Now consider an order-$(n+1)$ clause $\forall \yy.\,U \implies P\,(x\ \yy)$.  In the worst case, each variable $y$ of $\yy$ is order $n$, and we may assume there are $O(\Exp{n}{(k\size{\Pi})})$ clauses that concern $y$.  Hence, choosing this automaton clause amounts to choosing the predicate $P$ in the head and then, for each variable $y \in \yy$, choosing some subset of the $O(\Exp{n}{(k\size{\Pi})})$ different clauses that can be nested inside $U$.  Hence, we have $O(\Exp{n+1}{(k\size{\Pi})})$ many clauses at most.
\end{proof}

Automaton clauses $T$ that concern a function symbol $f:\gamma$, i.e. for which $\types T : o$ holds, are just automaton clauses of the above form in which we have replaced the variable $x$ by $f$.  Hence:

\begin{corollary}\label{cor:finite-clauses}
  There are finitely many automaton clauses of a given order.
\end{corollary}

\subsection{Decision procedure}

Our decision procedure takes a set of \MSL($\omega$) definite clauses $D$ as input and iteratively rewrites them into automaton form $V$.  By construction, the new set of automaton clauses $V$ is sufficiently strong, though generally weaker than $D$, to entail any goal $G$ that follows from $D$.

Although we have so far been discussing resolution on clauses, the rewriting will be defined only for goals.  This is because rewriting will introduce nested clauses, and it seems easier to reason with nested clauses compositionally \emph{a la} \citet*{miller-etal1991uniform}.


\begin{definition}[Rewriting]\label{def:rew}
  Given an automaton formula $V$, variables $\yy$, and two goal clauses $G$ and $H$, we introduce the rewrite relation $V,\,\yy \vdash G \rew H$ defined by the rules in Figure~\ref{fig:rew}.
  Note that, in \ruleName{Assm}, we implicitly assume that the length of $\zz$ is the same as the length of $\ss$.
\end{definition}

The rules are mostly straightforward, and consist of simulating certain standard logical inference steps directly on the syntax of the formula.
As discussed, the heavy lifting is done by \ruleName{Step} and \ruleName{Assm} which simulate resolution steps on a goal.
The first, \ruleName{Step}, applies when the head symbol of the goal is a function symbol and the second, \ruleName{Assm}, applies when it is a variable.  Thus the latter provides the abductive method of adding an assumption described above.

The automaton formula $V$ is the set of automaton clauses that we are allowed to use when performing resolution steps and the variables $\yy$ that appear before the turnstile are the set of variables that we are willing to make additional assumptions about, via \ruleName{Assm}.
The idea is that, although we are only rewriting goal formulas $V, \yy \vdash G \rew^* H$, we can think of the goal formula $G$ as the body of an \MSL($\omega$) clause $\forall \yy.\, G \Rightarrow A$.
Then the outcome of rewriting, namely $H$, will give us a new clause $\forall \yy.\, H \Rightarrow A$, but we need to be sure we have only made additional assumptions about the top-level universally quantified variables $\yy$.
This is justified by the following theorem and the remarks that follow.

\begin{theorem}[Soundness]\label{thm:soundness-rewriting}
  If $V,\,\yy \vdash G \rew^* H$ then $V \models \forall \yy.\, H \implies G$.
\end{theorem}


\begin{figure}
  \[
    \begin{array}{c}

      \prftree[l,r]{\ruleName{Refl}}
      {
        $P\,x \in V$
      }
      {
        V,\,\vv{y} \vdash P\,x \rew \truetm
      }
      \qquad 
      \prftree[l,r]{\ruleName{Step}}{$(\forall \vv{x}.\,U \implies P\ (f\ \vv{x})) \in V$}
      {
        V,\,\vv{y} \vdash P\ (f\ \vv{s}) \rew U[\vv{s}/\vv{x}]
      } \\[5mm]

      \prftree[l,r]{\ruleName{Assm}}
      {
        $(\forall \xx\zz.\ U \implies P\ (f\ \xx\zz)) \in V$
      }
      {
        V,\,\vv{y} \vdash P\ (y\ \vv{s}) \rew \res{U}{\zz}[\vv{s}/\zz] \wedge (\forall \zz.\, \res{U}{\zz} \implies P\ (y\ \zz))
      } \\[5mm]


      \prftree[r]{\ruleName{AndL}}
      {
        V,\,\vv{y} \vdash G \rew G'
      }
      {
        V,\,\vv{y} \vdash G \wedge H \rew G' \wedge H
      }\qquad

      \prftree[r, l]{\ruleName{Imp}}
      {
        $
          \begin{array}{l}
            G \neq P\ (y\ \zz) \\
            \;\textrm{for any}\ y \in \yy
          \end{array}
        $\!\!
      }
      {
        U \wedge V,\,\yy \vdash G \rew G'
      }
      {
        V,\,\vv{y} \vdash (\forall \vv{z}.\, U \implies G) \rew (\forall \vv{z}.\, U \implies G')
      } \\[5mm]

      \prftree[r]{\ruleName{AndR}}
      {
        V,\,\vv{y} \vdash H \rew H'
      }
      {
        V,\,\vv{y} \vdash G \wedge H \rew G \wedge H'
      } 
      \qquad
      \prftree[l,r]{\ruleName{Scope}}
      {
        $\zz \cap \fv(G) = \emptyset$
      }
      {
        V,\,\vv{y} \vdash \forall \vv{z}.\, U \implies G \rew G
      }
      \\[5mm]
      \prftree[r]{\ruleName{ImpAnd}}
      {
        V,\,\vv{y} \vdash (\forall \vv{z}.\, G_1 \implies G_2 \wedge G_3) \rew (\forall \vv{z}.\,G_1 \implies G_2) \wedge (\forall \vv{z}.\, G_1 \implies G_3)
      } 
    \end{array}
  \]
  \caption{Rewrite system}\label{fig:rew}
\end{figure}

It follows that, if $V, \yy \vdash G \rew^* H$, then $V \land (\forall \yy.\, G \Rightarrow A) \vDash (\forall \yy.\, H \Rightarrow A)$.
In practice, we are only interested in rewriting sequences that start with the body of a non-solved clause and end in an automaton formula, \changed[jj]{giving us new automaton clauses. 
Termination of branches that have successfully reached automaton form is ensured by e.g.~the side conditions of \ruleName{Imp} and \ruleName{Scope}.
The new automaton clauses arising this way unlock additional avenues} to rewrite the remaining \MSL($\omega$) clauses (via \ruleName{Step} and \ruleName{Assm}) and this, in turn, generates more automaton clauses and so on.
It follows from Corollary~\ref{cor:finite-clauses} that, for a given set of definite clauses $D$, the limit, $\apprx(D)$, is a finite object.

\begin{definition}[Canonical solved form]\label{def:canonical-solved-form}
  Define the \emph{canonical solved form}, written $\apprx(D)$, of a set of definite clauses $D$ inductively by the following rule.  \changed[jj]{The base case -- when $V$ is empty -- occurs for definite clauses that are already automaton, like $\forall y_1 y_2.\, P\,y_1 \land R\,y_1 \Rightarrow S\,(f\,y_1\,y_2)$.}
  \[
    \begin{array}{c}
      \prftree[l]{$\left.\begin{array}{r}(\forall \yy.\, G \implies A) \in D\\V,\,\yy \vdash G \rew^* U\end{array}\right|$}{V \subseteq \apprx(D)}{(\forall \yy.\, U \implies A) \in \apprx(D)}
    \end{array}
  \]
\end{definition}

Since every clause that we add to $\apprx(D)$ is weaker than a clause in $D$, we have $\apprx(D) \models G$ implies $D \models G$.
It remains to show the converse, from which we can deduce that to decide satisfiability, it suffices to compute $\apprx(D)$ and then check whether or not the goal follows.

\begin{restatable}[Completeness]{theorem}{rewcomp}
  If $D \vdash G$ then $\apprx(D) \models G$.
\end{restatable}

It follows from the fact that there are only finitely many automaton clauses, that the canonical solved form is finite.
Furthermore, as the rewrite system is well founded and terminates, we can effectively decide whether a given automaton clause is in the canonical solved form.
Hence, we can compute $\apprx(D)$ by enumerating the possible automaton clauses and checking if they are the solved form of any clause in $D$.

\begin{theorem}[Decidability]
  Let \( V \) be an automaton formula, \( \overline{y} \) variables, \( G \) an goal formula and \( U \) an automaton formula.
  Then \( V,\, \overline{y} \vdash G \rew^* U \) is decidable.  
  Hence, $\apprx(D)$ is computable.
\end{theorem}

\changed[jj]{
\begin{remark}
  In principle, all fragments identified in Section~\ref{sec:fragments} can be decided by our decision procedure. 
  Strictly speaking, the procedure takes as input an existential-free \MSL($\omega$) formula (and thus formulas of any of the syntactic sub-fragments \MSL($n$)).
  However, our translations from Section~\ref{sec:reductions} allow for transforming any formula of \HOMSL($\omega$) or a \HOMSL($n$) sub-fragment thereof into an equisatisfiable existential-free formula of \MSL($\omega$).  
  Existentials have been eliminated for simplicity; we could have introduced new rules to the calculus to handle them natively.
  
  
\end{remark}
}

\subsection{Rewriting example}

By way of an example, we show how to use rewriting to obtain automaton clauses (i.e.~a subset of $\apprx(D)$) from the given formulas in Figure~\ref{ex:file} that allow for a replay of the proof in Figure~\ref{fig:file-proof}.
One can view this as a concrete example of the completeness proof in action.
The strategy of the completeness argument is to start from the leaves of the proof tree, where resolution steps already involve automaton clauses\footnote{A resolution step whose conclusion is a leaf must use a clause of shape $\truetm \implies A$ as premise, which is already automaton.}, and work back towards the root.
We start at the leaf labelled with (18).

Since clauses (18) and (7) are already automaton, we can use (7) to rewrite the body of clause (10) using \ruleName{Assm}:
\[
  \Cex(x\ h\ (\putContS\ k)) \rew \IsClosed(h) \wedge (\forall h\,k.\ \IsClosed(h) \implies \Cex\,(x\ h\ k))
\]
From this, we obtain (10'): $\IsClosed(h) \wedge (\forall h\,k.\ \IsClosed(h) \implies \Cex\,(x\ h\ k)) \implies \Cex(\putStrS\ x\ h\ k)$ is an automaton clause in $\apprx(D)$.
Continuing down the tree, we can use (10') to rewrite the body of (16) using \ruleName{Step}:
\[
  \Cex(\putStrS\ x\ h_3\ \idS) \rew \IsClosed(h_3) \wedge (\forall h\,k.\ \IsClosed(h) \implies \Cex\,(x\ h\ k)) \\
\]
thus obtaining (16'): $\IsClosed(h_3) \wedge (\forall h\,k.\ \IsClosed(h) \implies \Cex\,(x\ h\ k)) \implies \Cex(\actSB\ x\ h)$, another automaton clause.
Following the proof branch down, we can use (16') to rewrite the body of (14) using \ruleName{Assm}:
\[
  \Cex(k\ y\ h_2) \;\rew
  \begin{array}{l}
    \IsClosed(h_2) \wedge (\forall h\,k.\ \IsClosed(h) \implies \Cex\,(y\ h\ k)) \\ 
    \qquad\wedge (\forall x\,h.\ \IsClosed(h) \wedge (\forall h\,k.\ \IsClosed(h) \implies \Cex\,(x\ h\ k)) \implies \Cex(k\ x\ h))
  \end{array}
\]
and, naming that formula $U$ for brevity, we obtain (14'): $\forall y\,k\,h_2.\ U \implies \Cex(\withFileSC\ y\ k\ h_2)$.  At any point we can stop and the automaton clauses we have obtained are sufficient to entail the goal at the same place in the tree.  
For example, it is easy to check that any model of the automaton clauses we have collected so far satisfies $\Cex(\withFileSC\ \readS\ \actSB\ \closedhdl)$.

\section{Automaton clauses: connecting logic, types, and automata}
\label{sec:automaton_formulas}

Recall that our `automaton clauses' are named after their first-order counterparts in \citet{Goubault-Larrecq2002}.
We argue that this name is equally justified for our higher-order automaton clauses.

While it is easy to see that a first-order automaton clause $\forall x\,y.\, P\,x \land Q\,y \land R\,y \Rightarrow S\,(f\,x\,y)$ fits the shape of a transition in a finite tree automaton (if $t_1$ is accepted from state $P$ and $t_2$ from states $Q$ and $R$, then $f\,t_1\,t_2$ is accepted from state $S$), this relation is more complicated for higher-order automaton clauses with variable-headed atoms and nested clauses.

Nonetheless, our automaton clauses share a vital trait with automaton transitions:
the monadic, shallow, linear heads `peel' a single constructor off a tree and separate its children without interdependencies.
Furthermore, the fact that the number of automaton clauses is $n$-exponential in the order $n$ of the program (Theorem~\ref{thm:no_of_automaton_clauses}) suggests automaton clauses could be a \emph{defunctionalisation} of \MSL($\omega$) and, thus, `first-order' in a sense -- maybe even regular, given they are essentially monadic.

This intuition turns out to be true: 
automaton clauses correspond to intersection types, which are known to give rise to \emph{alternating tree automata}, see~\citet{RehofUrzyczyn2011} and~\citet{broadbent2013saturation}, that are equivalent to ordinary non-deterministic finite tree automata.

This correspondence between automaton clauses and intersection types allows us to reduce \MSL($\omega$) provability to intersection untypeability and vice versa, giving us a new algorithm for solving higher-order recursion scheme (HORS) model checking along the way.

\subsection{Correspondence with intersection types}
\label{sec:intersection_types}

\newcommand\calG{{\mathcal{G}}}
\newcommand\calN{{\mathcal{N}}}
\newcommand\calR{{\mathcal{R}}}
\newcommand\conss{{\mathsf{c}}}
\newcommand\succs{{\mathsf{s}}}
\newcommand\zeros{{\mathsf{z}}}

In the remainder of this section, we reserve the word \emph{term} for terms of constructor types and $\tau$ for strict types.
We follow the practice from HORS model checking in stratifying the definitions of intersection types for convenience.

Fix $(q \in)\, \basetypes$ as a set of atomic \emph{base types}.
We define the \emph{intersection types} and \emph{strict types} over $Q$:
\[
  \textstyle
  \textsc{(Strict Types)} \;\; \tau \; \Coloneqq \; q \mid \sigmaInt \to\tau 
  \qquad \textsc{(Intersection Types)} \;\; \sigmaInt \; \Coloneqq \; \bigwedge^n_{i=1} \tau_i
\]
Arrows associate to the right; intersections bind tightest.
An \emph{intersection type environment}, written $\Gamma$, is a finite, partial function from variables to intersection types.
We denote an arbitrary strict or intersection type by $\theta$ and refer to it wlog as an intersection type (since strict types are intersection types with $n=1$).
We assume types are well typed, where base types $q$ have type $\iota$.
    
An \emph{intersection (refinement) type system} is a triple $(\Sigma, Q, \type)$ in which $\Sigma$ is a signature of constants, $Q$ a set of base types equipped with a preorder $\preorder$, and $\type$ assigns an intersection type of type $\sigma$ to each $a:\sigma \in \Sigma$.
The \emph{subtype relation} over $Q$ is the least relation on types, written $\leq$, that validates standard subtyping rules (see the long version of this work for details).
  
Terms can be typed in the following way:
\[
  \arraycolsep=5pt
  \begin{array}{cc}

    \prftree[l,r]
      { $\type(c) = \textstyle \bigwedge_{i=1}^n \tau_i$ }
      { \ruleName{T-Con} }
      { \Gamma \vdash \hasType{c}{\tau_i} }
        
      \hspace*{20pt}
        
    \prftree[l]
      { \ruleName{T-Var} }
      { }
      { \Gamma, \hasType{x}{\textstyle \bigwedge_{i=1}^n \tau_i} \vdash \hasType{x}{\tau_i} }
        
      \\[12pt]

    \prftree[l,r]
      { $\textstyle \bigwedge_{i=1}^n \tau_i \leq \sigmaInt$ }
      { \ruleName{T-App} }
      { \Gamma \vdash \hasType{s}{\sigmaInt \to \tau} }
      { \Gamma \vdash \hasType{t}{\tau_i}\;(\forall i \in [1..n]) }
      { \Gamma \vdash \hasType{s\,t}{\tau} }

  \end{array}
\]
    
We may write $\Gamma \vdash \hasType{t}{\sigmaInt}$ as a shorthand for $\bigwedge_{\tau \in \sigmaInt} (\Gamma \vdash \hasType{t}{\tau})$.

We consider intersection type systems $(\Sigma,\basetypes,\type)$ with a trivial preorder, i.e.~$\preorder = \set{(q,q) \mid q \in \basetypes}$.
This is not a restriction, because non-trivial preorders can be simulated with fresh base types.
For example, $\mathsf{nat} \leq \mathsf{int}$ can be enforced with a fresh base type $\mathsf{pos}$ by replacing $\mathsf{nat}$ by $\mathsf{int} \land \mathsf{pos}$.

\subsubsection{Type-clause correspondence}
\label{sec:types_typing-rewriting}

Recall that automaton clauses are essentially monadic;
each automaton clause contains precisely one symbol of type $\gamma$ that is not locally bound, either a free variable or a constructor from signature $\Sigma$.

Given an automaton clause $T$ with top-level variable $x$, an instantiation $T[t/x]$ rewrites to $\truetm$ just if $t$ satisfies the constraints imposed by $T$.
Clearly, $T[t/x]$ rewrites to $\truetm$ iff $T[y/x][t/y]$ does;
after all, $x$ and $y$ are the only free variables in $T$ and $T[y/x]$, resp.
This means automaton clauses $T$ and $T[y/x]$ are indistinguishable wrt the constraints they impose on their respective variables.

In this section, we discuss the constraints imposed by automaton clauses and formulas.
It will be helpful to think of an automaton clause as \emph{forgetful} with regards to its top-level symbol.
As we shall see below, forgetful automaton clauses (hereafter just `automaton clauses') correspond to strict intersection types over base types $\Pi$;
an instantiated automaton clause $T[t/x]$ rewrites to $\truetm$ precisely when closed term $t$ has the strict type corresponding to $T$ and vice versa.

Automaton formulas and intersection types are simply two (equivalent) ways of imposing constraints on $t$.

%
%

\begin{definition}[Correspondence between automaton clauses and intersection types]
  \label{def:HOMSL_intersection_type_correspondence}
  The following typing rules define a one-to-one correspondence between intersection types and \MSL($\omega$) automaton formulas.
  We assign an intersection type $\theta$ of type $\sigma_1 \to \dots \to \sigma_m \to \iota$ to automaton formula $U = \res{U}{x_1} \land \dots \land \res{U}{x_m}$ with $x_1:\sigma_1, \dots, x_m:\sigma_m$, written $\hasType{U}{\theta}$:
  \[
    \arraycolsep=5pt
    \begin{array}{cc}
      \prftree
      {}
      {\hasType{\truetm}{\top}}

      \hspace*{13pt}
      
      \prftree
      {}
      {\hasType{P\,\xi}{q_P}}
      
      \hspace*{13pt}
      
      \prftree
      {\hasType{\res{U}{x_1}}{\theta_1}}
      {\dots}
      {\hasType{\res{U}{x_m}}{\theta_m}}
      {\hasType{(\forall \vv{x}.\, U \Rightarrow P\,(\xi\,\vv{x}))}{\theta_1 \to \dots \to \theta_m \to q_P} }
      
      \hspace*{13pt}
      
      \prftree
      {\hasType{\res{{U_1}}{x}}{\theta_1}}
      {\hasType{\res{{U_2}}{x}}{\theta_2}}
      {\hasType{\res{{U_1}}{x} \land \res{{U_2}}{x}}{\theta_1 \land \theta_2} }
      
      
    \end{array}
  \]
  where $q_P \in Q$ is a basetype and $\xi$ a variable or constructor from $\Sigma$.
  
  We write $\fromClause{U} = \theta$ for $\hasType{U}{\theta}$, and $\toClause{\theta}{\vv{f}} = U$ for $\hasType{U}{\theta}$ with $\vv{f}$ as top-level constructors.
  
    We write $U \leqa U'$ for $U,U'$ over $x_1 \dots x_m$ just if $\forall i \in [1..m].\, \fromClause{\res{U}{x_i}} \leq \fromClause{\res{U'}{x_i}}$.
\end{definition}

\begin{example}
  Let $U = \forall v\,x\,y.\,Q\,y \wedge (\forall z.\, Q\,z \land P\,z \Rightarrow R\,(x\,z)) \Rightarrow P\,(c\,v\,x\,y)$.
  The following hold:
  \begin{align*}
    \fromClause{U}
      &= \top \to (q_Q \land q_P \to q_R) \to q_Q \to q_P \\
    U
      &= \toClause{\top \to (q_Q \land q_P \to q_R) \to q_Q \to q_P}{c}
  \end{align*}
  This is witnessed by:
  \[
    \prftree
      { \prfaxiom{\hasType{\truetm}{\top}} }
      { \prftree
        { \prfaxiom{\hasType{Q\,z}{q_Q}} }
        { \prfaxiom{\hasType{P\,z}{q_P}} }
        { \hasType{\forall z.\, Q\,z \land P\,z \Rightarrow R\,(x\,z)}{q_Q \land q_P \to q_R} } 
      }
      { \prfaxiom{\hasType{Q\,y}{q_Q}} }
      { \hasType{\forall v\,x\,y.\,Q\,y \wedge (\forall z.\, Q\,z \land P\,z \Rightarrow R\,(x\,z)) \Rightarrow P\,(c\,v\,x\,y)}{\top \to (q_Q \land q_P \to q_R) \to q_Q \to q_P} }
  \]
\end{example}

\begin{lemma}[Order isomorphism]
  \label{lem:order_iso}
  Let $\theta,\theta_1,\theta_2$ be intersection types of type $\sigma_1 \to \dots \to \sigma_m \to \iota$ and $U,U_1,U_2$ automaton clauses over $\xx = x_1\dots x_m$ with $x_i:\sigma_i$, for all $i \in [1..m]$.
  The following hold:
  \begin{enumerate}[label=(\roman*)]
    \item $\toClause{\fromClause{U}}{\xx} = U$ \label{enum:auto_iden}
    \item $\fromClause{\toClause{\theta}{\xx}} = \theta$ \label{enum:type_iden}
    \item $\theta_1 \leq \theta_2$ if, and only if, $\toClause{\theta_1}{\xx} \leqa \toClause{\theta_2}{\xx}$ \label{enum:type_to_clause_order} 
    \item $\fromClause{U_1} \leq \fromClause{U_2}$ if, and only if, $U_1 \leqa U_2$ \label{enum:clause_to_type_order}
  \end{enumerate}
\end{lemma}

This isomorphism between automaton formulas and intersection types extends to intersection type systems $(\Sigma,\basetypes,\type)$ and \MSL($\omega$) automaton formulas $V$ over signatures $(\Pi,\Sigma')$:
\MSL($\omega$) predicates $\Pi$ correspond to basetypes $\basetypes$;
automaton environment $V$ corresponds to intersection type environment $\Gamma \uplus \type$, both with domain $\dom{\Gamma} \cup \Sigma = \Sigma'$.
 
\begin{definition}[Correspondence between automaton formulas and intersection type systems]
  \label{def:type_system_automaton_formula}
  We map intersection type system $(\Sigma,\basetypes,\type)$ with type environment $\Gamma$ to \MSL($\omega$) automaton formula $V_{\Sigma,\Gamma}$ over signatures $(\Pi_Q,\Sigma \cup \dom{\Gamma})$ using:
  \begin{align*}
    \Pi_\basetypes
      \defeq \set{P \mid q_P \in \basetypes}
    \qquad
    V_{\Sigma,\Gamma} 
      \defeq \set{\toClause{\type(f)}{f} \mid f \in \Sigma} \cup \set{ \toClause{\theta}{a} \mid \hasType{a}{\theta} \in \Gamma }
  \end{align*}
  
  In the converse direction, we map an \MSL($\omega$) automaton formula $V$ over $(\Pi,\Sigma)$ to an intersection type system $(\emptyset,\basetypes_\Pi,\emptyset)$ with type environment $\Gamma_V$:
  \begin{align*}
    \basetypes_\Pi 
      \defeq \set{q_P \mid P \in \Pi} 
      \qquad
    \Gamma_V
      \defeq \left\{ \hasType{a}{\theta} \;\middle|\; a \in \dom{V} \land \fromClause{\res{V}{a}} = \theta \right\}
  \end{align*}
\end{definition}

The above is not strictly a one-to-one correspondence, because the separate environments $\type$ and $\Gamma$ for intersection-type `constants' and `variables', resp., are mapped to a single automaton formula $V$.
This distinction cannot be recovered in the other direction.
Note, however, that \ruleName{T-Con} and \ruleName{T-Var} both choose a strict type from $\Gamma \uplus \type$ for a symbol.
For the typeability of a term it does not matter whether a symbol lives in the domain of $\Gamma$ or $\type$.

Since the same typing judgements hold for $(\Sigma,\basetypes,\type)$ with type environment $\Gamma$ as for $(\emptyset,\basetypes,\emptyset)$ with type environment $\Gamma \uplus \type$, we identify $\Gamma_{V_{\Sigma,\Gamma}}$ and $\Gamma$ with $\type$ left implicit.
Clearly, $V_{\Gamma_V} = V$, which gives us a one-to-one correspondence after all.
Furthermore, because $\Gamma_{V_{\Sigma,\Gamma}}$ and $\Gamma$ type the same terms, we may assume WLOG our intersection type systems do not contain constants, thus, $V_{\Sigma,\Gamma} = V_\Gamma$, giving rise to the following lemma.

\begin{lemma}
  $V_{\Gamma_V} = V$ and $\Gamma_{V_\Gamma} = \Gamma$.
\end{lemma}

\subsubsection{Typing-rewriting correspondence}

Unless otherwise specified, we consider only goal formulas $G$ that result from rewriting some non-existential \MSL($\omega$) goal formula $G'$, i.e.~$V, \yy \vdash G' \rew^* G$.
Such a $G$ is itself non-existential, and all its implications have automaton bodies, as
rewriting does not introduce existentials and only introduces implications with automaton bodies.

Given an isomorphic type environment and automaton formula as per Definition~\ref{def:type_system_automaton_formula}, we can now formalise the equivalence between typing judgements and rewrites.
 
\begin{restatable}[Typing-rewriting correspondence]{proposition}{typeClause}
  \label{cor:type-clause}
  ~
  \begin{enumerate}[label=(\roman*)]
  
    \item For all $U$ over $\yy$, there exists $U'$ such that $U \leqa U'$ and $V, \yy \vdash \toClause{\sigmaInt}{t} \rew^* U'$ if, and only if, $\Gamma_V, \hasType{\yy}{\fromClause{U}} \vdash \hasType{t}{\sigmaInt}$. \label{enum:typeClause_without_goal_typing}
    
    \item $V \vdash \toClause{\sigmaInt}{t} \rew^* \truetm$ if, and only if, $\Gamma_V \vdash \hasType{t}{\sigmaInt}$ \label{enum:typeClause_true}
  
  \end{enumerate}
\end{restatable}

Thanks to $\Gamma_{V_\Gamma} = \Gamma$, the above not only provides an equivalence between $V$ and $\Gamma_V$ but also between $V_{\Gamma}$ and $\Gamma$.
%

\subsection{Reducing HORS intersection typing to \texorpdfstring{\MSL($\omega$)}{\MSL(omega)}}
\label{sec:types_reduction}

One of the most well-studied problems in higher-order model checking is the safety problem for higher-order recursion schemes (HORS): does the tree $\sem{\calG}$ generated by HORS $\calG$ satisfy safety property $\varphi$?
The property is typically expressed as an alternating trivial automaton -- or its negation as an alternating cotrivial automaton.
This problem reduces to intersection (un)typeability~\cite{Kobayashi2009}, which we shall use.

A \emph{higher-order recursion scheme} (HORS) is a set $\calR$ of well-typed ground definitions of type $\iota$ (i.e.~the RHSs are applicative terms over the formal parameters, constants $\consig$, and variables $\calN$):
\[
  \set{ 
    f_1\, \vv{y_1} = t_1,
      \quad\dots,\quad
    f_n\, \vv{y_n} = t_n
  }
\]
where $f_1,\dots,f_n \in \calN$ and $\vv{y_1},\dots,\vv{y_n} \in \boundvars$ are vectors of distinct variables.
There is a designated nullary start symbol $S:\iota$, so a HORS can be denoted by a quadruple $\calG = \langle \calN, \consig, \calR, S \rangle$.

See e.g.~\citet{Ong2006} for a full account of HORS.

For the HORS safety problem, we assume a HORS comes equipped with base types $Q_\iota$ and a \emph{negative} typing of constants $\consig$, denoted by $\type$, with respect $Q_\iota$.
This gives rise to an intersection type system $(\consig, Q_\iota, \type)$.
Typically, we will ask whether $\vdash \hasType{S}{q_0}$ is this system, for some $q_0 \in Q_\iota$.

\paragraph{Negative types.}

Negative constant types are easily computable via a DeMorgan dual;
if we read a (positive) type $\hasType{f}{\bigwedge_{i \in [1..n]} (\sigmaInt_{i,1} \to \dots \to \sigmaInt_{i,m} \to q_P})$ as $\hasType{f\,t_1\dots t_m}{q_P}$ iff $\bigvee_{i \in [1..n]} \bigwedge_{j \in [1..m]} \hasType{t_j}{\sigmaInt_{i,j}}$, then the corresponding negative type is the DeMorgan dual of this boolean formula (see e.g.~\citet{Muller1987}, who use this construction to negate alternating tree automata).

Since we use only negative types but forget about this for the remainder of the section, we shall simply write $\hasType{t}{\tau}$ to mean $t$ has negative type $\tau$.

\begin{example}
  Consider integer-division operator $\divi$ with an error type:
  \begin{align*}
    &(\top \to \zero \to \err) \land (\err \to \top \to \err) \land (\top \to \err \to \err) 
  \end{align*}
  This means that $\hasType{n \mathbin{\divi} m}{\err}$ iff $(\hasType{n}{\top} \land \hasType{m}{\zero}) \lor (\hasType{n}{\err} \land \hasType{m}{\top}) \lor (\hasType{n}{\top} \land \hasType{m}{\err})$.
  The corresponding negative type can be computed via DeMorgan as $(\err \to \zero \land \err \to \err)$, which should be read as: 
  ``$n \mathbin{\divi} m$ does not have type $\err$ if neither $n$ or $m$ has type $\err$ and $m$ does not have type $\zero$.''
  
\end{example}

\begin{example}
  Consider order-2 HORS $\calG_2 = \langle \set{S,F,B},\set{\conss,\succs,\zeros},\calR,S \rangle$ that generates a binary $\conss$-labelled spine whose left subtrees are unary $\succs^{2^0}\,\zeros$, $\succs^{2^1}\,\zeros$, $\succs^{2^2}\,\zeros$, etc. in that order, with $\calR$ given by:
  \[
    S = F\,\succs
    \qquad F\,\varphi = \conss\,(\varphi\, \zeros)\,(F\,(B\,\varphi\,\varphi))
    \qquad B\,\varphi\,\psi\,x = \varphi \, (\psi \, x)
  \]
  
  Let $Q = \set{q_S,q_0,q_1,q_2,q_3}$ be our base types that count the number of $\succs$ in a sequence modulo 4.
  Because we are considering a negative typing, the terminal symbol (the constant) $\zeros$ has type $q_1 \land q_2 \land q_3 = \type(\zeros)$ instead of $q_0$ and:
  \[
    \type(\conss) = (q_1 \land q_3 \to q_1 \land q_3 \to q_S) \land \textstyle\bigwedge_{i \in [0..3]} (q_i \to \top \to q_i )
    \quad\, \type(\succs) = \textstyle\bigwedge_{i \in [0..3]} (q_i \to q_{i+1\!\!\!\!\!\mod 4} )
  \]
  Now we may ask whether $\vdash \hasType{S}{q_S}$.
  Note that $\hasType{\conss}{q_1 \land q_3 \to q_1 \land q_3 \to q_S}$ means that the left subtree of $\conss$ has an even number of $\succs$ (and so does the left subtree of its right child), which holds.
\end{example}

\paragraph{The HORS untypeability problem.}

We call an intersection type environment $\Gamma$ \emph{$\calG$-coconsistent} just if, 
\begin{enumerate*}[label=(\arabic*)]
  \item $\Gamma$ is empty, or
  \item there exist $\hasType{f}{\tau} \in \Gamma$ and $(f\,\yy = t) \in \calG$ such that $\Gamma \backslash \set{ \hasType{f}{\tau}}$ is $\calG$-coconsistent and $\Gamma \backslash \set{ \hasType{f}{\tau}}, \hasType{\yy}{\vv{\sigmaInt}} \vdash \hasType{t}{q_P}$, where $\tau = \sigmaInt_1 \to \dots \to \sigmaInt_m \to q_P$. 
\end{enumerate*}
Intuively, every intersection type in a $\calG$-coconsistent $\Gamma$ is (finitely) required by some program definition $(f\,\yy = t) \in \calG$.
Thus, a $\calG$-coconsistent type environment corresponds to a finite trace of an intersection typing for $\calG$.

We prove that an instance of the HORS untypeability problem
(does there exist a $\calG$-coconsistent type environment $\Gamma$ such that $\Gamma \vdash \hasType{S}{q_0}$?) reduces to \MSL($\omega$) provability by adapting the rewriting algorithm to construct an intersection type environment instead of a canonical solved form.

\begin{restatable}{theorem}{untypeabilityToProvability}
  \label{thm:untypeability_reduces_to_provability}
  HORS intersection untypeability reduces to \MSL($\omega$) provability.
\end{restatable}

We convert ground definitions from $\calG$ to definite clauses $D_{\calG}$ by wrapping both sides in a predicate $P$ for each base type $q_P \in Q_\iota$.
This gives us $(f\,\yy = t) \in \calG$ if, and only if, $(\forall \yy.\, P\,t \Rightarrow P\,(f\,\yy)) \in D_\calG$ for all $q_P \in Q_\iota$.
Furthermore, the types of $\consig$ are added to $D_{\calG}$ as automaton clauses, giving rise to:
\[
  D_\calG 
    \defeq \set{\forall \yy.\, P\,t \Rightarrow P\,(f\,\yy)  \mid (f\,\yy = t) \in \calG \land t:\iota \land q_P \in Q_\iota}
      \cup \set{ \toClause{\tau}{c} \mid \hasType{c}{\tau} \in \consig  } 
        \footnote{
          We allow clauses $\toClause{\tau}{c}$ in $D_\calG$ even though they are not generally definite clauses.
          Because they are automaton, we could instead directly include them in the first iteration of the rewriting algorithm.
        }
\]
%
%
For our example $\calG_2$, $D_{\calG_2}$ looks like:
\begin{align*}
  D_\calG \defeq &\set{ ( P\, (F\,\succs) \Rightarrow P\,S) \mid q_P \in Q }
    \cup \set{ (\forall \varphi.\, P\,(\conss\,(\varphi\, \zeros)\,(F\,(B\,\varphi\,\varphi))) \Rightarrow P\,(F\,\varphi) ) \mid q_P \in Q } \\
    &\cup \set{ (\forall \varphi \, \psi\,x.\, P\,(\varphi \, (\psi \, x)) \Rightarrow P\,(B\,\varphi\,\psi\,x)) \mid q_P \in Q }
\end{align*}

\begin{definition}[Typing Algorithm]
  \label{def:typing_algorithm}
  Given definite formula $D$, we construct a type environment $\typeAppr{D} = \Gamma_{\apprx(D)}$ using the (inductive) \MSL($\omega$) rewrite algorithm:
  \[
      \begin{array}{c}
        \prftree[l]
          { $\left.\begin{array}{r}
            (\forall \yy.\, G \Rightarrow P\,(f\,\yy)) \in D \\ 
            V_\Gamma, \yy \vdash G \rew^* U
            \end{array}\right|$ }
          { \Gamma \subsetpluseq \typeAppr{D} }
          { \hasType{f}{\fromClause{\forall \yy.\, U \Rightarrow P\,(f\,\yy)}} \in \typeAppr{D} }
      \end{array}
    \]
\end{definition}

\begin{proof}[Proof sketch]
  For Theorem~\ref{thm:untypeability_reduces_to_provability}, it suffices to show the following:
  \[
    \exists \calG \text{-coconsistent } \Gamma.\, \Gamma \vdash \hasType{t}{\tau}
    \text{ if, and only if, }
    \apprx(D_\calG) \vdash \toClause{\tau}{t} \rew^* \truetm
  \]
  We rely on the correspondence between typing and rewriting (Proposition~\ref{cor:type-clause}), and the fact that $\typeAppr{D_\calG} = \Gamma_{\apprx(D_\calG)}$ is the largest $\calG$-coconsistent environment.
  The claim follows from the restricted case where the LHS is $\vdash \hasType{S}{q_0}$. 
\end{proof}

%
%

The reduction from HORS untypeability to \MSL($\omega$) provability is clearly polynomial.
This gives us a lower bound on the complexity of \MSL($\omega$) in the order $n$ of the program and highest arity $k$ of any function symbol, based on the known complexity of HORS untypeability.

\begin{theorem}
  Deciding \MSL($\omega$) provability is at least $\Exp{n-1}(k\size{\Pi})$-hard.
\end{theorem}

\subsection{Reducing \texorpdfstring{\MSL($\omega$)}{\MSL(omega)} to HORS intersection typing}
\label{sec:MSL_reduction}

\newcommand\conSym{C}
\newcommand\conFun[2]{\conSym \, {#1} \, {#2}}
\newcommand\disSym{D}
\newcommand\disFun[2]{\conSym \, {#1} \, {#2}}
\DeclarePairedDelimiter{\toTerm}{\sslash}{\sslash}
\DeclarePairedDelimiter{\asMSL}{\sslash}{\sslash}
\newcommand\tPred{T}

\newcommand\bbs{{o_{\mathsf{B}}}}
\newcommand\indBoolType{{q_{\tPred}}}

For the converse reduction, we reduce \MSL($\omega$) directly to HORS cotrivial automaton model checking.
We rely on an extension of HORS by \citet{Neatherway2012} called \emph{HORS with cases} that enables us to use nondeterminism and a case-switch on the base types $Q_\iota$ (i.e.~predicates $\Pi$), due to \MSL($\omega$) lacking a clean separation between a state-agnostic rewrite system and a property automaton.

The \MSL($\omega$) constants $\truetm$ and $\land$ are encoded as a HORS constant and variable, resp., making clause bodies monadic.
Now a clause $(\forall \yy.\, P'\,t \Rightarrow P\,(f\,\yy))$ can be mapped to $f\,\yy\,p = t \, p'$, where $p$ is a constant corresponding to $P \in \Pi$.

\paragraph{\MSL($\omega$)-to-HORS transformation.}

We transform \MSL($\omega$) constructor types $\gamma$ to HORS types $\gamma^+$ by setting $\iota^+ \defeq \iota \to \iota$ and $(\gamma_1 \to \gamma_2)^+ \defeq \gamma_1^+ \to \gamma_2^+$.  Then, the goal transformation to HORS bodies encodes $\truetm$ and $\land$ as:
\[
  \truetm^+ \defeq \truetm
  \qquad (G \land H)^+ \defeq G^+ \land H^+
  \qquad (P\,s)^+ \defeq s \, P 
\]
where, by some abuse, $\truetm$ and $\land$ on the RHS are a HORS constant and variable, resp.

\paragraph{\MSL($\omega$)-as-HORS}

Given an existential-free \MSL($\omega$) definite formula $D_0$ and goal formula $G_0$ over $\Pi$ and $\Sigma$, we construct the HORS $\calG = \langle \consig, \calN, \calR, S \rangle$ defined by:
\begin{align*}
  \consig &\defeq \set{ \truetm : \iota } \cup \set{ P : \iota \to o \in \Pi } \\
  \calN &\defeq \set{ f: \gamma^+ \mid f:\gamma \in \Sigma } \cup \set{S : \iota} \cup \set{ \land : \iota \to \iota \to \iota } \\
  \calR &\defeq \set{ f \, \xx \, P = G^+ \mid (\forall \xx.\, G \Rightarrow P \, (f \, \xx)) \in D_0 } \cup \set { S = G_0^+ } \cup \set{ \land \, \truetm \, \truetm = \truetm }
\end{align*}

\paragraph{The automaton.}

Because \MSL($\omega$) does not have a clean separation between automaton and state-agnostic definitions, our automaton is trivial;
it consists of a single state that accepts only the non-terminating/non-finished tree $\bot$.
Intuitively, $G_0^+$ rewrites to $\truetm$ precisely if $D_0 \vDash G_0$.

\begin{proposition}
  $D_0 \vDash G_0$ if, and only if, $\sem{\calG} \in \lang{A}$
\end{proposition}

This provides the missing link for the following theorem.

\begin{theorem}
  \label{cor:homc}
  \MSL($\omega$) provability and HORS (cotrivial) model checking are interreducible.
\end{theorem}

\section{Implementation \& Application}\label{sec:application}

We have implemented a decision procedure for \( \MSL(\omega) \) satisfiability in Haskell.
Recall that the full \( \HOMSL(\omega) \) language reduces to this fragment, see Section~\ref{sec:reductions}.
Our implementation incrementally rewrites clause bodies towards automaton form according to the rewrite relation from Section~\ref{sec:rewriting}, using automaton clauses that have already been discovered.
When a clause body is fully rewritten so a clause is automaton, further rewrites may become possible in other clause bodies, which are then reconsidered.
It is therefore important to retain partially rewritten clauses.
This procedure continues until no more automaton clauses can be produced and the set of clauses has been saturated.



\changed[jj]{To assess the viability of our \MSL($\omega$) decision procedure for higher-order verification, we study the case of socket programming in Haskell, where higher-order constraints arise naturally from the use of continuations in effectful code.

We implemented a Haskell library which provides an abstract typeclass of socket effects, one instance of which generates constraints whose satisfiability implies correct usage of the sockets. 
This approach alleviates the need for a heavyweight analysis front-end by exploiting a common pattern of coding with effects.
Typically, a program analysis front-end would take the source code of the program as input, internalise it as an AST and then walk over the AST to generate constraints; then a separate back-end would solve the constraints.
Our approach instead allows us to use a typeclass instance to generate constraints directly, without any need to process the program AST.
Doing this has practical advantages, because a standalone front-end usually requires regular updating to stay in sync with the syntax of a constantly evolving programming language.}

\paragraph{Socket API}

Socket APIs require the user to adhere to a strict protocol where only certain operations are permitted in each state.
Correctly tracking the state of sockets throughout a program can be difficult, much like with lazy IO, and is often the source of bugs.
Our implementation allows us to track not a single resource (e.g.\ file handler or socket) but countably many!

A socket can be in one of the following states: \emph{Ready}, \emph{Bound}, \emph{Listen}, \emph{Open}, \emph{Closed}.
%
%
The primitives modifying the state of a socket are summarised by the automaton 
in Figure~\ref{fig:socket-automaton}.
As these primitives operate in the \lstinline{IO}-monad, we encode them in an explicit continuation-passing style such that each primitive takes a socket and a continuation as arguments.
The socket and continuation are individuals (i.e.~type \( \iota \)), except in the case of \( \mathtt{Accept} \) that also creates a new socket and thus has a continuation of type \( \iota \rightarrow \iota \).
Each state is encoded as a predicate, with an additional \( \mathtt{Untracked} \) predicate whose meaning we explain below.

To account for the use of countably many resources, we employ a known trick that tracks the state of just one resource and non-deterministically chooses whether to track a newly created socket (unless one is already tracked)~\cite{kobayashi2013model, cookcook2007proving}.
In our case, the fresh socket is \changed[jj]{either
\begin{enumerate*}[label=(\arabic*)]
  \item labelled \( \underline{s} \) and subsequent operations acting on it contribute to the overall state or
  \item labelled \( \underline{u} \) and is untracked.
\end{enumerate*}
This approach suffices, because for each socket there exists a branch in which its behaviour is tracked and incorrect usage violates the overall state.}

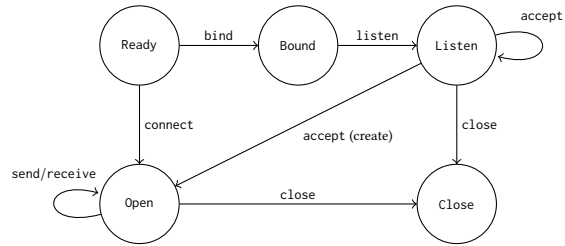
\begin{wrapfigure}{R}{0.55\textwidth}
  \begin{tikzpicture}[
      ->, auto, node distance=100pt,
      state/.style={circle, draw, minimum size=50pt}
    ]
    \begin{scope}[scale=0.6, transform shape]
      \node[state] (A)                    {$\mathtt{Ready}$};
      \node[state]         (B) [right of=A]       {$\mathtt{Bound}$};
      \node[state]         (D) [right of=B]       {$\mathtt{Listen}$};
      \node[state]         (C) [below of=A]       {$\mathtt{Open}$};
      \node[state]         (E) [below of=D]       {$\mathtt{Close}$};

      \path
      (A) edge              node {$\mathtt{bind}$} (B)
      (A) edge              node {$\mathtt{connect}$} (C)
      (B) edge              node {$\mathtt{listen}$} (D)
      (D) edge [loop right] node[above=0.4cm] {$\mathtt{accept}$} (D)
      (C) edge [loop left] node[above=0.4cm] {$\mathtt{send}/\mathtt{receive}$} (C)
      (D) edge node {$\mathtt{accept}$ (create)} (C)
      (C) edge              node {$\mathtt{close}$} (E)
      (D) edge              node {$\mathtt{close}$} (E);
    \end{scope}
  \end{tikzpicture}
  \caption{Socket states and operations manipulating them}\label{fig:socket-automaton}
\end{wrapfigure}

\paragraph{The \MSL($\omega$) clauses.}

Given a socket-manipulating Haskell program, the implementation computes a two-part \MSL($\omega$) formula that models its behaviour:
\begin{enumerate*}
  \item a formula that captures the socket protocol and \label{app:protocol}
  \item a clausal representation of the semantics of the program. \label{app:program}
\end{enumerate*}
Intuitively, a predicate is satisfied by a program when that program's usage of the tracked socket violates the protocol for the corresponding state.
The \( \mathtt{Untracked} \) predicate is satisfied by a program that violates the protocol \emph{for any} socket.
When predicates are supplied with the tracked socket, the clauses encode the complement of the automaton from Figure~\ref{fig:socket-automaton}; otherwise the state is unchanged.
Furthermore, when sockets are created in the \( \mathtt{Untracked} \) state, as described above, there are two clauses to account for whether the new socket is to be tracked or not.
If a socket is created in any other state, it is simply untracked to prevent junk branches where multiple sockets are tracked with overlapping states.

\paragraph{Extracting constraints.}

\changed[jj]{Part~(\ref{app:program}) the two-part \MSL($\omega$) formula is determined by the program.
However, our approach does not need to process the Haskell source code and obtain an AST.
Instead, we rely on a Haskell typeclass instance.}

The socket primitives are provided as methods of a typeclass refining the monad class, which is further parametrised by the type of sockets.
The instance of this typeclass for \lstinline{IO} behaves in the usual manner, but we also supply an instance for analysis whose sockets are variable names and which merely accumulates the effects as raw syntax, ignoring any parameters other than the socket and continuation.

The advantage of this approach is that processing the source code of the program is not required, instead relying on normal, program evaluation to construct the constraints.
One complication, however, is that we require any unbounded recursion to be made explicit to prevent an infinite evaluation of the program's definition.
We provide a method capturing this as part of the typeclass interface.
Recursion points are given a fresh name, to simulate as top-level definition of the program, and their bodies are analysed and attributed to those function names.
Once a collection of recursive top-level definitions has been identified with an additional entry point we generate clauses unfolding each definition, without changing state, as in Theorem~\ref{thm:untypeability_reduces_to_provability}.

Furthermore, the program cannot arbitrarily depend on runtime data such as the value received by a socket, the effects of the program must be statically known.
Branching code is not completely precluded however.
Inspired by the \emph{selective} extension of applicative functors that supports finite branching on runtime data, we add a branch combinator \( \mathtt{branch} : \mathtt{Bool} \rightarrow f a \rightarrow f a \rightarrow f a \), encoded as multiple clauses that disregard the condition~\cite{mokhov2019selective}.

\paragraph{Examples.}

We tested our tool on example socket-manipulating programs taken from StackOverflow (with values modified).
When presenting these examples, we will use Haskell's \lstinline{do}-notation and the ``bind'' operators (\( \gg \), \( \gg\!= \)) for monadic actions as exposed in the user-interface, which are to be understood as syntactic-sugar for the underlying continuation passing-style.
The first, with the original program on left, violates the protocol on line 8 where it attempts to send a message over \( \texttt{soc} \) which is in the \( \mathtt{Listening} \) state after line 5\footnote{\url{https://stackoverflow.com/questions/62052147/haskell-sendall-message-to-socket-client-results-in-exception-network-socket}}.
The tool was able to correctly detect the bug in 4.9ms, and accepted the correction (on the right) after 3.4ms.

\begin{minipage}[t]{0.45\textwidth}
  \begin{lstlisting}[label={ex:socket1-wrong}]
main = do
  soc <- socket
  bind soc 1234
  listen soc
  x <- accept soc
  forever $ do
    receive x
    send soc "Hi!"
\end{lstlisting}
\end{minipage}
\begin{minipage}[t]{0.49\textwidth}
  \begin{lstlisting}[label={ex:socket1-correct}]
main = do
  soc <- socket
  bind soc 1234
  listen soc
  x <- accept soc
  forever $ do
    receive x
    send x "Hi!"
\end{lstlisting}
\end{minipage}

The following toy example makes use of our branching construct.
When run in the \lstinline{IO} monad, this will behave just like an \lstinline{if-then-else} clause, for analysis, however, both branches are explored.
The snippet initialises a socket and repeatedly receives a message until it is ``closed'' when it closes the socket.
In the version on the left, the loop continues regardless thus attempting to receive from a closed socket.
This implementation violates the protocol and is detected by our tool in 4.2ms.
The fix, on the right, exits the loop once the socket is closed 5.1ms.

\begin{minipage}[t]{0.49\textwidth}
  \begin{lstlisting}
main = do
  soc <- socket
  bind soc 1234
  listen soc
  x <- accept soc
  forever $ do
    msg <- receive x
    branch (msg == "close")
      (close x)
      (pure ())
\end{lstlisting}
\end{minipage}
\begin{minipage}[t]{0.49\textwidth}
  \begin{lstlisting}
main = do
  soc <- socket
  bind soc 1234
  listen soc
  x <- accept soc
  fix $ \k -> do
    msg <- receive x
    branch (msg == "close")
      (close x >> k)
      (k ())
\end{lstlisting}
\end{minipage}


\section{Conclusion and Related Work}
\label{sec:related}

We have proposed new classes of constraints that are designed to capture the complex, higher-order behaviours of programs with first-class procedures. 
We developed their theory to (a) show decidability of the classes and (b) situate them with respect to higher-order program verification. 
We also described an implementation and its application to the verification of socket programming.

\changed[jj]{
\paragraph{Complexity.}
Our reduction of intersection typeability to \MSL($\omega$) satisfiability gives us $(n-1)$-EXPTIME hardness of \MSL($\omega$) satisfiability.
Furthermore, the reduction of order-$n$ \MSL($\omega$) to order-$(n+1)$ HORS with cases provides an $(n+1)$-EXPTIME upper bound, thanks to a result by \citet{Clairambault2018}. 
We derive this same naive upper bound directly from the decision procedure, where every application of \ruleName{Step} or \ruleName{Assm} on an order-$n$ symbol has $n$-exponentially many candidate side conditions.
Further study is required to obtain a tighter upper bound.
}

\subsubsection*{Related work}
We survey some of the work that is most closely related to our own.

\paragraph{Automata, types, and clauses.}
MSL was proposed independently by \citet{Weidenbach1999} and \citet{nielson2002normalizable} (as \HOne{}), with \citet{goubault-larrecq2005} providing the bridge between the two.
Since then, it has been extended beyond Horn and with the addition of straight dismatching constraints in \citet{teucke2017sdm}.
Recall that the solved form of clauses for the first-order \HOne{} fragment were named \emph{automaton} clauses because of their shape, a connection that has also been made in~\citet{Weidenbach1999,Nagaya2002rewriting}.
This name is equally justified for our (higher-order) automaton clauses, since they, too, define finite tree automata \cite[via intersection types,][]{broadbent2013saturation}.
The relationship between higher-order automata, types, and particular sets of clauses goes back to \citet{Fruhwirth1997}.

\paragraph{Set constraints}
Set constraints are a powerful language that has been very influential in program analysis \cite{aiken1999sets}.
They are known to be equivalent to the monadic class \cite{bachmair-etal1993monadic} and, therefore, have a very close connection with \( \MSL \).
Higher-order set constraints have also been considered, defining sets of terms rather than higher-order predicates much like \( \MSL(\omega) \)~\cite{goubault2002higher}.
Although the relationship between our constraints and those of \emph{loc cit} is not well understood, we point out that their constraints are solvable in $2$-NEXPTIME, whereas satisfiability in our class is $(n-1)$-EXPTIME hard.

\paragraph{HORS model checking}
There is a strong connection between traditional higher-order model checking with higher-order recursion schemes \cite[e.g.][]{kobayashi2013model} and \( \MSL(\omega) \) problems, as witnessed by their interreducibility.
Many approaches to inferring and verifying types for higher-order recursion schemes have been considered, but the most closely related to our work is the saturation-based approach considered by \citet{broadbent2013saturation}.
The main novelty of their algorithm is that typing constraints are propagated backwards starting from the final (unaccepted) states, rather than the forward from the target state.
While backward propagation is analogous to goal-orientation search, attempting to derive clauses in order to rewrite the goal, their saturation-based approach is similar to our accumulation of automaton clauses in a bottom-up manner.
Furthermore, follow up work improved upon the efficiency of the saturation-based approach by representing intersection types as a type of binary decision diagrams that compactly describes a family of sets~\cite{terao2014zdd}.
More work needs to be done to draw a detailed comparison between these algorithms and our own.
Furthermore, as many of these algorithms are in their second or third generation, there will be possible optimisations that can be transferred to our own setting, in addition to novel optimisations that take advantage of our setting.



\paragraph{HFL model checking}
Higher-order fixed-point logic, HFL, is a very expressive logic also designed as an appropriate language for program verification~\cite{kobayashi2021overview}.
It is more expressive than higher-order constrained Horn clauses in general, and our fragment in particular, by supporting both the greatest and least fixed-point.
This duality allows it to express liveness properties as well as safety properties.
Furthermore, this logic allows for a background constraint theory.
In the pure case, HFL is known to be decidable  by reduction to intersection typing problem~\cite{hosoi2019type}.

\paragraph{Refinement type checking and constrained Horn clauses}
It was observed by \citet{rybalchenko-etal2012} that a standard approach taken to solving refinement type inference problems, such as \citet{jhala-etal-cav2011,unno-kobayashi-PPDP2009,terauchi-popl2010}, is essentially a reduction to constrained Horn clause solving.  
Although only first-order, these systems of constraints are extremely expressive since they incorporate an arbitrary background theory, such as linear arithmetic or the theory of algebraic datatypes.
Consequently, they are typically undecidable.
Constrained Horn clauses were lifted to higher-order by \citet{burn2018hochc}, and the theory further explored in \citet{OngWagner2019}. 
In a follow-up work the same authors identified a family of decidable fragments intended for applications in database aggregation \cite{burn2021datalog}.

\changed[ch]{
\paragraph{Uniform proofs and logic programming.}
The formulation of our fragments and the proof system that underlies them follows the elegant presentation in the work of Miller and his collaborators, such as \citet{miller-etal1991uniform,miller_nadathur_2012}.  In particular, one can recognise their \emph{fohc}, \emph{hohc}, and \emph{hohh} as the underlying formalisms behind our $\MSL(1)$ clauses, $\HOMSL(\omega)$ clauses, and higher-order automaton clauses respectively.  Of course, we could have presented our fragments of HOL in a more traditional format for automated reasoning (e.g. with clauses as multisets of literals), but we consider the compositional characterisation that is characteristic of Miller's work essential for a clear exposition once we have to deal with the combination of nested clauses (in the sense of hereditary Harrop) and higher-order constructs.
}

\changed[ch]{
\paragraph{Constructive logic and `Horn Clauses as Types'}
Over a series of papers, Fu, Komendantskaya, and co-authors have presented a comprehensive analysis of Horn clauses and resolution according to the propositions-as-types tradition \cite{fu-Komendantskaya-lopstr2015,fu-etal-flops2016,fu-Komendantskaya-facs2017,farkaThesis}.
Like our work, they cast resolution as a form of rewriting, studying a number of different variations on the standard approach that have been motivated by the desire to capture computations with infinite data.  Using  Howard's System \textbf{H} \cite{howard}, they give a type theoretic semantics to each form of resolution, and this allows for a more meaningful notion of soundness and completeness than the traditional method using Herbrand models.  Since it is in the propositions-as-types tradition, their work views a Horn clause as the type of its proofs.  By contrast, we view an automaton clause with a single free variable $x$ as a type inhabited by the terms that satisfy the clause (when substituted for $x$).  Consequently, we do not make use of the constructive content of the resolution proofs themselves, but rather view resolution simply as a mechanism for generating a new clause from two given clauses -- i.e. a way to infer new types.
}

\begin{acks}                            
  We gratefully acknowledge the support of the \grantsponsor{EPSRC}{Engineering and Physical Sciences Research Council}{http://https://epsrc.ukri.org} (\grantnum{EPSRC}{EP/T006595/1}) and the National Centre for Cyber Security via the UK Research Institute in Verified Trustworthy Software Systems.
  We are also very grateful for the help of the reviewers in making the paper more clear and accurate, and for suggesting related work.
\end{acks}

\bibliography{references}

\appendix
\section{Semantics}
\label{appx:sem}

\newcommand\calA{\mathcal{A}}
\newcommand\trees[1]{\mathfrak{T}_{#1}}
\newcommand\cons[1]{c^{#1}}
\newcommand\limmod{\beta_\mathrm{lim}}  
\newcommand\ifun[1]{\mathsf{i}_{#1}}
\newcommand\jfun[1]{\mathsf{j}_{#1}}

Our fragments use the standard semantics of HOL.
An explicit denotational semantics is provided below.

\paragraph{Interpretation of types.}
Let $\Sigma$ be a constructor signature. 
A $\Sigma$-\emph{structure} $\calA$ assigns a non-empty set $A_\iota$ to the semantic domain of $\iota$ and an element $c^a \in \sem{\gamma}$ to each constructor $a : \gamma \in \Sigma$, with types interpreted as:
\[
  \sem{\iota} \defeq A_\iota
  \qquad \sem{o} \defeq \mathbb{B}
  \qquad \sem{\sigma_1 \to \sigma_2} \defeq \left[\sem{\sigma_1} \to \sem{\sigma_2}\right]
\]
where $\mathbb{B}$ denotes the boolean lattice and $\left[\sem{\sigma_1} \to \sem{\sigma_2}\right]$ denotes the full function space between $\sem{\sigma_1}$ and $\sem{\sigma_2}$, as dictated by the standard semantics of HOL.

\paragraph{Interpretation of terms and goal formulas.}
Let $\Pi'$ consist of predicate signature $\Pi$ and a variable environment $\Delta$. 
A $\Pi'$-\emph{valuation} (typically $\alpha$) is a function such that $\alpha(x) \in \sem{\sigma}$, for all $x:\sigma \in \Pi'$ (taking the liberty to write $x$ for both variables and predicates in $\Pi'$.)

By some abuse, we consider goal formulas as goal \emph{terms}, treating logical symbols as constants with their usual interpretations from HOL:
\begin{align*}
    \sem{x : \sigma}(\alpha) &= \alpha(x) \\
    \sem{a : \gamma}(\alpha) &= \cons{a} \\
    \sem{\truetm : o}(\alpha) &= 1 \\
    \sem{s \, t : \sigma_2}(\alpha) &= \sem{s : \sigma_1 \to \sigma_2}(\alpha)(\sem{t : \sigma_1}(\alpha)) \\
    \sem{ \lambda x:\sigma_1.\, t : \sigma_2}(\alpha) &= \lambda s \in \sem{\sigma_1}.\, \sem{t : \sigma_2}(\alpha[x \mapsto s])\\
    \sem{\wedge : o \to o \to o}(\alpha) &= \lambda b_1 \, b_2.\, \min\{ b_1, b_2 \} \\
    \sem{\exists_\sigma : (\sigma \to o) \to o}(\alpha) &= \lambda r.\, \max \{r\,s \mid s \in \sem{\sigma}\}
\end{align*}
where $\sigma$ is either a constructor type $\gamma$ or predicate type $\rho$, and the $\Sigma$-structure $\calA$ and environment $\Pi'$ are left implicit.

\paragraph{Interpretation of clauses.}
Interpreting clauses requires additional constants, again with their standard semantics.
Since the types are clear, we omit them.
\begin{align*}
    \sem{\vee}(\alpha) &= \lambda b_1 \, b_2.\, \max\{ b_1, b_2 \} &
    \sem{\neg}(\alpha) &= \lambda b.\, 1 - b \\
    \sem{{\Rightarrow}}(\alpha) &= \lambda b_1 \, b_2.\, \sem{\lor}\,(\sem{\neg}\,b_1)\,b_2  &
    \sem{\forall_\sigma}(\alpha) &= \lambda r.\, \min \{r\,s \mid s \in \sem{\sigma}\}
\end{align*}
We have allowed some redundancy in the logical symbols for convenience:
the semantics of $\forall_\sigma$ can equivalently be formulated in terms of $\neg$ and $\exists_\sigma$, and $\lor$ in terms of $\neg$ and $\land$.

The above determines the semantics of clauses-treated-as-terms.
This behaves as expected, with the semantics $\sem{\forall \yy.\, G \Rightarrow A}(\alpha)$ of definite clause $\forall \yy.\, G \Rightarrow A$ given by:
\begin{align*}
    \min \left\{\max \{\sem{A}(\alpha[\yy \mapsto \vv{s}]), 1 - \sem{G}(\alpha[\yy \mapsto \vv{s}]) \} \;\middle|\; \vphantom{\bigcup} \forall y_i \in \yy.\, s_i \in \sem{\sigma_i} \right\}
\end{align*}
I.e.~a clause valuates to $0$ iff there exists a valuation of the top-level universal variables such that the body valuates to $1$ and the head to $0$.

\paragraph{Satisfaction.}

A model of $D$ over $\Sigma$ and $\Pi$ consists of a $\Sigma$-structure $\calA$ and a $\Pi$-valuation $\alpha$ such that $\sem{D}(\alpha) = 1$.
We say $D \vDash G$ just if $\sem{G}(\alpha) = 1$ for all models of $D$.

\subsection{Completeness of the proof system}
\label{appx:complete_proof_system}

This section demonstrates soundness and completeness of the Figure~\ref{fig:prf-goal} proof system with respect to the standard semantics of HOL.

\proofSystemCompleteness*

Soundness follows from each rule being clearly admissible in HOL.

For completeness, suppose $D \models G$ over constructor signature $\Sigma$ and predicate signature $\Pi$.
We assume $D$ contains a universal clause $(\forall \yy.\, \truetm \Rightarrow \univRel_\rho\,\yy)$ for every type $\rho$ occurring in $D$ or $G$.
We refer to $\univRel_\rho$ as a \emph{universal relation}.
  
We construct an no-instance of the HOCHC satisfiability problem over the equational theory of first-order trees as background.
The instance is constructed by defunctionalising the higher-order constructors and simulating both pattern-matching clause heads and the higher-order typing discipline on first-order trees with equational constraints.

\paragraph{The HOCHC signature.}

Types $\tau$ are transformed to HOCHC types $\trSem{\tau}$:
\[
  \trSem{\gamma} \defeq \iota
  \qquad \trSem{o} \defeq o
  \qquad \trSem{\sigma \to \rho} \defeq \trSem{\sigma} \to \trSem{\rho}
\]

We flatten our constructors to first-order, where the only non-nullary constructor is binary application, written $\_\cdot{}\_$.
A family of new unary predicate symbols $T_\gamma$ is axiomatised so as to restrict $x : \gamma$ to represent a well-formed tree of type $\gamma$ under the original higher-order signature $\Sigma$.
  
Let $\existSort{D \land G}$ contain all types of existential variables in $D \land G$, so we can define first-order signature $\Sigma_\FO$ and set of predicates $\Delta$.
\begin{align*}
  \Sigma_\FO
    &\defeq \set{ f : \iota \mid f : \gamma \in \Sigma } \cup \set{ \_\cdot{}\_ : \iota \to \iota \to \iota  } \\
  \Delta 
    &\defeq \Pi \cup \set{ T_\gamma: \iota \to o \mid \gamma \in \existSort{D \land G} \lor (f : \gamma_1 \to \dots \to \gamma_n \in \Sigma \land \gamma \in \set{\gamma_1,\dots,\gamma_n}) } 
\end{align*}

\paragraph{The transformation.}

We define a term transformation $\trSem{t}$ by:
\changed[jj]{
\[
  \trSem{x} \defeq x 
  \qquad \trSem{c} \defeq c
  \qquad \trSem{P} \defeq P
  \qquad \trSem{s\,t} \defeq \left\{ \begin{array}{ll}
      \trSem{s} \cdot \trSem{t} & \text{if } s\,t : \gamma \\
      \trSem{s}\,\trSem{t} & \text{o/w}
    \end{array}\right.
\] }
We extend this to goal formulas $\trSem{G}$ by: 
\begin{alignat*}{3}
  \trSem{\truetm} & \defeq \truetm & \qquad\quad
  \trSem{\exists x:\rho.\,G} & \defeq \exists x:\trSem{\rho}.\, \trSem{G} \\
  \trSem{G \wedge H} & \defeq \trSem{G} \wedge \trSem{H} & \qquad\quad 
  \trSem{\exists x:\gamma.\,G} & \defeq \exists x:\iota.\, T_\gamma\,x \land \trSem{G} 
\end{alignat*}
We extend to definite formulas $\trSem{C}$ by:
\begin{alignat*}{3}
  \trSem{\truetm} & \defeq \truetm & \qquad\quad
  \trSem{\forall \vv{y}.\, G \implies P\,\vv{y}} & \defeq \forall \vv{y}.\, \trSem{G} \implies P\,\vv{y}\\
  \trSem{C \wedge D} & \defeq \trSem{C} \wedge \trSem{D} & \qquad\quad 
  \trSem{\forall \vv{y}.\, G \implies P\,(c\,\vv{y})} &\defeq \forall z.\, \trSem{G}_{z=c\yy} \implies P\,z 
\end{alignat*}
where
\[
  \trSem{G}_{z=c\yy} \defeq \exists \yy.\, z = \trSem{c\,\yy} \land \bigwedge_{y_i:\gamma_i \in \yy} T_{\gamma_i}\,y_i \land \trSem{G}
\]
  
Finally, we define an equi-satisfiable HOCHC instance over the first-order equational theory of trees with signature $\Sigma_\FO$ and predicates $\Delta$:
\begin{align*}
  \Dhochc &\defeq \trSem{D} \land \set{ \forall x.\, \trSem{\truetm}_{x=f y_1 \dots y_n} \implies T_\gamma\,x \mid f:\gamma_1 \to \dots \to \gamma_n \to \gamma \in \Sigma } \\
  \Ghochc &\defeq \trSem{G}
\end{align*}

The transformation $\trSem{\_}$ has a straightforward inverse $\rtSem{\_}$ on existential-free goal formulas.

Note that all universal variables are guarded by a $T_\gamma$ of the appropriate type.
Since universal variables do not occur in goal formulas, those do not need to be guarded.

%

\paragraph{The HOCHC proof system.}
      
\citet{OngWagner2019} propose a simple, sound and refutationally complete resolution proof system for HOCHC (under the standard semantics of HOL) that we shall call the \emph{HOCHC proof system}.
Any resolution proof for our HOCHC instance requires only two of their three proof rules, due to the lack of internal $\lambda$s and the fact that the HOCHC proof system does not introduce those.
In the terminology of our paper, these two proof rules are:
\[
  \begin{array}{cc}
    \!
    \prftree[l]
      { \OWRef }
      { \textstyle \exists \yy. \bigwedge_{i \in [1..n]} \varphi_i \land \exists \xx. \bigwedge_{j \in [1..m]} x_{k_j}\,\vv{M_j} }
      { \top }
      
    \;\,

    \prftree[l,r]
      { \!\!\!\! $(\forall \yy.\,G' \Rightarrow P\,\yy) \in \Dhochc$ }
      { \OWRes }
      { \exists \zz.\, G \land P\,\vv{M} }
      { \exists \zz.\, G \land G'[\vv{M}/\yy] }
      \end{array}
  \]
where $\varphi_1,\dots,\varphi_n$ are constraints from the equational theory of trees for which there exists a satisfying assignment $\alpha$ (i.e.~$\alpha \vDash \varphi_1 \land \dots \land \varphi_n$ such that $\alpha(y)$ is a finite tree over $\Sigma_\FO$, for all $y \in \yy$).
Note that $\xx$ in \OWRef~are higher-order and $\yy$ first-order existentials.
Proof rule \OWRef~implicitly assigns universal relations to $\xx$;
these are part of the semantic domain and guaranteed to satisfy the formula. 
            
We write $\Dhochc \vdash G' \tto G''$ to denote $G'$ as a premise for $G''$ in the HOCHC proof system, with $\Dhochc \vdash G' \tto^* G''$ its reflexive, transitive closure.

\begin{notation}
  \label{lem:sem_shape}
  Every HOCHC goal formula $G'$ can be written as $\exists \yfo \, \xho.\, \gphi \land \gv \land \ga$ where:
  \begin{itemize}
    \item $\gphi$ is existential-free and contains the constraints and atoms headed by $T_\gamma$,
    \item $\gv$ is existential-free and contains the atoms headed by $x \in \xho$, and
    \item $\ga$ is existential-free and contains the remaining atoms.
  \end{itemize}
  We may write $\gnphi$ for $\gv \land \ga$.
\end{notation}

\paragraph{HOCHC semantics.}

Recall that $D \vDash G$ means ``every model of $D$ satisfies $G$'', much like $\Dhochc \vDash \Ghochc$ means ``every model of $\Dhochc$ satisfies $\Ghochc$''.
Note, however, that a HOCHC instance has a fixed background theory from which it inherits $\sem{\iota}$ and interpretations of the constants in $\Sigma$.
Intuitively, our quantifier ranges over all $\Pi$-valuations and $\Sigma$-structures, while the HOCHC quantifier only ranges over $\Pi$-valuations that extend the background theory.

Consequently, we can choose our $\Sigma$-structure in the HOCHC case to match the background theory of trees, and it is not difficult to see that $D \vDash G$ implies $\Dhochc \vDash \Ghochc$.

\paragraph{Proof outline.}

By refutation completeness of the HOCHC proof system, there exists a HOCHC proof $\Dhochc \vdash \Ghochc \tto^* \top$ with satisfying assignment $\alpha$ to the first-order existentials.

We replay this HOCHC proof in our Figure~\ref{fig:prf-goal} proof system starting from $G$.  
Since all existential variables are guarded by typing predicates $T_\gamma$ and all elements of the universe of the theory of (finite) trees are representable by terms, the satisfying assignment can be explicitly constructed and implemented using \ruleName{Ex}.
Furthermore, the presence of universal clauses allow us to prove existential-headed atoms using \ruleName{Ex} as well\footnote{Although the type of a HOCHC existential may not match the corresponding original existential, the existence of the appropriate universal clause in $D$ is guaranteed.}.
Hence $D \vdash G$.

\paragraph{From constraints to type-respecting substitutions.}

The purpose of our equality constraints and $T_\gamma$ predicates is to restrict the domain of HOCHC variable $y : \iota$ to well-formed trees of type $\gamma$ under the original higher-order signature;
they guarantee that every satisfying assignment $\alpha$ from an application of \OWRef~respects the original higher-order types.
 
We say that an assignment $\alpha$ to $\yy$ is \emph{type-respecting} just if, for all $y \in \yy$, there exists a closed tree $t:\gamma$ over $\Sigma$ such that $\alpha(y) = \trSem{t}$ and $\gamma$ is the original type of $y$. 
 
\begin{lemma}
  \label{lem:sem_single_constraint}
  If $\Dhochc \vdash \exists y.\, T_{\gamma}\,y \tto^* H \tto \top$ with satisfying assignment $\alpha$, then $\alpha$ is type-respecting.
\end{lemma}
\begin{proof}
  \citet{OngWagner2019} leave existential quantifiers implicit in the HOCHC proof system.
  For convenience, we do the same here.
  We use a straightforward induction on length $k$ in $\Dhochc \vdash T_\gamma\,y \tto^k  H \tto \top$.
  
  Clearly, $k$ cannot be zero, so suppose $k=1$.
  This first proof step must be \OWRes~with a side condition $(\forall x.\, x=\trSem{f} \Rightarrow T_\gamma\,x) \in \Dhochc$ with $f:\gamma \in \Sigma$, so $H = (y = \trSem{f})$.
  This means that $y$ equals $\trSem{f}$ under $\alpha$, so $\alpha(y) = \trSem{f}$ witnesses that $\alpha$ is type-respecting.
  
  Suppose $k > 1$ with side condition $(\forall x.\, \exists \zz.\, x = \trSem{f\,\zz} \land \bigwedge_{j \in [1..m]} T_{\gamma_j}\,z_j \Rightarrow T_\gamma\,x) \in \Dhochc$ with $f:\gamma_1 \to \dots \to \gamma_m \to \gamma \in \Sigma$.
  The result of this first step is $G'' = (y = \trSem{f\,\zz}) \land \bigwedge_{j \in [1..m]} T_{\gamma_j}\,z_j$, still omitting quantifiers.
  
  For all $j \in [1..m]$, $\alpha$ is a satisfying assignment for some $\Dhochc \vdash T_{\gamma_j}\,z_j \tto^{k_j} \bot$ with $k_j < k$.
  By the IH, $\alpha(z_j) = \trSem{t_j}$ for closed tree $t_j:\gamma_j$ over $\Sigma$. 
  
  Now, $\alpha$ satisfies $H$ and, by soundness of the HOCHC proof system, also $G''$.
  In particular, this means $y$ equals $\trSem{f\,\zz}$ under $\alpha$, so $\alpha(y)$ equals $\trSem{f\,\vv{t}}$ for closed tree $f\,\vv{t} : \gamma$ over $\Sigma$, as required.
\end{proof}

Observe from the transformation $\trSem{\_}$ that all first-order existentials occurring in a HOCHC proof of a $\trSem{G'}$ are guarded by the appropriate $T_\gamma$. 
It is easy to see that the above lemma generalises to $\Dhochc \vdash \exists \yfo\,\xho.\, \bigwedge_{y_i:\gamma_i \in \yfo} T_{\gamma_i} \, y_i \land G^0 \tto^* H \tto \top$ with existential-free $G^0$.
Since every $\trSem{G'}$ can be written in this shape, 
the following holds.

\begin{lemma}
  \label{lem:sem_type_respecting}
  If $\Dhochc \vdash \trSem{G'} \tto^* \top$ with satisfying assignment $\alpha$, then $\alpha$ is type-respecting.
\end{lemma}

Note that any satisfying assignment $\alpha$ to $\Dhochc \vdash \exists \yfo\, \xho.\, \gphi \land \gnphi \tto^* \top$ is also a satisfying assignment to $\Dhochc \vdash \exists \yfo\, \xho.\, \gnphi \tto^* \top$. 
We know that $\alpha$ respects original types, so $[\rtSem{\res{\alpha}{\yfo}}/\yfo]$ is a well-defined substitution in our original program, giving rise to the following lemma that allows us to replay a HOCHC proof in our Figure~\ref{fig:prf-goal} proof system.

\begin{lemma}[Replay lemma]
  \label{lem:sem_replay}
  If $\Dhochc \vdash \exists \yfo \, \xho.\, \gphi \land \gnphi \tto^* \top$ with satisfying assignment $\alpha$, then 
  \[
    D \vdash \rtSem{\gnphi }[\rtSem{\res{\alpha}{\yfo}}/\yfo][\vv{\univRel}/\xho]
  \]
\end{lemma}
\begin{proof}
  Suppose $\Dhochc \vdash G' \tto^* H \tto \top$ with satisfying assignment $\alpha$.
  Because \OWRes~is a local rule and \OWRef~a global one, and $T_\gamma$-clauses only have other $T_{\gamma'}$-predicates in their bodies, we may assume WLOG that all $T_\gamma$-headed resolutions in a HOCHC proof are performed last.

  Suppose our HOCHC proof has shape $\Dhochc \vdash G' \tto^k G'' \tto^\ell H \tto \top$ for a $G''$ where the first $k$ steps do not use $T_\gamma$-headed clauses, followed by $\ell$ steps that (exclusively) do.
  We use induction on $k$.
  \begin{itemize}

    \item If $k=0$, then let $G' = G'' = \exists \yfo \, \xho.\, \gphi \land \gv$ with $\xho : \trSem{\vv{\rho}}$.
      Let $\gv = \bigwedge_{j \in [1..m]} x_{i_j}\,\vv{M_j}$ with $x_{i_1},\dots,x_{i_m} \in \xho$.
      We prove $D \vdash \rtSem{ \gv }[\rtSem{\res{\alpha}{\yfo}}/\yfo][\vv{\univRel}/\xho]$ where $\vv{\univRel} : \vv{\rho}$:
      \[
 \prftree[l]
            { \ruleName{And}$^*$ }
            { \forall j \in [1..m].\!\!\!\! }
            { \prftree[l]
              { \ruleName{Res} }
              { \prftree[l]
                { \ruleName{T} }
                { \Dhochc \vdash \truetm } 
              }
              { \Dhochc \vdash \univRel_{i_j}\,\rtSem{\vv{M_j}}[\rtSem{\res{\alpha}{\yfo}}/\yfo][\vv{\univRel}/\xho] } 
            }
            { \Dhochc \vdash \bigwedge_{j \in [1..m]} \univRel_{i_j}\,\rtSem{\vv{M_j}}[\rtSem{\res{\alpha}{\yfo}}/\yfo][\vv{\univRel}/\xho] }
      \]
   
    \item If $k > 0$, then let $G' = \exists \yfo\, \xho.\, \gphi\land \gnphi \land P\,\ss$ such that the first step is \OWRes~on $P\,\ss$. 
      Suppose $\xho$ and $\vv{\univRel}$ with $\rho$ as in the previous case, and $\xho^0$ and $\vv{\univRel^0}$ with $\rho^0$ analogous.
      
      We shall abbreviate $[\rtSem{\res{\alpha}{\vv{x_1}}}/\vv{x_1},\dots,\rtSem{\res{\alpha}{\vv{x_n}}}/\vv{x_n}]$ to $[\rtSem{\res{\alpha}{\vv{x_1}\dots\vv{x_n}}}/\vv{x_1}\dots\vv{x_n}]$.
      \begin{itemize}

        \item Suppose the side condition is $(\forall \zz.\, \trSem{\exists \yfo^0\,\xho^0.\, G^0} \Rightarrow P\,\zz) \in \Dhochc$ with existential-free $G^0$ from some $(\forall \zz.\, \exists \yfo^0\,\xho^0.\, G^0 \Rightarrow P\,\zz) \in D$.
          In this case, $\trSem{\exists \yfo^0\,\xho^0.\, G^0}$ equals some $\exists \yfo^0\,\xho^0.\, \bigwedge_{y_j^0 \in \yfo^0} T_{\gamma_j^0}\,y_j^0 \land \trSem{G^0}$.
          This gives us:
          \[ 
            G'' = \exists \yfo\,\yfo^0\, \xho\,\xho^0.\, \gphi \land \gnphi \land \bigwedge_{y_j^0 \in \yfo^0} T_{\gamma_j^0}\,y_j^0 \land \trSem{G^0}[\ss/\zz] 
          \]
          
          We aim to show $D \vdash \rtSem{ \gnphi \land P\,\ss }[\rtSem{\res{\alpha}{\yfo}}/\yfo][\vv{\univRel}/\xho]$, which holds if both of these claims hold:
          \begin{align}
            D &\vdash \rtSem{\gnphi}[\rtSem{\res{\alpha}{\yfo}}/\yfo][\vv{\univRel}/\xho] \label{eq:sem_first_claim} \tag{C1}\\
            D &\vdash P\,\rtSem{\ss}[\rtSem{\res{\alpha}{\yfo}}/\yfo][\vv{\univRel}/\xho] \label{eq:sem_second_claim} \tag{C2}
          \end{align}
  
          Note $D \vdash \rtSem{\gnphi \land \trSem{G^0}[\ss/\zz] }[\rtSem{\res{\alpha}{\yfo\yfo^0}}/\yfo\,\yfo^0][\vv{\univRel}/\xho, \vv{\univRel^0}/\xho^0]$, i.e. $D \vdash(\rtSem{\gnphi} \land G^0[\rtSem{\ss}/\zz] )[\rtSem{\res{\alpha}{\yfo\yfo^0}}/\yfo\,\yfo^0][\vv{\univRel}/\xho, \vv{\univRel^0}/\xho^0]$, by the IH on $G''$.
          This implies the following, by \ruleName{And}: 
          \begin{align}
            D &\vdash \rtSem{\gnphi}[\rtSem{\res{\alpha}{\yfo\yfo^0}}/\yfo\,\yfo^0][\vv{\univRel}/\xho, \vv{\univRel^0}/\xho^0] \label{eq:sem_first_conj} \tag{IH1} \\
            D &\vdash G^0[\rtSem{\ss}/\zz][\rtSem{\res{\alpha}{\yfo\yfo^0}}/\yfo\,\yfo^0][\vv{\univRel}/\xho, \vv{\univRel^0}/\xho^0] \label{eq:sem_second_conj} \tag{IH2}
          \end{align}
  
          Since $\rtSem{\gnphi}$ does not contain $\yfo^0$ or $\xho^0$, \ref{eq:sem_first_conj} gives us \ref{eq:sem_first_claim} immediately.

          We apply \ruleName{Ex} to \ref{eq:sem_second_conj} for $D \vdash \exists \yfo^0\,\xho^0.\, G^0[\rtSem{\ss}/\zz][\rtSem{\res{\alpha}{\yfo}}/\yfo,\vv{\univRel}/\xho]$.
          This is now an instance of the body of the original definite clause $(\forall \zz.\, \exists \yfo^0\,\xho^0.\, G^0 \Rightarrow P\,\zz) \in D$, so we apply \ruleName{Res} to obtain \ref{eq:sem_second_claim}.
              
        \item Suppose the side condition is $(\forall v.\, \trSem{\exists \yfo^0\,\xho^0.\, G^0}_{f\,\zz} \Rightarrow P\,v) \in \Dhochc$ with existential-free $G^0$ from some $(\forall \zz.\, \exists \yfo^0\,\xho^0.\, G^0 \Rightarrow P\,(f\,\zz)) \in D$.
          Then, $\trSem{\exists \yfo^0\,\xho^0.\, G^0}_{f\,\zz}$ equals some $\exists\yfo^0\,\zz\, \xho^0.\,  \bigwedge_{y_j^0:\gamma^0_j \in \yy} T_{\gamma^0_j}\,y_j^0 \land (v = \trSem{f\,\zz}) \land \bigwedge_{z_i:\gamma^1_i \in \zz} T_{\gamma^1_i}\,z_i \land \trSem{G^0}$.

          This gives us:
          \begin{align*}
            G'' = \exists \yfo\,\yfo^0\,\zz\, \xho\,\xho^0.\, 
              &\gphi \land \gnphi \land \bigwedge_{y_j^0:\gamma^0_j \in \yy} T_{\gamma^0_j}\,y_j^0 \land (s_1 = \trSem{f\,\zz}) \\
              &\land\; \bigwedge_{z_i:\gamma^1_i \in \zz} T_{\gamma^1_i}\,z_i \land \trSem{G^0}[s_1/v]
          \end{align*}

          By the IH, $D \vdash \rtSem{\gnphi \land \trSem{G^0}[s_1/v]}[\rtSem{\res{\alpha}{\yfo\yfo^0\zz}}/\yfo\,\yfo^0\,\zz][\vv{\univRel}/\xho, \vv{\univRel^0}/\xho^0]$.
          However, $v$ does not occur in $\trSem{G^0}$, so $\trSem{G^0}[s_1/v] = \trSem{G^0}$, which yields:
          \[
            D \vdash (\rtSem{\gnphi} \land G^0)[\rtSem{\res{\alpha}{\yfo\yfo^0\zz}}/\yfo\,\yfo^0\,\zz][\vv{\univRel}/\xho, \vv{\univRel^0}/\xho^0]
          \]  
                
          Compared with the IH in the previous case, only the substitution to $\zz$ is different;
          instead of $\zz$ first being partially instantiated with $\rtSem{\ss}$ (in the previous case), $\zz$ is directly fully instantiated by $\alpha$.
                
          We know from Lemma~\ref{lem:sem_type_respecting} that $\alpha$ assigns each $z_i \in \zz$ a closed term $\trSem{t'_i}$ such that $t'_i : \gamma^1_i$.
          Furthermore, $\alpha$ satisfies $H$ and, by soundness of the HOCHC proof system, also $G''$.
          In particular, this means $s_1$ equals $\trSem{f\,\zz}$ under $\alpha$, and $\rtSem{s_1}[\rtSem{\res{\alpha}{\yfo}}/\yfo]$ equals $f\,\vv{t'}$.
          
          Since we set out to prove $D \vdash \rtSem{\gnphi \land P\,s_1}[\rtSem{\res{\alpha}{\yfo}}/\yfo][\vv{\univRel}/\xho]$, it suffices to show 
          \[
            D \vdash \rtSem{\gnphi}[\rtSem{\res{\alpha}{\yfo}}/\yfo][\vv{\univRel}/\xho] \land P\,(f\,\vv{t'})
          \]
                
          We now proceed as in the previous case.
      \end{itemize}
   
  \end{itemize}
\end{proof}

\begin{proof}[Proof of Theorem~\ref{thm:proof_system}]
  Since $\Ghochc = \trSem{G}$, it remains to show $\Dhochc \vdash \trSem{G} \tto^* \top$ implies $D \vdash G$.
  
  Let $G = \exists \yfo\, \xho.\, \gv \land \ga$, so $\Ghochc = \trSem{G} = \exists \yfo\, \xho.\, \gphi \land \trSem{\gv \land \ga}$.
  This gives us $\Dhochc \vdash \exists \yfo\, \xho.\, \gphi \land \trSem{\gv \land \ga} \tto^* \top$ for a $\gphi$ and satisfying, type-respecting assignment $\alpha$.
  
  By Replay Lemma~\ref{lem:sem_replay}, $D \vdash \rtSem{ \trSem{\gv \land \ga} }[\rtSem{\res{\alpha}{\yfo}}/\yfo][\vv{\univRel}/\xho]$, which is:
  \[
    D \vdash (\gv \land \ga)[\rtSem{\res{\alpha}{\yfo}}/\yfo][\vv{\univRel}/\xho]
  \]
  Finally, we prove $Ds \vdash G$ with \ruleName{Ex} and conclude that the Figure~\ref{fig:prf-goal} proof system is complete with respect to the standard semantics of HOL.
\end{proof}

\subsection{Elimination of predicate-type existentials}

\begin{corollary}
  Predicate-type existentials do not add expressive power to \HOMSL($\omega$).
\end{corollary}
\begin{proof}
  We prove that $D \vdash \exists x_0:\rho.\, G$ if, and only if, $D \vdash G[\mho_\rho/x_0]$, relying on the reduction of provability in \HOMSL($\omega$) to provability in HOCHC with a decidable equational theory of first-order trees (see Theorem~\ref{thm:proof_system}, soundness and completeness of the proof system).
  
  Direction $\Leftarrow$ is immediate from \ruleName{Ex}, so suppose $D \vdash \exists x_0:\rho.\, G$.
  There exists a HOCHC proof:
  \[
    \textstyle \Dhochc \vdash \exists x_0:\trSem{\rho}.\, \trSem{G} \tto^* \exists x_0\,\xx\,\yy.\,\bigwedge_{i \in [1..n]} \varphi_i \bigwedge_{j \in [1..m]} x_{k_j}\,\vv{M_j} \tto \top
  \]
  with a satisfying assignment $\alpha$ to the existentials $\yy$ with constructor types.
  The existentials $x_0$ and $\xx$ with predicate types are implicitly assigned the maximal element in the semantic domain of their respective types, which is captured by $\mho_{\trSem{\rho'}}$ for type $\trSem{\rho'}$.
  
  This means there is a proof $\Dhochc \vdash \trSem{G}[\mho_{\trSem{\rho}}/x_0] \tto^* \top$.
  Because $\trSem{G}[\mho_{\trSem{\rho}}/x_0] = \trSem{G[\mho_{\rho}/x_0]}$, the final part of the proof of Theorem~\ref{thm:proof_system} gives us $D \vdash G[\mho_\rho/x_0]$.
  
  We conclude that preprocessing \HOMSL($\omega$) can eliminate predicate-type existentials by instantiating them with universal relations.
\end{proof}

\section{Supporting Material for Section~\ref{sec:reductions}}
\label{appx:reductions}

\subsection{Elimination of higher-order predicates}
\label{appx:HO_elim}

Note that the transformation on goal formulas has a straightforward retraction $\tr{\_}^{-1}$ which removes occurrences of $\mathsf{T}$ and the $\#$.

\begin{lemma}
  $D \vdash_\omega G$ implies $\tr{D} \vdash_1 \tr{G}$
\end{lemma}
\begin{proof}
  The proof is by induction on the proof object in the antecedent.
  \begin{description}
    \item[\ruleName{T}] Clear by definition.
    \item[\ruleName{Res}] In this case, $G$ is of the form $A[\vv{s}/\vv{y}]$ and there is a clause $(\forall \yy.\, H \implies A) \in D$.  By induction, we may assume that $\tr{D} \vdash_1 \tr{H[\vv{s}/\vv{y}]}$.  Note that, by definition, $\tr{H[\vv{s}/\vv{y}]} = \tr{H}[\tr{\vv{s}}/\vv{y}]$ and $\tr{A[\vv{s}/\vv{y}]} = \tr{A}[\tr{\vv{s}}/\vv{y}]$.  We consider two cases for $A$.
    \begin{itemize}
      \item If $A$ is of shape $P\,(c\,\vv{y})$, then $\tr{A} = \mathsf{T}\,(p^\#\,(c\,\vv{y}))$ and we  prove $\tr{A}[\tr{\vv{s}}/\vv{y}] = \mathsf{T}\,(p^\#\,(c\,\tr{\vv{s}}))$ by two applications of \ruleName{Res}.  The first uses the reflection clause from $\tr{D}$ to unwrap the $\mathsf{T}$ predicate, obtaining $P\,(c\,\tr{\vv{s}})$, and the second uses the transformed copy of the side premise to obtain $\tr{H}[\tr{\vv{s}}/\vv{y}]$, which is provable by the induction hypothesis.
      \item Otherwise, $A$ is of the form $P\,\vv{y}$ and $\tr{A} = \mathsf{T}\,(p^\#\,\vv{y})$.  We can prove $\tr{A}[\vv{\tr{s}}/\vv{y}]$ by \ruleName{Res} directly using the transformed copy of the side premise and the induction hypothesis.  
    \end{itemize} 
    \item[\ruleName{Ex}] In this case, $G$ is of the form $\exists x.\,H$ and $\tr{G} = \exists x.\, \tr{H}$.  By induction, we may assume $\tr{H[t/x]}$, which is synonymous with $\tr{H}[\tr{t}/x]$, is already provable.  Therefore, the proof can be concluded using \ruleName{Ex}.
    \item[\ruleName{And}] Follows immediately from the induction hypotheses.
  \end{description}
\end{proof}

We need to generalise a little bit in the converse direction.

\begin{lemma}
  $\tr{D} \vdash_1 G$ implies $D \vdash_\omega \rt{G}$
\end{lemma}
\begin{proof}
  The proof is by induction.
  \begin{description}
    \item[\ruleName{T}] Clear by definition.
    \item[\ruleName{Res}] In this case, $G$ has the form $A[\vv{s}/\vv{y}]$ and there is a clause $(\forall \yy.\,H \implies A) \in \tr{D}$.  We distinguish cases on the form of $A$ and the provenance of the side premise.  
    \begin{itemize}
      \item If $A$ is of the form $P\,(c\,\vv{y})$, then there is a clause $(\forall \yy.\,\rt{H} \implies P\,(c\,\vv{y})) \in D$ from which the side premise originates.  From the induction hypothesis, we can assume that $D \vdash_\omega \rt{H[\vv{s}/\vv{y}]}$, which is, by definition, $D \vdash_\omega \rt{H}[\rt{\vv{s}}/\vv{y]}]$.  Hence, we can conclude $\rt{G}$, namely $P\,(c\,\rt{\vv{s}})$, directly by \ruleName{Res}.
      \item If $A$ is of the form $\mathsf{T}\,(p^\#\,\vv{y})$ and the side premise is not a reflection, then there is a clause $(\forall \yy.\, \rt{H} \implies P\,\yy) \in D$.  From the induction hypothesis we can assume $D \vdash_\omega \rt{H[\vv{s}/\vv{y}]}$, which is, by definition, $D \vdash_\omega \rt{H}[\rt{\vv{s}}/\vv{y}]$.  We can conclude $\rt{G}$, namely $P\,\rt{\vv{s}}$, directly by \ruleName{Res}.
      \item Otherwise the side premise is a reflection, has shape $\forall y.\, P\,y \implies \mathsf{T}\,(p^\#\,y)$, and $A$ is of the form $\mathsf{T}\,(p^\#\,y)$.  We may assume, by the induction hypothesis, that $D \vdash_\omega \rt{P\,s}$, which is $D \vdash_\omega P\,\rt{s}$ by definition.  But this is already $D \vdash_\omega \rt{G}$, so we are done.
    \end{itemize}
    \item[\ruleName{Ex}] In this case, $G$ has the form $\exists x.\,H$ and it follows from the induction hypothesis that there is some $t$ such that $D \vdash_\omega \rt{H[t/x]}$.  By definition, this means we have $D \vdash_\omega \rt{H}[\rt{t}/x]$ and so we can conclude $D \vdash_\omega \rt{G}$, being $\exists x.\, \rt{H}$, immediately by \ruleName{Ex}.
    \item[\ruleName{And}] This case follows directly from the induction hypotheses.
  \end{description}
\end{proof}

It now follows immediately from these lemmas that $D \proves_\omega G$ iff $\tr{D} \vdash_1 \tr{G}$, proving Theorem~\ref{thm:HOMSL(omega)_to_MSL(omega)}.

\subsection{Elimination of existentials}
\label{appx:existentials}

\begin{lemma}
  \label{lem:aux_ex_generates_all_instances}
  For all terms \changed[jj]{$\exists_\gamma\,s$ over signatures $\exf{\Sigma}$ and $\exf{\Pi}$,} $\exfD \vdash \exists_{\gamma}\, s$ if, and only if, $\exfD \vdash s\, t$ for a closed $t:\gamma$.
\end{lemma}
\begin{proof}
  For $\Rightarrow$, let $\tree$ be an underlying proof tree of $\exfD \vdash \exists_{\gamma}\, s$.
  We prove $\tree$ contains $\exfD \vdash s\, t$ for a closed $t:\gamma$, by induction on the number $m$ of proof steps headed by any $\exists_{\gamma'}$ external to $s$.
  \begin{itemize}

    \item If $m=1$, then, \changed[jj]{by our assumption of adequate signatures,} there exists $c:\gamma \in \Sigma$ such that a $\tree$ has shape:
          \[
            \prftree[l]
            { \ruleName{Res} }
            { \prftree[l]
              { \ruleName{Res} }
              { \exfD \vdash s\,c }
              { \exfD \vdash \mathsf{Comp}_c^{0,0}\,s } 
            }
            { \exfD \vdash \exists_{\gamma}\, s }
          \]

    \item If $m>1$, then $\tree(1) = \exfD \vdash \mathsf{Comp}^{0,k}_f\,s$ for a $k>0$ and constructor $f:\gamma_1 \to \dots \to \gamma_k \to \gamma$.
          We apply the IH to the subproof rooted at node $111$ (which clearly has $m'_1 < m$).

          Now, if $k=1$, then we use \ruleName{Res} with $\forall v\,x_1.\, v\,(f\,x_1) \Rightarrow \mathsf{Comp}_f^{1,1} \, v \, x_1$ to give us $\exf{D} \vdash s\,(f\,t_1)$, as required.
          Otherwise, we use \ruleName{Res} with $\forall v\, x_1.\, \exists_{\gamma_2}\,(\mathsf{Comp}^{2,k}\,v\,x_1) \Rightarrow \mathsf{Comp}^{1,k}_f\, v\, x_1$ to give us $\exf{D} \vdash \exists_{\gamma_2}\,(\mathsf{Comp}^{2,k}_f\,s\,t_1)$ to which we apply the IH again.
          
          We use the IH a total of $k$ times, once for each $\gamma_i$, for a proof tree of the following shape:
          \[
            \prftree[l]
            { \ruleName{Res} }
            { \prftree[l]
              { \ruleName{Res} }
              { \prftree[l]
                { \ruleName{IH} }
                {
                  \prftree
                  { \prftree
                    { \prftree[l]
                      { \ruleName{Res} }
                      {\exfD \vdash s\ (f\ t_1 \dots t_k) }
                      { \exfD \vdash \mathsf{Comp}_f^{k,k}\ s\ t_1\ \cdots\ t_k } 
                    }
                    { \vdots }
                  }
                  { \exfD \vdash \mathsf{Comp}^{1,k}_f\,s\,t_1 }
                }
                { \exfD \vdash \exists_{\gamma_1}\,(\mathsf{Comp}^{1,k}_f\,s) } 
              }
              { \exfD \vdash \mathsf{Comp}^{0,k}_f\,s  }
            }
            { \exfD \vdash \exists_{\gamma}\, s }
          \]

  \end{itemize}

  The direction $\Leftarrow$ is a straightforward inverse of hereof, using induction on the number of symbols in $t$.
  %
  %
  %
  %
\end{proof}

\begin{lemma}
  For all closed subformulas $G$ of the original clauses, $D \vdash G$ if and only if $\exfD \vdash \exfMap{G}$.
\end{lemma}
\begin{proof}
  In the forwards direction, we prove the following for all subformulas $G$ of the original clauses with free variables $\vv{z}$:
  \[
    D \vdash G[\vv{t/z}] \quad\text{implies}\quad \exfD \vdash \exfMap{G}[\vv{t/z}]
  \]
  The proof is by induction on $D \vdash G$.
  \begin{description}
    \item[\ruleName{T}] Immediate.
    \item[\ruleName{And}] Immediate from the induction hypothesis and the definition of the mapping.
    \item[\ruleName{Ex}] In this case $G[\vv{t/z}]$ is of shape $\exists x:\gamma.\, G'[\vv{t/z}]$ and $\exfMap{G} = \exists_\gamma\,(\Lambda_{\gamma,\,G'}\,\vv{z})$.  It follows from the induction hypothesis that $\exfD \vdash \exfMap{G'}[\vv{t/z}][u/x]$ for some closed $u$.  Hence, it follows from \ruleName{Res} that $\exfD \vdash \Lambda_{\gamma,G'}\,\vv{t}\,u$.  Then, from Lemma~\ref{lem:aux_ex_generates_all_instances}, we obtain $\exfD \vdash \exists_\gamma\,(\Lambda_{\gamma,G'}\,\vv{t})$, which is just $\exfD \vdash \exfMap{G}[\vv{t/z}]$ as required.
    \item[\ruleName{Res}] In this case $G[\vv{t/z}]$ is of shape $A[\vv{s/y}]$ and there is a clause $\forall \vv{y}.\ G' \Rightarrow A$ in $D$.  Then there is a clause $\forall \vv{y}.\ \exfMap{G'} \Rightarrow A$ in $\exfD$.  It follows from the induction hypothesis that $\exfD \vdash \exfMap{G'}[\vv{s/y}]$.  Therefore, by \ruleName{Res}, we get $\exfD \vdash A[\vv{s/y}]$.  Since $G[\vv{t/z}] = A[\vv{s/y}]$ with $A$ the head of a clause, we have that $\exfMap{G} = G$ and hence our conclusion is already in the desired form.
  \end{description}

  For the backward direction, we prove the following for all subformulas $H'$ of the original clauses with free variables $\vv{z}$:
  \[
    \exfD \vdash H \quad\text{and}\quad H = \exfMap{H'}[\vv{t/z}] \quad\text{and}\quad \fv(H') = \vv{z} \quad\text{implies}\quad D \vdash H'[\vv{t/z}]
  \]
  The proof is by strong induction on $\exfD \vdash H$:
  \begin{description}
    \item[\ruleName{T}] Immediate.
    \item[\ruleName{And}] Immediate from the induction hypothesis and the definition of the mapping.
    \item[\ruleName{Ex}] In this case, $H$ is of shape $\exists x\:\gamma.\ G$, but there is no $G'$ such that $\exists x\:\gamma.\ G = \exfMap{G'}[\vv{t/z}]$.  So the result is vacuously true.
    \item[\ruleName{Res}] In this case, $H$ is of shape $A[\vv{s/y}]$ and we have $\forall \vv{y}.\ G \implies A \in \exfD$.  We analyse cases on the predicate that heads $A[\vv{s/y}]$ (which is just the predicate that heads $A$).  As in the previous case, this predicate cannot be $\Lambda_{\gamma,K}$ or $\mathsf{Comp}_f^{i,j}$ since neither of those heads a formula in the image of $\exfMap{-}$.
      \begin{itemize}
        \item If the predicate is $\exists_\gamma$, we reason as follows.  In this case, since $A$ heads a clause in $\exfD$, it must be that $A$ is $\exists_\gamma\,y$ for some variable $y$.  Therefore, $\vv{s}$ is a singleton $s$.  Suppose $A[\vv{s/y}] = \exfMap{H'}[\vv{t/z}]$ for some $H'$.  Since there are no existential quantifiers in $A[\vv{s/y}]$ and $A$ is headed by $\exists_\gamma$, $\exfMap{H'}$ must be of shape $\exists_\gamma\,(\Lambda_{\gamma,\,K}\,\vv{z})$ for some $K$ and $\vv{z} = \fv(K)$.  Hence, $s = (\Lambda_{\gamma,\,K}\,\vv{z})$ with $\vv{z} = \fv(K)$.  Thus $A[\vv{t/z}] = \exists_\gamma\,(\Lambda_{\gamma,\,K}\,\vv{t})$ and, since we know this is provable in the context $\exfD$, it follows from the proof of from Lemma~\ref{lem:aux_ex_generates_all_instances} (although we did not state the lemma this way) that there is a closed term $u$ and there is actually a subproof concluding $\exfD \vdash (\Lambda_{\gamma,\,K}\,\vv{z})\,u$.  This subproof must be concluded using the rule \ruleName{Res} with premise $\exfD \vdash \exfMap{K}[\vv{t/z}][u/x]$.  Since this is a strict subproof of the original, we can apply the induction hypothesis to obtain $D \vdash K[\vv{t/z}][u/x]$.  Then it follows from \ruleName{Ex} that $D \vdash (\exists x:\gamma.\ K)[\vv{t/z}]$, as required.
        \item Otherwise the predicate is some $P$ from $\Pi$ (the old signature).  In this case, the side premise comes from $\exfMap{D}$ and $G$ of the side premise must be of shape $\exfMap{G'}$ for some $G'$ from the original clauses with $\fv(G') = \vv{y}$.  Then, the induction hypothesis applies to $\exfD \vdash \exfMap{G'}[\vv{s/y}]$, so we can deduce $D \vdash G'[\vv{s/y}]$.  Then, by \ruleName{Res}, we obtain $A[\vv{s/y}]$.  Now, suppose that $A[\vv{s/y}] = \exfMap{H'}[\vv{t/z}]$ as required by the statement.  Since $A$ is a head with predicate $P$, it can only be that $\exfMap{H'} = H'$ and so what we obtained is already in the correct form.
      \end{itemize}
  \end{description}







\end{proof}

\section{Supporting Material for Section~\ref{sec:rewriting}}
\label{appx:rewriting}

\begin{lemma}[Reflexivity of implication]\label{lem:imp-refl}
  $V,\,\yy \vdash (\forall \vv{z}.\, U \implies U) \rew^* \truetm$
\end{lemma}

\begin{lemma}[Progress]\label{lem:progress}
  If $V,\,\zz \vdash G[\uu/\yy] \rew H$ then there is $K$ and $U$ such that: 
  \[
    \textit{(a)} \quad V,\,\yy\zz \vdash G \rew^* K \wedge U
  \]
  and we can split $H$ as $H_1 \wedge H_2$ so that both the following hold:
  \[
    \textit{(b)} \quad V,\,\zz \vdash U[\uu/\yy] \rew^* H_1 \qquad\qquad \textit{(c)} \quad K[\uu/\yy] = H_2
  \]
  and, moreover, $\fv(U) \subseteq \yy$.
\end{lemma}
\begin{proof}
  By induction on the rewrite $V,\,\zz \vdash G[\uu/\yy] \rew H$.
  \begin{description}
    \item[\ruleName{Refl}] In this case, $H$ is $\truetm$ and $G[\uu/\yy]$ is some $P\,x \in V$.  Take $K = H_1 = H_2 = \truetm$ and (c) follows immediately.  It must be that $G = P\,v$ for some variable $v$.  We distinguish cases on $v$: 
      \begin{itemize}
        \item If $v \in \yy$ then take $U$ to be $P\ v$.  Then (a) follows in 0 rewrite steps and (b) follows by \ruleName{Refl} because $U[\uu/\yy] = P\ x$.  
        \item Otherwise, take $U$ to be $\truetm$.  Then (a) follows by \ruleName{Refl} and (b) is immediate.
      \end{itemize}
  \item[\ruleName{Step}] In this case, $G[\uu/\yy] = P\ (f\ \ss)$, $H = U'[\ss/\xx]$ and there is a clause $(\forall \xx.\, U' \implies P\ (f\ \xx)) \in V$.  In general, $G$ has shape $P\,s$ and so we distinguish cases on $s$:
  \begin{itemize} 
    \item If the head symbol of $s$ is a variable $y \in \yy$, it must be that $s$ is of shape $y\ \vv{v}$, $y[\uu/\yy] = f\,\vv{s_1}$, $\vv{v}[\uu/\yy] = \vv{s_2}$ and $\vv{s_1s_2} = \ss$ and we may split $\xx$ as $\vv{x_1}\vv{x_2}$ (i.e. $|\vv{v}| = |\vv{s_2}| = |\vv{x_2}|$).  
    First observe that $H = U'[\vv{s}/\vv{x}]$ can also be written as $\res{U'}{\vv{x_1}}[\vv{s}/\vv{x_1}] \wedge \res{U'}{\vv{x_2}}[\vv{v}/\vv{x_2}][\vv{u}/\vv{y}]$.  Take $H_1$ as the former and $H_2$ as the latter conjunct.  
    Take $K$ to be $\res{U'}{\vv{x_2}}[\vv{v}/\vv{x_2}]$ and $U$ to be $\forall \vv{x_2}.\, \res{U'}{\vv{x_2}} \implies P\,(y\ \vv{x_2})$.  Thus (c) holds immediately.
    Then (a) follows by \ruleName{Assm}:
    \[
      V,\,\yy\zz \types P(y\ \vv{v}) \rew \res{U'}{\vv{x_2}}[\vv{v}/\vv{x_2}] \wedge (\forall \vv{x_2}.\, \res{U'}{\vv{x_2}} \implies P\,(y\ \vv{x_2}))
    \]
    For (b), observe that $U[\uu/\yy]$ is $\forall \vv{x_2}.\, \res{U'}{\vv{x_2}} \implies P\ (f\,\vv{s_1}\vv{x_2})$.  By \ruleName{Imp} and \ruleName{Step}, \ruleName{ImpAnd}, and \ruleName{Scope} respectively, we have:  
    \begin{align*}
      V, \zz \vdash \forall \vv{x_2}.\, \res{U'}{\vv{x_2}} \implies P\ (f\,\vv{s_1}\vv{x_2}) 
      &\rew \forall \vv{x_2}.\, \res{U'}{\vv{x_2}} \implies U'[\vv{s_1}\vv{x_2}/\vv{x}] \\
      &\rew^* (\forall \vv{x_2}.\, \res{U'}{\vv{x_2}} \implies \res{U'}{\vv{x_1}}[\vv{s_1}/\vv{x_1}]) \wedge (\forall \vv{x_2}.\, \res{U'}{\vv{x_2}} \implies \res{U'}{\vv{x_2}}) \\
      &\rew^* \res{U'}{\vv{x_1}}[\vv{s_1}/\vv{x_1}] \wedge (\forall \vv{x_2}.\, \res{U'}{\vv{x_2}} \implies \res{U'}{\vv{x_2}})
    \end{align*}
    with the latter conjunct rewriting to $\truetm$ by the reflexivity of implication, leaving $H_1$.

    \item Otherwise, $s$ is of shape $f\ \ww$ and $\ww[\uu/\yy] = \ss$.  Then take $K \defeq U'[\ww/\xx]$, $H_2 \defeq H$, and $U \defeq \truetm$ and $H_1 \defeq \truetm$.  Now (a) follows by \ruleName{Step}, (b) is trivial, and (c) from $K[\uu/\yy] = U'[\ww/\xx][\uu/\yy] = H_2$.
  \end{itemize}

  \item[\ruleName{Assm}] In this case, $G[\uu/\yy]$ has shape $P\,(z\,\ss)$ and $H$ has shape $U'[\ss/\xx] \land (\forall \xx.\, U' \Rightarrow P\,(z\,\xx))$ with $z \in \zz$.  In general, $G$ has shape $P\,s$ and so we distinguish cases on $s$:
  \begin{itemize} 
    \item This case is similar to the corresponding \ruleName{Step} case.  If the head symbol of $s$ is a variable $y \in \yy$, it must be that $s$ is of shape $y\ \vv{v}$, $y[\uu/\yy] = z\,\vv{s_1}$, $\vv{v}[\uu/\yy] = \vv{s_2}$ and $\vv{s_1s_2} = \ss$ and we may split $\xx$ as $\vv{x_1}\vv{x_2}$ (i.e. $|\vv{v}| = |\vv{s_2}| = |\vv{x_2}|$).  
    First observe that $H = U'[\vv{s}/\vv{x}] \land (\forall \xx.\, U' \Rightarrow P\,(z\,\xx))$ can also be written as $\res{U'}{\vv{x_1}}[\vv{s}/\vv{x_1}] \land (\forall \xx.\, U' \Rightarrow P\,(z\,\xx)) \land \res{U'}{\vv{x_2}}[\vv{v}/\vv{x_2}][\vv{u}/\vv{y}]$.  Take $H_1$ as the first two conjuncts and $H_2$ as the third conjunct.  
    Take $K$ to be $\res{U'}{\vv{x_2}}[\vv{v}/\vv{x_2}]$ and $U$ to be $\forall \vv{x_2}.\, \res{U'}{\vv{x_2}} \implies P\,(y\ \vv{x_2})$.  Thus (c) holds immediately.
    Then (a) follows by \ruleName{Assm}:
    \[
      V,\,\yy\zz \types P(y\ \vv{v}) \rew \res{U'}{\vv{x_2}}[\vv{v}/\vv{x_2}] \wedge (\forall \vv{x_2}.\, \res{U'}{\vv{x_2}} \implies P\,(y\ \vv{x_2}))
    \]
    For (b), observe that $U[\uu/\yy]$ is $\forall \vv{x_2}.\, \res{U'}{\vv{x_2}} \implies P\ (z\,\vv{s_1}\vv{x_2})$.  This rewrites with \ruleName{Assm} (with \ruleName{Imp}) followed by \ruleName{ImpAnd}:  
    \begin{align*}
      V, \zz \vdash\, &\forall \vv{x_2}.\, \res{U'}{\vv{x_2}} \implies P\ (z\,\vv{s_1}\vv{x_2}) \\
      &\rew \forall \vv{x_2}.\, \res{U'}{\vv{x_2}} \implies U'[\vv{s_1}\vv{x_2}/\vv{x}] \land (\forall \vv{x}.\, U' \Rightarrow P\,(z\,\vv{x})) \\
      &\rew (\forall \vv{x_2}.\, \res{U'}{\vv{x_2}} \implies U'[\vv{s_1}\vv{x_2}/\vv{x}]) \land (\forall \vv{x_2}.\, \res{U'}{\vv{x_2}} \implies (\forall \vv{x}.\, U' \Rightarrow P\,(z\,\vv{x})))
    \end{align*}
    The right conjunct rewrites to $\forall \vv{x'}.\, U' \Rightarrow P\,(z\,\vv{x'})$ by \ruleName{Scope}.  Note $U'[\vv{s_1}\vv{x_2}/\vv{x}]$ is $\res{U'}{\vv{x_1}}[\vv{s_1}/\vv{x_1}] \land \res{U'}{\vv{x_2}}$, so the left conjunct rewrites with \ruleName{ImpAnd} and \ruleName{Scope}:
    \begin{align*}
      V, \zz \vdash\, &\forall \vv{x_2}.\, \res{U'}{\vv{x_2}} \implies \res{U'}{\vv{x_1}}[\vv{s_1}/\vv{x_1}] \land \res{U'}{\vv{x_2}} \\
      &\rew (\forall \vv{x_2}.\, \res{U'}{\vv{x_2}} \implies \res{U'}{\vv{x_1}}[\vv{s_1}/\vv{x_1}]) \land (\forall \vv{x_2}.\, \res{U'}{\vv{x_2}} \implies \res{U'}{\vv{x_2}})  \\
      &\rew \res{U'}{\vv{x_1}}[\vv{s_1}/\vv{x_1}] \land (\forall \vv{x_2}.\, \res{U'}{\vv{x_2}} \implies \res{U'}{\vv{x_2}})
    \end{align*}
    with the latter conjunct rewriting to $\truetm$ by the reflexivity of implication.
    This gives us a rewrite of $U[\uu/\yy]$ to $H_1$, proving (b):
    \[
      V, \zz \vdash (\forall \vv{x_2}.\, \res{U'}{\vv{x_2}} \implies P\ (z\,\vv{s_1}\vv{x_2}))
        \rew^* \res{U'}{\vv{x_1}}[\vv{s_1}/\vv{x_1}] \land (\forall \vv{x}.\, U' \Rightarrow P\,(z\,\vv{x}))
    \]
    \item Otherwise, $s$ is of shape $z\ \ww$ and $\ww[\uu/\yy] = \ss$.  Then take $K \defeq U'[\ww/\xx]$, $H_2 \defeq U'[\ss/\xx]$, and $U \defeq (\forall \vv{x}.\, U' \Rightarrow P\,(z\,\vv{x}))$ and $H_1 \defeq (\forall \vv{x}.\, U' \Rightarrow P\,(z\,\vv{x}))$.  Now (a) follows by \ruleName{Assm}, (b) is trivial, and (c) from $K[\uu/\yy] = U'[\ww/\xx][\uu/\yy] = H_2$.
  \end{itemize}

  \item[\ruleName{AndL}, \ruleName{AndR}] Since both are similar, we only show the former.  In this case $G[\uu/\yy]$ has shape $G_1[\uu/\yy] \wedge G_2[\uu/\yy]$, $H$ has shape $H_1 \wedge G_2[\uu/\yy]$ and we have that $V, \zz \vdash G_1[\uu/\yy] \rew H_1$.  It follows from the induction hypothesis that there are $K'$, $U'$ and we may split $H_1$ as $H_1'$ and $H_2'$ such that the required properties hold.  Then take $U := U'$, $H_1 := H_1'$, $K := K' \wedge G_2$ and $H_2 := H_2' \wedge G_2[\uu/\yy]$.  Then (a) holds because $G_1 \wedge G_2 \rew^* U \wedge K' \wedge G_2$ and both (b) and (c) follow straightforwardly from the induction hypothesis.

  \item[\ruleName{Imp}] In this case, $G[\uu/\yy]$ has shape $\forall \xx.\, W \Rightarrow G'[\uu/\yy]$ and $H$ has shape $\forall \xx.\, W \Rightarrow H'[\uu/\yy]$ and we have $V, \zz \vdash G'[\uu/\yy] \rew H'$. Note that, as \( W \) is an automaton body it only contain \( \xx \). It follows from the induction hypothesis that there are $K'$ and $U'$ and that we can split $H'$ as $H_1' \wedge H_2'$, such that the required properties hold.  Then take $U := U'$, $H_1 := H_1'$, $K := \forall \zz.\, W \Rightarrow K'$ and $H_2 := \forall \zz.\, W \Rightarrow  H_2'$.  Then (a) follows because $U$ can only contain the $\yy$ as free variables, we can apply scoping to deduce:
  \[
    V, \yy\zz \vdash \forall \xx.\, W \Rightarrow G' \rew \forall \xx.\, W \Rightarrow K' \wedge U' \rew (\forall \xx.\, W \Rightarrow K') \wedge U'
  \] 
  and properties (b) and (c) follow from the induction hypothesis.

  \item[\ruleName{ImpAnd}] In this case, $G[\uu/\yy]$ has shape $\forall \vv{x}.\, U' \implies G_2[\uu/\yy] \wedge G_3[\uu/\yy]$.  Set $U := \truetm$, $H_1 := \truetm$, $K := (\forall \vv{x}.\, U' \implies G_2) \wedge (\forall \vv{x}.\, U' \implies G_3)$ and $H_2 := H$.  Then (a) follows by \ruleName{ImpAnd}, (b) and (c) are immediate.
  
  \item[\ruleName{Scope}] In this case, $G[\uu/\yy]$ has shape $\forall \xx.\, U' \implies G[\uu/\yy]$ and $H$ has shape $G[\uu/\yy]$.  We assume $G[\uu/\yy]$ does not contain any variable from $\xx$.  Hence, neither does $G$ and so we can rewrite:
  \[
    V, \yy\zz \vdash \forall \xx.\, U' \implies G \rew G
  \] 
  Hence, the result follows by taking $U \defeq \truetm$, $H_1 \defeq \truetm$, $K \defeq G$ and $H_2 \defeq G[\uu/\yy]$.
  \end{description}
\end{proof}


\subsection{Completeness}

To show that $D \models G$ implies $\apprx(D) \models G$, we will argue that one can use a proof of $D \proves G$ as a guide to follow in order to obtain exactly the automaton clauses $V \subseteq \apprx(D)$ needed to replay the proof as $V \proves G \rew \truetm$.
Clearly, the main aim of our argument will be to show that, wherever we made a resolution step using a clause from $D$ in the proof of $D \proves G$, we can find an automaton clause in $\apprx(D)$ with which to simulate it when rewriting.

To see why, consider again the proof in Figure~\ref{fig:file-proof}.
We can observe that the resolution steps at the leaves of the proof necessarily involve an automaton clause as the side premise: it cannot be otherwise since a resolution step at a leaf must be via a clause that has a vacuous body, such as $\truetm \implies \IsClosed(\closedhdl)$, and clauses of this form are necessarily automaton.

Then, there may be many consecutive resolution steps leading up to this leaf that all happen to involve clauses from $D$ that are automaton already.
We observe that such a situation gives rise to a maximal subproof, say of $D \types H$ which can be viewed equivalently as a proof $\apprx(D) \proves H$, since it uses only clauses that are automaton (and hence in $\apprx(D)$).
In the case of Figure~\ref{fig:file-proof}, the subproof rooted at \changed[jj]{$\Cex(\readS\ \closedhdl\ (\putContS\ \idS))$ is of this form, since (18) and (7) both happen to be automaton.}

If the subproof $D \types H$ is the largest that contains only automaton inferences, then continuing down the proof tree, the next resolution step involves a clause from $D$ that is not itself automaton.
Let us say that this inference is of shape:
\[
  \prftree[r,l]{$(\forall \yy.\ G' \implies A) \in D$}{\ruleName{Res}}{
  D \proves G'[\vv{s}/\yy]
  }
  {
  D \proves A[\vv{s}/\yy]
  }
\]
and, for simplicity, that $H = G'[\vv{s}/\yy]$.
The fact that $\forall \yy.\ G' \implies A$ is not automaton indicates that $G'$ is not of the correct form.
However, the subproof $\apprx(D) \proves H$ that sits on top of this inference shows us how to prove $G'[\vv{s}/\yy]$ using only resolution with automaton clauses.
We can view this proof as a specification for rewriting $G'[\vv{s}/\yy]$ to $\truetm$, we just need to perform the same resolution steps in the same order.
The following lemma shows that performing the same rewriting sequence on the raw goal $G'$ (which is the body that is preventing the \ruleName{Res} inference from being automaton) will give rise to an automaton formula.





\begin{restatable}[Decomposition]{lemma}{decomp}
  If $V \vdash G[\vv{u}/\vv{y}] \rew^* \truetm$ then there is some $U$ such that:
  \begin{enumerate}[label=(\alph*)]
    \item $V,\,\yy \vdash G \rew^* U$
    \item and $V \vdash U[\vv{u}/\vv{y}] \rew^* \truetm$
  \end{enumerate}
\end{restatable}
\begin{proof}
  By induction on the length $n$ of the rewrite sequence $V \vdash G[\vv{u}/\vv{y}] \rew^n \truetm$.

  \begin{itemize}
  \item When $n=1$, $G = \truetm$ and so we can take $U \defeq \truetm$ too.

  \item Otherwise, the sequence has shape $V \vdash G[\uu/\yy] \rew H \rew^{n-1} \truetm$ and we can invoke progress (Lemma~\ref{lem:progress}) to obtain some $U'$ and $K$ such that (a) $V \vdash G \rew^* U' \wedge K$ and we can split $H$ as $H_1 \wedge H_2$ with $V \vdash U'[\uu/\yy] \rew^* H_1$ and $K[\uu/\yy] = H_2$.  Since $V \vdash H \rew^{n-1} \truetm$, there are sequences $V \vdash H_1 \rew^* \truetm$ and $V \vdash H_2 \rew^* \truetm$ of length at most $n-1$.  So, by the induction hypothesis, we obtain some $U''$ such that $V \vdash K \rew^* U''$ and $V \vdash U''[\uu/\yy] \rew^* \truetm$.  Hence, we choose $U \defeq U' \wedge U''$.
  \end{itemize}

\end{proof}

For example, in Figure~\ref{fig:file-proof}, although the resolution step rooted at \changed[jj]{$\Cex(\putStrS\ \readS\ \closedhdl\ \idS)$ uses clause (10), namely $\forall x\,h\,k.\ \Cex(x\ h\ (\putContS \ k))  \implies \Cex(\putStrS\ x\ h\ k)$, the proof above this point shows us that resolving with automaton clauses (7) and (18) allows for rewriting $\Cex(\readS\ \closedhdl\ (\putContS\ \idS))$ to $\truetm$:
\[
  \Cex(\readS\ \closedhdl\ (\putContS\ \idS)) \rew \IsClosed(\closedhdl) \rew \truetm
\]
Using \ruleName{Assm}, we can simulate these same resolution steps on the raw body of clause (10) to obtain a new automaton formula:
\[
  \Cex(x\ h\ (\putContS \ k)) \rew \IsClosed(h) \wedge (\forall h'\,k'.\ \IsClosed(h') \implies \Cex(x\ h'\ k'))
\]
Here, we are adding the assumption that $x$ behaves like $\readS$ does, according to clause (7).
This then justifies adding a new automaton clause to $\apprx(D)$ corresponding to (10):
\[
  \forall x\,h\,k.\ \IsClosed(h) \wedge (\forall h'\,k'.\ \IsClosed(h') \implies \Cex(x\ h'\ k')) \implies \Cex(\putStrS\ x\ h\ k)
\]
This new clause allows us to extend our subproof now one further deduction, to:
\[
  D \proves \Cex(\putStrS\ \readS\ \closedhdl\ \idS)
\]
This reflects that $\apprx(D) \proves \Cex(\putStrS\ \readS\ \closedhdl\ \idS) \rew^* \truetm$.
By repeating this process we can replay the whole deduction using only automaton clauses from $\apprx(D)$.  Hence:}

\rewcomp*
\begin{proof}
  We show $D \vdash G$ implies $\apprx(D) \vdash G \rew^* \truetm$ by induction on the proof $D \vdash G$.  The result then follows from soundness (Theorem \ref{thm:soundness-rewriting}).
  \begin{description}
    \item[\ruleName{True}] The result is immediate.
    \item[\ruleName{Res}] In this case, $G$ is of shape $P\ s[\uu/\yy]$.  The premise is $D \vdash H[\uu/\yy]$ for some $(\forall \yy.\, H \implies P\ s) \in D$.  It follows from the induction hypothesis that $\apprx(D) \vdash H[\uu/\yy] \rew^* \truetm$.  Hence, by decomposition, there is some $U$ such that $\apprx(D),\,\yy \vdash H \rew^* U$ and $\apprx(D) \vdash U[\uu/\yy] \rew^* \truetm$. Then, by definition, $(\forall \yy.\, U \implies P\ s) \in \apprx(D)$ and so we can use \ruleName{Step} to rewrite $\apprx(D) \vdash P\ s[\uu/\yy] \rew U[\uu/\yy] \rew^* \truetm$, as required.
    \item[\ruleName{And}] In this case, $G$ is of shape $H_1 \wedge H_2$ and it follows from the induction hypothesis that $\apprx(D) \vdash H_1 \rew^* \truetm$ and $\apprx(D) \vdash H_2 \rew^* \truetm$.  Hence, the result follows straightforwardly.
  \end{description}
\end{proof}

\subsection{Decidability}
\begin{theorem}[Decidability]
  Let \( V \) be an automaton formula, \( \overline{y} \) is a sequence of variables, \( G \) an arbitary goal formula and \( U \) an automaton formula.
  Then \( V,\, \overline{y} \vdash G \rew^* U \) is decidable.
\end{theorem}
\begin{proof}
  \begin{definition}
    The \emph{weight} of a clause is defined as follows:

    \begin{align*}
      \mathsf{weight}(\textsf{true})                                                                   & \coloneqq 0                                                                                                        \\
      \mathsf{weight}(P\ t)                                                                            & \coloneqq \depth{t}                                                                                                \\
      \mathsf{weight}(G \wedge H)                                                                      & \coloneqq \mathsf{weight}(G) + \mathsf{weight}(H)                                                                  \\
      \mathsf{weight}(\forall \xx.\, U \Rightarrow \textsf{true}) = 1                                                                                                                                                       \\
      \mathsf{weight}(\forall \overline{x}.\, U \Rightarrow P\ t)                                      & \coloneqq \begin{cases}
                                                                                                                     \depth{t}     & \textrm{if}\ \fv(t) \cap \{ x \} \neq \emptyset \\
                                                                                                                     1 + \depth{t} & \textrm{otherwise}
                                                                                                                   \end{cases}                                          \\
      \mathsf{weight}(\forall \overline{x}.\, U \Rightarrow H_1 \wedge H_2)                            & \coloneq 1 + \mathsf{weight}(\forall \xx.\, U \Rightarrow H_1) + \mathsf{weight}(\forall \xx.\, U \Rightarrow H_2) \\
      \mathsf{weight}(\forall \overline{x}.\, U \Rightarrow (\forall \overline{y}.\, V \Rightarrow G)) & \coloneq 1 + \mathsf{weight}(\forall \overline{y}.\, V \Rightarrow G)
    \end{align*}
    When multiple cases apply due to the quotient structure of conjunction, we take the smallest weight.
  \end{definition}

  \begin{lemma}
    If \( \mathsf{measure}(G) < \mathsf{measure}(G') \), then \( \mathsf{measure}(\forall \xx. U \Rightarrow G) < \mathsf{measure}(\forall \xx. U \Rightarrow G') \).
  \end{lemma}

  \begin{lemma}
    If \( U \) is an automaton formula, then the weight of \( U[\ss/\xx] \) is at most \( \max \{ \depth{s} \mid s \in \ss \} \).
    \begin{proof}
      In the case of \( (P\ x)[\ss/\xx] \), we have that weight is bounded by \( \max \{ \depth{s} \mid s \in \ss \} \) as \( x \) may be instantiated with any \( \ss \).
      For \( (\forall \zz.\, U \Rightarrow P\ (x\ \zz))[\ss/xx] \), we have weight \( \depth{s} \leq \max \{ \depth s \mid \in \ss \} \) as \( U \) only contains the variables \( \zz \).
      Otherwise, \( U \) is a conjunction of both cases and we can proceed by structural induction.
    \end{proof}
  \end{lemma}

  \begin{lemma}
    Let \( V \) be an automaton formula, \( \overline{y} \) is a sequence of variables, and \( G \) and \( H \) formulas such that \( V,\, \yy \vdash G \rew H \).
    Then \( \mathsf{weight}(H) < \mathsf{weight}(G) \).
    \begin{proof}
      We proceed by induction on \( V,\, \yy \vdash G \rew H \):
      \begin{description}
        \item[\ruleName{Refl}] Trivial.
        \item[\ruleName{Step}]
          This case \( G \) is the atom \( P\ (f\ \ss) \) and has depth \( 1 + \max \{ \depth s \mid s \in \ss \} \).
          By the previous lemma, the weight of \( U[\ss/\xx] \) is at most \( \max \{ \depth s \mid s \in \ss \} \) as it is an automaton body.
          Therefore, we have a proper decrease as required.
        \item[\ruleName{Assm}]
          In this case, \( G \) is the atom \( P\ (x\ ss) \) which has depth \( 1 + \max \{ \depth s \mid s \in \ss \} \) and the reduct is \( U[\ss/\zz] \wedge (\forall \zz.\, U \Rightarrow P\ (x\ \zz)) \) which has depth \( \max \{ \depth{s} \mid s \in \ss \} \), again by the previous lemma.
          Thus this case is also satisfied.
        \item[\ruleName{AndL}, \ruleName{AndR}] Trivial.
        \item[\ruleName{Imp}]
          In this case, the goal is \( \forall \xx.\, U \Rightarrow G \) where \( G \) is not of the form \( P\ (y\ \xx) \).
          By induction on \( V \wedge U,\, \yy \vdash G \rew G' \), we have that \( \mathsf{weight}(G') < \mathsf{weight}(G) \).
          Therefore, \( \textsf{measure}(\forall \xx.\, U \Rightarrow G) < \textsf{measure}(\forall \xx.\, U \Rightarrow G') \) as required.
        \item[\ruleName{ImpAnd}] Trivial.
        \item[\ruleName{Scope}]
          In this case, the goal is \( \forall \xx.\, U \Rightarrow G \) where \( \fv{G} \cap \xx = \emptyset \).
          Consider inductively the possible shapes of \( G \):
          \begin{itemize}
            \item \( \mathsf{true} \) --- the measure is \( 1 \) and after rewritting is \( 0 \).
            \item \( P\ t \) --- the measure is \( 1 + \depth{t} \) and after rewriting is \( 1 \).
            \item \( G_1 \wedge G_2 \) --- the measure is \( 1 + \textsf{measure}(\forall \xx.\, U \Rightarrow G_1) + \textsf{measure}(\forall \xx.\, U \Rightarrow G_2) > 1 + \textsf{measure}(G_1) + \textsf{measure}(G_2) \) and after rewritting it is just \( \textsf{measure}(G_1) + \textsf{measure}(G_2) \).
            \item Finally, for \( \forall \yy.\, U' \Rightarrow G \), the measure is \( 1 + \textsf{measure}(\forall \yy.\, U' \Rightarrow G) \) and after rewriting it is \( \textsf{measure}(\forall \yy.\, U' \Rightarrow G) \).
          \end{itemize}
          In each case, we have a strict decrease as required.
      \end{description}
    \end{proof}
  \end{lemma}

  \begin{corollary}
    For any choice of \( V \), \( \overline{y} \) and \( G \), there are finitely many reducts \( G' \) such that \( V,\, \yy \vdash G \rew G' \).
  \end{corollary}
\end{proof}
\section{Supporting Material for Section~\ref{sec:automaton_formulas}}
\label{appx:intersection_types}

\subsection{Types}

\paragraph{Subtyping rules.} 
\begin{description}
  \item [\ruleName{Q-Bas}] if $(q_1,q_2) \in \preorder$, then $q_1 \leq q_2$
  \item [\ruleName{Q-Arr}] if $\sigmaInt_2 \leq \sigmaInt_1$ and $\tau_1 \leq \tau_2$, then $\sigmaInt_1 \to \tau_1 \leq \sigmaInt_2 \to \tau_2$
  \item [\ruleName{Q-Fun}] if $\bigwedge_{i=1}^n \tau_i \leq \tau$, then $\bigwedge_{i=1}^n (\sigmaInt \to \tau_i) \leq \sigmaInt \to \tau$
  \item [\ruleName{Q-Prj}] for all $1 \leq i \leq n$, $\bigwedge_{j=1}^n \tau_j \leq \tau_i$
  \item [\ruleName{Q-Glb}] if, for all $1 \leq i \leq n$, $\sigmaInt \leq \tau_i$, then $\sigmaInt \leq \bigwedge_{j=1}^n \tau_j$
  \item [\ruleName{Q-Trs}] if $\theta_1 \leq \theta_2$ and $\theta_2 \leq \theta_3$, then $\theta_1 \leq \theta_3$
\end{description}

\begin{lemma}
  \label{lem:subtyping}
  \[
    \textstyle \bigwedge_{i=1}^n \tau_i \leq \bigwedge_{j=1}^m \tau_j'
      \quad \iff \quad
    \forall j \in [1..m].\,\exists i_j \in [1..n].\, \tau_{i_j} \leq \tau_j'
  \]
\end{lemma}
\begin{proof}
  Direction $\Leftarrow$ follows from \ruleName{Q-Glb} followed by \ruleName{Q-Trs} and \ruleName{Q-Prj}.
  
  For $\Rightarrow$, suppose $ \bigwedge_{i=1}^n \tau_i \leq \bigwedge_{j=1}^m \tau_j'$.
  We use induction on the subtyping rules, of which \ruleName{Q-Bas}, \ruleName{Q-Prj}, and \ruleName{Q-Arr} are immediate.
  \begin{description}
    
    \item[\ruleName{Q-Fun}] Here, $j=1$ and the assumption is some $ \bigwedge_{i=1}^n (\sigmaInt \to \tau_i) \leq (\sigmaInt \to \tau')$.
    The premise is $\bigwedge_{i=1}^n \tau_i \leq \tau'$.
    By the IH, there exists $i \in [1..n]$ such that $\tau_i \leq \tau'$.
    It follows from \ruleName{Q-Arr} that $\sigmaInt \to \tau_i \leq \sigmaInt \to \tau'$, which proves the claim.
    
    \item[\ruleName{Q-Glb}] The premises are $\bigwedge_{i=1}^n \tau_i \leq \tau_j'$, for all $j \in [1..m]$.
      By the IH, there exists $i_j \in [1..n]$ such that $\tau_{i_j} \leq \tau_j'$, for all $j \in [1..m]$, as required.
    
    \item[\ruleName{Q-Trs}] The premises give us $\theta$ such that $ \bigwedge_{i=1}^n \tau_i \leq \theta$ and $\theta \leq \bigwedge_{j=1}^m \tau_j'$.
    Let $j \in [1..m]$.
      By the IH, there exists $\tau''_k \in \theta$ such that $\tau''_k \leq \tau'_j$.
      Furthermore, there exists $i_k \in [1..n]$ such that $\tau_{i_k} \leq \tau''_k$.
      By transitivity, $\tau_{i_k} \leq \tau'_j$, as required.
  \end{description}
\end{proof}

\begin{lemma}[Weakening]
  \label{lem:weakening}
  If $\Gamma \vdash \hasType{t}{\tau}$ and $\zz \cap \fv(t) = \emptyset$, then $\Gamma, \hasType{\zz}{\vv{\sigmaInt}} \vdash \hasType{t}{\tau}$ for all $\vv{\sigmaInt}$ of the appropriate types.
  Idem for $\hasType{G}{\boolType}$.
\end{lemma}

\subsection{Typing-rewriting correspondence}
\label{appx:types_typing-rewriting}

\begin{lemma}
  \label{lem:rew_conj}
  If $V, \yy \vdash (\forall \zz.\, U \Rightarrow G_1 \land G_2) \rew^* U'$, then there exist $U_1',U_2'$ such that $V, \yy \vdash (\forall \zz.\, U \Rightarrow G_1) \rew^* U_1'$ and $V, \yy \vdash (\forall \zz.\, U \Rightarrow G_2) \rew^* U_2'$ and $U' = U_1' \land U_2'$.
\end{lemma}

\begin{lemma}
  \label{lem:true_automaton}
  For automaton formulas $U$ and $U'$ over $\zz$,  $V, \yy \vdash (\forall \zz.\, U \Rightarrow U') \rew^* \truetm$ if and only if $U \leqa U'$. 
\end{lemma}
\begin{proof}
  We prove the claim by induction on automaton clause $U'$:
  an outer induction on the types of the variables in $U'$ and an inner induction on the number of conjuncts.
  First, note that both $V, \yy \vdash (\forall \zz.\, U \Rightarrow U') \rew^* \truetm$ and $U \leqa U'$ hold for $U' = \truetm$.
  
  Proof of direction $\Rightarrow$:
      \begin{itemize}
      
         \item Suppose $V, \yy \vdash (\forall \zz.\, U \Rightarrow P\,z_i) \rew^* \truetm$.
           This rewrite must start with \ruleName{Imp} with \ruleName{Refl}, which has $P\,z_i \in U$ as side condition.
           It is immediate that $\fromClause{\res{U}{z_i}} \leq \fromClause{\res{{P\,z_i}}{z_i}}$ and $U \leqa P\,z_i$ follows.
           
         \item Suppose $V, \yy \vdash (\forall \zz.\, U \Rightarrow U_1' \land U_2') \rew^* \truetm$ for $U_1',U_2'$ containing only $z_i$ of type $\iota$ (and at least one such $z_i$ each).
           By Lemma~\ref{lem:rew_conj}, we may assume WLOG that this rewrite starts with \ruleName{ImpAnd}, i.e.~$V, \yy \vdash (\forall \zz.\, U \Rightarrow U_1' \land U_2') \rew (\forall \zz.\, U \Rightarrow U_1') \land (\forall \zz.\, U \Rightarrow U_2') \rew^* \truetm$.
           Thanks to \ruleName{AndL} and \ruleName{AndR}, $V, \yy \vdash (\forall \zz.\, U \Rightarrow U_1') \rew^* \truetm$ and $V, \yy \vdash (\forall \zz.\, U \Rightarrow U_2') \rew^* \truetm$.
           We apply the IH for $U \leqa U_1'$ and $U \leqa U_2'$.
           Clearly, also $U \leqa U_1' \land U_2'$.
           
         \item Suppose $V, \yy \vdash (\forall \zz.\, U \Rightarrow (\forall \xx.\, U'' \Rightarrow P\,(z_i\,\xx))) \rew^* \truetm$ for some $\xx = x_1 \dots x_m$.
           \begin{itemize}
           
               
                              
             \item Suppose this rewrite starts with \ruleName{Imp}.
               The only option for the premise is another \ruleName{Imp} with $V \land U \land U'', \yy \vdash P\,(z_i\,\xx)) \rew U'''$ as premise.
               We proceed as in the previous subcase.
               
           \end{itemize}
           
         \item Suppose $V, \yy \vdash (\forall \zz.\, U \Rightarrow U_1' \land U_2') \rew^* \truetm$.
           Again, Lemma~\ref{lem:rew_conj} allows us to assume WLOG that this rewrite starts with \ruleName{ImpAnd}.
           We proceed as in the first-order case.
         
      \end{itemize}
  
  Proof of direction $\Leftarrow$:
  \begin{itemize}
  
    \item Suppose $U \leqa P\,z_i$.
      This means that $\fromClause{\res{U}{z_i}} \leq q_P$ and, thus, $q_P \in \fromClause{\res{U}{z_i}}$ and $P\,z_i \in U$.
      We prove $V, \yy \vdash (\forall \zz.\, U \Rightarrow P\,z_i) \rew^* \truetm$ by \ruleName{Imp} with \ruleName{Refl}, followed by \ruleName{Scope}. 
      
    \item Suppose $U \leqa U_1' \land U_2'$ containing only $z_i$ of type $\iota$ (and at least one such $z_i$ each).
      This means that $\fromClause{\res{U}{z_i}} \leq \fromClause{\res{(U_1' \land U_2')}{z_i}}$, for all $i \in [1..n]$.
      Clearly, also $\fromClause{\res{U}{z_i}} \leq \fromClause{\res{{U_1}'}{z_i}}$ and $\fromClause{\res{U}{z_i}} \leq \fromClause{\res{{U_2}'}{z_i}}$, for all $i \in [1..n]$.
      We now apply the IH to $U \leqa U_1'$ and $U \leqa U_2'$ for $V, \yy \vdash (\forall \zz.\, U \Rightarrow U_1') \rew^* \truetm$ and $V, \yy \vdash (\forall \zz.\, U \Rightarrow U_2') \rew^* \truetm$.
      This gives us a rewrite $V, \yy \vdash (\forall \zz.\, U \Rightarrow U_1') \land (\forall \zz.\, U \Rightarrow U_2') \rew^* \truetm$ to which we prepend an \ruleName{ImpAnd}-step for $V, \yy \vdash (\forall \zz.\, U \Rightarrow U_1' \land U_2') \rew^* \truetm$. 
      
    \item Suppose $U \leqa (\forall \xx.\, U'' \Rightarrow P\,(z_i\,\xx))$ for some $\xx = x_1 \dots x_m$.
      This means $\fromClause{\res{U}{z_i}} \leqa \fromClause{\forall \xx.\, U'' \Rightarrow P\,(z_i\,\xx)}$.
      By Lemma~\ref{lem:subtyping}, there exists $(\forall \xx.\, U''' \Rightarrow P\,(z_i\,\xx)) \in \res{U}{z_i}$ such that $\fromClause{\forall \xx.\, U''' \Rightarrow P\,(z_i\,\xx)} \leqa \fromClause{\forall \xx.\, U'' \Rightarrow P\,(z_i\,\xx)}$. 
           
      This gives us $\fromClause{\res{U''}{x_j}} \leq \fromClause{\res{U'''}{x_j}}$, for all $j \in [1..m]$.
      This implies $U'' \leqa U'''$, to which we can apply the IH for $V, \yy \vdash (\forall \xx.\, U'' \Rightarrow U''') \rew^* \truetm$, because the types of $\xx$ are smaller than the type of $z_i$.
      Trivially, also $V \land U, \yy \vdash (\forall \xx.\, U'' \Rightarrow U''') \rew^* \truetm$.
      
      We prepend \ruleName{Step} with side condition $(\forall \xx.\, U''' \Rightarrow P\,(z_i\,\xx)) \in U$ to this rewrite to obtain $V \land U, \yy \vdash (\forall \xx.\, U'' \Rightarrow P\,(z_i\, \xx)) \rew (\forall \xx.\, U'' \Rightarrow U''') \rew^* \truetm$.
      We replicate this rewrite under an implication using \ruleName{Imp} and conclude it with \ruleName{Scope}:
      \[
        V, \yy \vdash
          (\forall \zz.\, U \Rightarrow (\forall \xx.\, U'' \Rightarrow P\,(z_i\, \xx)) )
          \rew^* (\forall \zz.\, U \Rightarrow \truetm)
          \rew \truetm
      \]
      This proves the case.
      
    \item Suppose $U \leqa U_1' \land U_2'$.
      As the inductive first-order case.

  \end{itemize}
\end{proof}

\begin{lemma}[Monotonicity of rewrite]
  \label{lem:monotonicity_of_rewrite}
  For all terms $t: \sigma_1 \to \dots \to \sigma_k \to \iota$ over $\yy$, if $\sigmaInt \leq \sigmaInt'$ and $V, \yy \vdash \toClause{\sigmaInt}{t} \rew^* U$, then there exists $U'$ such that $V, \yy \vdash \toClause{\sigmaInt'}{t} \rew^* U'$ and $U \leqa U'$.
\end{lemma}
\begin{proof}
  Let $\sigmaInt = \bigwedge_{i=1}^n \tau_i$ and $\sigmaInt' = \bigwedge_{j=1}^m \tau_j'$ such that $\sigmaInt \leq \sigmaInt'$ and $V, \yy \vdash \toClause{\sigmaInt}{t} \rew^* U$.
  By Subtyping Lemma~\ref{lem:subtyping}, $\forall j \in [1..m].\, \exists i_j \in [1..n].\, \tau_{i_j} \leq \tau_j'$.
  Clearly, $\sigmaInt = \sigmaInt'' \land \bigwedge_{j=1}^m \tau_{i_j}$ for some $\sigmaInt''$ and, thus, $\toClause{\sigmaInt}{t} = \toClause{\sigmaInt''}{t} \land \toClause{\bigwedge_{j=1}^m \tau_{i_j}}{t}$.
  
  This allows us to derive separate rewrites $V, \yy \vdash \toClause{\sigmaInt''}{t} \rew^* U_1$ and $V, \yy \vdash \toClause{\bigwedge_{j=1}^m \tau_{i_j}}{t} \rew^* U_2$ such that $U = U_1 \land U_2$ from assumption $V, \yy \vdash \toClause{\sigmaInt}{t} \rew^* U$.
  Furthermore, $\toClause{\bigwedge_{j=1}^m \tau_{i_j}}{t} = \toClause{\tau_{i_1}}{t} \land \dots \land \toClause{\tau_{i_m}}{t}$ allows us to separate the latter rewrite into $V, \yy \vdash \toClause{\tau_{i_j}}{t} \rew^* U_{2,j}$, for all $j \in [1..m]$, such that $U_2 = U_{2,1} \land \dots \land U_{2,m}$.
  
  We aim to show that $V, \yy \vdash \toClause{\tau_{i_j}}{t} \rew^* U_{2,j}$ determines a $V, \yy \vdash \toClause{\tau_j'}{t} \rew^* U_{2,j}'$ with $U_{2,j} \leqa U_{2,j}'$, for all $j \in [1..m]$.
  This suffices due to $\toClause{\sigmaInt'}{t} = \bigwedge_{j=1}^m \toClause{\tau_j'}{t}$ and:
  \[
    U \leqa U_2 = \textstyle \bigwedge_{j=1}^m U_{2,j} \leqa \bigwedge_{j=1}^m U_{2,j}'
  \]
  We use case analysis on the type of $t$. 
  
  Suppose $t:\iota$.
  We only have trivial subtyping for type $\iota$, so $\forall j \in [1..m].\, \exists i_j \in [1..n].\, \tau_{i_j} = \tau_j'$.
  The claim is immediate. 
  
  Suppose $t: \sigma_1 \to \dots \to \sigma_k \to \iota$ for $k>0$.
  Fix a $j \in [1..m]$ with a corresponding $i_j \in [1..n]$.
  Let $\tau_{i_j} = \theta_1 \to \dots \to \theta_k \to q_P$ and $\tau_j' = \theta_1' \to \dots \to \theta_k' \to q_P$.
  We know from $\tau_{i_j} \leq \tau_j'$ that $\theta_\ell' \leq \theta_\ell$, for all $\ell \in [1..k]$. 
  
  Consider the rewrite $V, \yy \vdash \toClause{\tau_{i_j}}{t} \rew^* U_{2,j}$ we obtained from assumption $V, \yy \vdash \toClause{\sigmaInt}{t} \rew^* U$;
  since $\tau_{i_j} = \theta_1 \to \dots \to \theta_k \to q_P$, it is equal to:
  \begin{align}
    V, \yy \vdash (\forall \zz.\, \toClause{\vv{\theta}}{\zz} \Rightarrow P\,(t\,\zz)) \rew^* U_{2,j} \tag{\dag} \label{eq:rew_types}
  \end{align}
  We proceed by case analysis on the length of this rewrite~\ref{eq:rew_types}.
  \begin{itemize}
  
    \item Suppose rewrite~\ref{eq:rew_types} has length 0.
      Then, $U_{2,j} = (\forall \zz.\, \toClause{\vv{\theta}}{\zz} \Rightarrow P\,(t\,\zz))$ is automaton wrt $\yy$.
      Clearly, $(\forall \zz.\, \toClause{\vv{\theta'}}{\zz} \Rightarrow P\,(t\,\zz))$ is also automaton wrt $\yy$, and trivially:
      \[
        V, \yy \vdash (\forall \zz.\, \toClause{\vv{\theta'}}{\zz} \Rightarrow P\,(t\,\zz)) \rew^* (\forall \zz.\, \toClause{\vv{\theta'}}{\zz} \Rightarrow P\,(t\,\zz)) 
      \]
      The inequality $(\forall \zz.\, \toClause{\vv{\theta}}{\zz} \Rightarrow P\,(t\,\zz)) \leqa (\forall \zz.\, \toClause{\vv{\theta'}}{\zz} \Rightarrow P\,(t\,\zz))$ concludes this subcase, which follows from $\theta_\ell' \leq \theta_\ell$, for all $\ell \in [1..k]$. 
    
    \item Suppose rewrite~\ref{eq:rew_types} has length greater than 0.
      The rewrite starts with either \ruleName{Step} or \ruleName{Assm}, depending on the top-level symbol in $t$.
      The case for \ruleName{Assm} is analogous to \ruleName{Step}; 
      it merely introduces an extra conjunct that is already automaton (in which case both the rewrite of $\toClause{\tau_{i_j}}{t}$ and $\toClause{\tau'_{j}}{t}$ will introduce it; it must be a conjunct in $U_{2,j}$ and $U_{2,j}'$).
      Therefore, we assume $t = f\,\vv{t'}$ for some $f \in \Sigma$, so \ruleName{Step} with some $(\forall \vv{x_1}\,\vv{x_2}.\, U'' \Rightarrow P\,(f\,\vv{x_1}\,\vv{x_2})) \in V$ is the first rewrite step in~\ref{eq:rew_types}, resulting in: 
     \[
        V, \yy \vdash (\forall \zz.\, \toClause{\vv{\theta}}{\zz} \Rightarrow P\,(t\,\zz)) 
          \rew (\forall \zz.\, \toClause{\vv{\theta}}{\zz} \Rightarrow U''[\vv{t'}/\vv{x_1},\zz/\vv{x_2}])
          \rew^* U_{2,j} 
      \]
      By Lemma~\ref{lem:rew_conj}, $U_{2,j} = U_{2,j,1} \land U_{2,j,2}$ such that $V, \yy \vdash (\forall \zz.\,\toClause{\vv{\theta}}{\zz} \Rightarrow \res{U''}{\vv{x_1}}[\vv{t'}/\vv{x_1}]) \rew^* U_{2,j,1}$ and $V, \yy \vdash(\forall \zz.\,\toClause{\vv{\theta}}{\zz} \Rightarrow \res{U''}{\vv{x_2}}[\zz/\vv{x_2}]) \rew^* U_{2,j,2}$.
      Let us consider both parts individually.
      \begin{itemize}
  
        \item The goal formula $\res{U''}{\vv{x_1}}[\vv{t'}/\vv{x_1}]$ does not contain $\zz$, so we obtain $V, \yy \vdash (\forall \zz.\,\toClause{\vv{\theta'}}{\zz} \Rightarrow \res{U''}{\vv{x_1}}[\vv{t'}/\vv{x_1}]) \rew^* U_{2,j,1}$ immediately from the former rewrite.
    
        \item The goal formula $\res{U''}{\vv{x_2}}[\zz/\vv{x_2}]$ does not contain $\yy$, so $U_{2,j,2} = \truetm$.

          Inequality $\vv{\theta'} \leq \vv{\theta}$ gives us $\toClause{\vv{\theta'}}{\zz} \leqa \toClause{\vv{\theta}}{\zz}$.
          We use this in combination with Lemma~\ref{lem:true_automaton} to prove a sequence of implications:
          \begin{align*}
            &V, \yy \vdash (\forall \zz.\, \toClause{\vv{\theta}}{\zz} \Rightarrow \res{U''}{\vv{x_2}}[\zz/\vv{x_2}]) \rew^* \truetm \\
              &\quad\Rightarrow \toClause{\vv{\theta}}{\zz} \leqa \res{U''}{\vv{x_2}}[\zz/\vv{x_2}]  \\
              &\quad\Rightarrow \toClause{\vv{\theta'}}{\zz} \leqa \res{U''}{\vv{x_2}}[\zz/\vv{x_2}] \\
              &\quad\Rightarrow V, \yy \vdash (\forall \zz.\, \toClause{\vv{\theta'}}{\zz} \Rightarrow \res{U''}{\vv{x_2}}[\zz/\vv{x_2}]) \rew^* \truetm
          \end{align*}
      
      \end{itemize}
      We can put these two rewrites back together using \ruleName{AndL} and \ruleName{AndR}:
      \[
         V, \yy \vdash (\forall \zz.\,\toClause{\vv{\theta'}}{\zz} \Rightarrow \res{U''}{\vv{x_1}}[\vv{t'}/\vv{x_1}]) \land (\forall \zz.\,\toClause{\vv{\theta'}}{\zz} \Rightarrow \res{U''}{\vv{x_2}}[\zz/\vv{x_2}])
           \rew^* U_{2,j}
      \]
      All that remains is replicating the beginning of rewrite~\ref{eq:rew_types}:
      \begin{align*}
        &V, \yy \vdash (\forall \zz.\, \toClause{\vv{\theta'}}{\zz} \Rightarrow P\,(t\,\zz)) \\
          &\qquad\rew (\forall \zz.\, \toClause{\vv{\theta'}}{\zz} \Rightarrow U''[\vv{t'}/\vv{x_1},\zz/\vv{x_2}]) \\
          &\qquad\rew (\forall \zz.\,\toClause{\vv{\theta'}}{\zz} \Rightarrow \res{U''}{\vv{x_1}}[\vv{t'}/\vv{x_1}]) \land (\forall \zz.\,\toClause{\vv{\theta'}}{\zz} \Rightarrow \res{U''}{\vv{x_2}}[\zz/\vv{x_2}])\\
          &\qquad\rew^* U_{2,j}
      \end{align*} 
  
  \end{itemize}
 
   We have now proved that $V, \yy \vdash \toClause{\tau_{i_j}}{t} \rew^* U_{2,j}$ determines $V, \yy \vdash \toClause{\tau_j'}{t} \rew^* U_{2,j}'$ with $U_{2,j} \leqa U_{2,j}'$, for all $j \in [1..m]$.
   Thus, for every conjunct in $\toClause{\sigmaInt'}{t}$, there exists a conjunct in $\toClause{\sigmaInt}{t}$ that rewrites to something smaller, which proves $V, \yy \vdash \toClause{\sigmaInt'}{t} \rew^* U'$ for a $U'$ such that $U \leqa U'$.
\end{proof}

Given an isomorphic type environment and automaton formula as per Definition~\ref{def:type_system_automaton_formula}, we establish an equivalence between typing judgements and rewrites.
To bridge the gap, we define $\boolType$ as a type-equivalent of truth and type goal formulas as in Figure~\ref{fig:typing_goals}.
This typing system distills the term-typing judgements that make a goal formula true.
  
\begin{figure}
  \[
    \arraycolsep=5pt
    \begin{array}{cc}
    
    \prftree[l]
        { \ruleName{GT-Atom} }
        { \Gamma \vdash \hasType{t}{q_P} }
        { \Gamma \vdash \hasType{P\,t}{\boolType} }
        
        \hspace*{20pt}
        
      \prftree[l]
        { \ruleName{GT-Imp} }
        { \Gamma, \hasType{\zz}{\fromClause{U}} \vdash \hasType{G}{\boolType} }
        { \Gamma \vdash \hasType{\forall \zz.\, U \Rightarrow G}{\boolType} }  
     
        \\[12pt]

      \prftree[l]
        { \ruleName{GT-And} }
        { \Gamma \vdash \hasType{G_1}{\boolType} }
        { \dots }
        { \Gamma \vdash \hasType{G_n}{\boolType} }
        { \Gamma \vdash \hasType{G_1 \land \dots \land G_n}{\boolType} }

    \end{array}
  \]
  \caption{Typing rules for goal formulas}\label{fig:typing_goals}
\end{figure}

\begin{lemma}
  \label{lem:typeable_automaton}
  For automaton formulas $U,U'$ over $\yy$, $U \leqa U'$ implies $\Gamma_V, \hasType{\yy}{\fromClause{U}} \vdash \hasType{U'}{\boolType}$.
\end{lemma}
\begin{proof}
  We proceed by structural induction on $U'$.
  Suppose $U \leqa U'$.
  \begin{itemize}
  
    \item Suppose $U' = \truetm$.
      Both $U \leqa \truetm$ and $\Gamma_V, \hasType{\yy}{\fromClause{U}} \vdash \hasType{\truetm}{\boolType}$ trivially hold.
    
    \item Suppose $U' = P\,y_i$.
      This means that $\fromClause{\res{U}{y_i}} \leq q_P$ and, thus, $q_P \in \fromClause{\res{U}{y_i}}$.
      We aim to show $\Gamma_V, \hasType{\yy}{\fromClause{U}} \vdash \hasType{P\,y_i}{\boolType}$, which holds just if $\Gamma_V, \hasType{\yy}{\fromClause{U}} \vdash \hasType{y_i}{q_P}$.
      We prove the latter with \ruleName{T-Var}, thanks to $q_P \in \fromClause{\res{U}{y_i}}$.
           
    \item Suppose $U' = (\forall \xx.\, U'' \Rightarrow P\,(y_i\,\xx))$ for some $\xx = x_1 \dots x_m$.
      This means $\fromClause{\res{U}{y_i}} \leqa \fromClause{\forall \xx.\, U'' \Rightarrow P\,(y_i\,\xx)}$.
      By Lemma~\ref{lem:subtyping}, there exists $T \in \res{U}{y_i}$ such that $\fromClause{T} \leqa \fromClause{\forall \xx.\, U'' \Rightarrow P\,(y_i\,\xx)}$.
           
      Let $T = (\forall \xx.\, U''' \Rightarrow P\,(y_i\,\xx))$ be such a $T$.
      This gives us $\fromClause{\res{U''}{x_j}} \leq \fromClause{\res{U'''}{x_j}}$ by subtyping rule \ruleName{Q-Arr}, for all $j \in [1..m]$, which implies $U'' \leqa U'''$.
           
      Finally, we apply \ruleName{T-App} to:
      \begin{itemize}
      
        \item $\Gamma_V, \hasType{\yy}{\fromClause{U}}, \hasType{\xx}{\fromClause{U''}} \vdash \hasType{y_i}{\fromClause{\forall \xx.\, U''' \Rightarrow P\,(y_i\,\xx)}}$, by \ruleName{T-Var}
        
        \item for all $j \in [1..m]$: $\Gamma_V, \hasType{\yy}{\fromClause{U}}, \hasType{\xx}{\fromClause{U''}} \vdash \hasType{x_j}{\fromClause{\res{U''}{x_j}}}$, by \ruleName{T-Var}
        
        \item for all $j \in [1..m]$: $\fromClause{\res{U''}{x_j}} \leq \fromClause{\res{U'''}{x_j}}$
        
      \end{itemize}
      This gives us $\Gamma_V, \hasType{\yy}{\fromClause{U}}, \hasType{\xx}{\fromClause{U''}} \vdash \hasType{y_i\,\xx}{q_P}$, to which we apply \ruleName{GT-Atom} and \ruleName{GT-Imp} to prove the claim.
          
    \item Suppose $U' = U_1' \land U_2'$.
      This means that $\fromClause{\res{U}{y_i}} \leq \fromClause{\res{(U_1' \land U_2')}{y_i}}$, for all $i \in [1..n]$.
      Clearly, also $\fromClause{\res{U}{y_i}} \leq \fromClause{\res{{U_1}'}{y_i}}$ and $\fromClause{\res{U}{z_i}} \leq \fromClause{\res{{U_2}'}{y_i}}$, for all $i \in [1..n]$.
      By the IH, $\Gamma_V, \hasType{\yy}{\fromClause{U}} \vdash \hasType{U_1'}{\boolType}$ and $\Gamma_V, \hasType{\yy}{\fromClause{U}} \vdash \hasType{U_2'}{\boolType}$.
      Finally, \ruleName{GT-And} gives us $\Gamma_V, \hasType{\yy}{\fromClause{U}} \vdash \hasType{U_1' \land U_2'}{\boolType}$.
  
  \end{itemize}
\end{proof}

\begin{lemma}[Type-clause correspondence]
  \label{lem:type-clause}
  
  For all $U$, there exists $U'$ such that $U \leqa U'$ and $V, \yy \vdash G \rew^* U'$ if, and only if $\Gamma_V, \hasType{\yy}{\fromClause{U}} \vdash \hasType{G}{\boolType}$.
\end{lemma}
\begin{proof}
  Lemma~\ref{lem:typeable_automaton} proves $\Gamma_V, \hasType{\yy}{\fromClause{U}} \vdash \hasType{U'}{\boolType}$ from $U \leqa U'$, so it suffices to show the following for direction $\Rightarrow$:
  if $V, \yy \vdash G_1 \rew G_2$ and $\Gamma_V, \hasType{\yy}{\fromClause{U}} \vdash \hasType{G_2}{\boolType}$, then $\Gamma_V, \hasType{\yy}{\fromClause{U}} \vdash \hasType{G_1}{\boolType}$.
  
  We use case analysis on the rewrite rule used in $V, \yy \vdash G_1 \rew G_2$.
  \begin{description}
  
    \item[\ruleName{Refl}]
      Suppose $V, \yy \vdash P\,x \rew \truetm$ and $\Gamma_V, \hasType{\yy}{\fromClause{U}} \vdash \hasType{\truetm}{\boolType}$.
      The side condition is $P\,x \in V$, which means $q_P \in \fromClause{\res{V}{x}}$ and $\hasType{x}{q_P} \in \Gamma_V$.
      Now, $\Gamma_V, \hasType{\yy}{\fromClause{U}} \vdash \hasType{P\,x}{\boolType}$ can be proved using \ruleName{GT-Atom} and \ruleName{T-Var}. 
    
    \item[\ruleName{Step}]
      Suppose $V, \yy \vdash P\,(f\,\ss) \rew U'[\ss/\xx]$ and $\Gamma_V, \hasType{\yy}{\fromClause{U}} \vdash \hasType{U'[\ss/\xx]}{\boolType}$.
      The side condition is $(\forall \xx.\, U' \Rightarrow P\,(f\,\xx)) \in V$, which means $\fromClause{\forall \xx.\, U' \Rightarrow P\,(f\,\xx)} \in \fromClause{\res{V}{f}}$ and $\hasType{f}{\fromClause{\forall \xx.\, U' \Rightarrow P\,(f\,\xx)}} \in \Gamma_V$. 
      
      Let $U' = T_1 \land \dots \land T_k$.
      Clearly, there exists a proof $\Gamma_V, \hasType{\yy}{\fromClause{U}} \vdash \hasType{U'[\ss/\xx]}{\boolType}$ just if $\Gamma_V, \hasType{\yy}{\fromClause{U}} \vdash \hasType{T_j[\ss/\xx]}{\boolType}$, for all $j \in [1..k]$, using \ruleName{GT-And}.
      If $T_j$ contains a (single) variable $x_i \in x_1 \dots x_n$ of type $\iota$, then the proof $\Gamma_V, \hasType{\yy}{\fromClause{U}} \vdash \hasType{T_j[\ss/\xx]}{\boolType}$ uses \ruleName{GT-Atom} at the root.
      Otherwise, if the (single) variable $x_i \in x_1 \dots x_n$ in $T_j$ does not have type $\iota$, then the proof $\Gamma_V, \hasType{\yy}{\fromClause{U}} \vdash \hasType{T_j[\ss/\xx]}{\boolType}$ uses \ruleName{GT-Imp} at the root, followed by \ruleName{GT-Atom};
      i.e.~if $T_j = (\forall \zz.\, U'' \Rightarrow P'\,(x_i\,\zz))$, then we have $\Gamma_V, \hasType{\yy}{\fromClause{U}} \vdash \hasType{T_j[\ss/\xx]}{\boolType}$ just if $\Gamma_V, \hasType{\yy}{\fromClause{U}}, \hasType{\zz}{\fromClause{U''}} \vdash \hasType{P'\,(s_i\,\zz)}{\boolType}$. 
      
      We have now identified all $k$ nodes in $\Gamma_V, \hasType{\yy}{\fromClause{U}} \vdash \hasType{U'[\ss/\xx]}{\boolType}$ at which \ruleName{GT-Atom} is used.
      Thanks to their premises, $T_j = (\forall \zz.\, U'' \Rightarrow P'(x_i\,\zz))$ implies $\Gamma_V, \hasType{\yy}{\fromClause{U}}, \hasType{\zz}{\fromClause{U''}} \vdash \hasType{s_i\,\zz}{q_{P'}}$.
      Let $\zz = z_1 \dots z_m$.
      By \ruleName{T-App}, there exist $\sigmaInt_\ell,\sigmaInt_\ell'$, for all $\ell \in [1..m]$, such that:
      \begin{itemize}
      
        \item $\Gamma_V, \hasType{\yy}{\fromClause{U}}, \hasType{\zz}{\fromClause{U''}} \vdash \hasType{s_i}{\sigmaInt_1 \to \dots \to \sigmaInt_m \to q_{P'}}$
        
        \item for all $\ell \in [1..m]$: $\Gamma_V, \hasType{\yy}{\fromClause{U}}, \hasType{\zz}{\fromClause{U''}} \vdash \hasType{z_\ell}{\sigmaInt_\ell'}$
        
        \item for all $\ell \in [1..m]$: $\sigmaInt_\ell' \leq \sigmaInt_\ell$
      
      \end{itemize}
      Typing rule \ruleName{T-Var} guarantees $\sigmaInt_\ell' \subseteq \fromClause{\res{U''}{z_\ell}}$ and, thus, $\fromClause{\res{U''}{z_\ell}} \leq \sigmaInt_\ell' \leq \sigmaInt_\ell$.
      It follows that $\sigmaInt_1 \to \dots \to \sigmaInt_m \to q_{P'} \leq \fromClause{\res{U''}{z_1}} \to \dots \to \fromClause{\res{U''}{z_m}} \to q_{P'} = \fromClause{T_j}$.
      
      Let us refer to $\sigmaInt_1 \to \dots \to \sigmaInt_m \to q_{P'}$ as $\tau_j$, the type extracted from $T_j$ for (its only variable) $x_i$.
      Recall $\tau_j \leq \fromClause{T_j}$.
      We define an type $\theta_i$ for each $s_i$ such that $\theta_i \leq \fromClause{\res{U'}{x_i}}$:
      \begin{align*}
         \theta_i &\defeq \set{\tau_j \mid T_j \in \res{U'}{x_i} \land j \in [1..k]}
      \end{align*}

      Finally, we have the following:
      \begin{itemize}
      
        \item $\Gamma_V, \hasType{\yy}{\fromClause{U}} \vdash \hasType{f}{\fromClause{\forall \xx.\, U' \Rightarrow P\,(f\,\xx)}}$, by \ruleName{T-Con}
        
        \item for all $i \in [1..n]$: $\Gamma_V, \hasType{\yy}{\fromClause{U}} \vdash \hasType{s_i}{\theta_i}$
        
        \item for all $i \in [1..n]$: $\theta_i \leq \fromClause{\res{U'}{x_i}}$
        
      \end{itemize}
      We conclude the case with \ruleName{T-App} with those premises, followed by \ruleName{GT-Atom}.  
      
    \item[\ruleName{Assm}]
      Let $\xx = x_1 \dots x_n$ and $\yy = y_1 \dots y_m$.
      Suppose $V, \yy \vdash P\,(y_j\,\ss) \rew U''[\ss/\xx] \land (\forall \xx.\, U'' \Rightarrow P\,(y\,\xx))$ and $\Gamma_V, \hasType{\yy}{\fromClause{U}} \vdash \hasType{U''[\ss/\xx] \land (\forall \xx.\, U'' \Rightarrow P\,(y_j\,\xx))}{\boolType}$, for some $j \in [1..m]$.
      
      The typing proof uses \ruleName{GT-And} at the root with $\Gamma_V, \hasType{\yy}{\fromClause{U}} \vdash \hasType{U''[\ss/\xx]}{\boolType}$ on the LHS and on the RHS: 
      \[
        \prftree[l]
          { \ruleName{GT-Imp} }
          {\prftree[l]
            { \ruleName{GT-Atom} }
            { \Gamma_V, \hasType{\yy}{\fromClause{U}}, \hasType{\xx}{\fromClause{U''}} \vdash \hasType{y_j\,\xx}{q_P} }
            { \Gamma_V, \hasType{\yy}{\fromClause{U}}, \hasType{\xx}{\fromClause{U''}} \vdash \hasType{P\,(y_j\,\xx)}{\boolType} }
          }
          { \Gamma_V, \hasType{\yy}{\fromClause{U}} \vdash \hasType{\forall \xx.\, U'' \Rightarrow P\,(y_j\,\xx)}{\boolType} }
      \]
      Typing rule \ruleName{T-App} now guarantees the existence of $\sigmaInt_i,\sigmaInt_i'$, for all $i \in [1..n]$, such that:
      \begin{itemize}
      
       \item $\Gamma_V, \hasType{\yy}{\fromClause{U}}, \hasType{\xx}{\fromClause{U''}} \vdash \hasType{y_j}{\sigmaInt_1 \to \dots \to \sigmaInt_m \to q_P}$
        
        \item for all $i \in [1..n]$: $\Gamma_V, \hasType{\yy}{\fromClause{U}}, \hasType{\xx}{\fromClause{U''}} \vdash \hasType{x_i}{\sigmaInt'_i}$
        
        \item for all $i \in [1..n]$: $\sigmaInt'_i \leq \sigmaInt_i$
        
      \end{itemize}
      The first gives us $\sigmaInt_1 \to \dots \to \sigmaInt_m \to q_P \in \fromClause{\res{U}{y_j}}$ thanks to \ruleName{T-Var}, and the second $\sigmaInt_i' \subseteq \fromClause{\res{U''}{x_i}}$.
      This implies $\fromClause{\res{U''}{x_i}} \leq \sigmaInt_i' \leq \sigmaInt_i$.
      The third bullet point now yields $\sigmaInt_1 \to \dots \to \sigmaInt_n \to q_P \leq \fromClause{\res{U''}{x_1}} \to \dots \to \fromClause{\res{U''}{x_n}} \to q_P = \fromClause{\forall \xx.\, U'' \Rightarrow P\,(y_j\,\xx)}$. 
      
      The left premise $\Gamma_V, \hasType{\yy}{\fromClause{U}} \vdash \hasType{U''[\ss/\xx]}{\boolType}$ of the assumption allows us to prove the existence of $\theta_i$ such that $\theta_i \leq \fromClause{\res{U''}{x_i}}$ and $\Gamma_V, \hasType{\yy}{\fromClause{U}} \vdash \hasType{s_i}{\theta_i}$ for all $i \in [1..n]$, as we did in \ruleName{Step}.
      This means that $\theta_i \leq \fromClause{\res{U''}{x_i}} \leq \sigmaInt_i' \leq \sigmaInt_i$.
      
      Finally, we prove $\Gamma_V, \hasType{\yy}{\fromClause{U}} \vdash \hasType{P\,(y_j\,\ss)}{\boolType}$ with \ruleName{T-App} and \ruleName{GT-Atom} on:
      \begin{itemize}
      
       \item $\Gamma_V, \hasType{\yy}{\fromClause{U}} \vdash \hasType{y_j}{\sigmaInt_1 \to \dots \to \sigmaInt_m \to q_P}$
        
        \item for all $i \in [1..n]$: $\Gamma_V, \hasType{\yy}{\fromClause{U}} \vdash \hasType{s_i}{\theta_i}$
        
        \item for all $i \in [1..n]$: $\theta_i \leq \sigmaInt_i$
        
      \end{itemize} 
    
    \item[\ruleName{ImpAnd}]
      Suppose $V, \yy \vdash  (\forall \zz.\, U'' \Rightarrow G_1' \land G_2') \rew  (\forall \zz.\, U'' \Rightarrow G_1') \land (\forall \zz .\, U'' \Rightarrow G_2')$ and $\Gamma_V, \hasType{\yy}{\fromClause{U}} \vdash \hasType{ (\forall \zz.\, U'' \Rightarrow G_1') \land (\forall \zz .\, U'' \Rightarrow G_2')}{\boolType}$.
      
      Our assumption $\Gamma_V, \hasType{\yy}{\fromClause{U}} \vdash \hasType{ (\forall \zz.\, U'' \Rightarrow G_1') \land (\forall \zz .\, U'' \Rightarrow G_2')}{\boolType}$ uses \ruleName{GT-And} at the root with the subproofs given by, for $i \in [1,2]$:
      \[
          \prftree[l]
            { \ruleName{GT-Imp} }
            { \Gamma_V, \hasType{\yy}{\fromClause{U}}, \hasType{\zz}{\fromClause{U''}} \vdash \hasType{G_i'}{\boolType} }
            { \Gamma_V, \hasType{\yy}{\fromClause{U}} \vdash \hasType{ \forall \zz.\, U'' \Rightarrow G_i'}{\boolType} }
      \]
      Let us refer to premise $\Gamma_V, \hasType{\yy}{\fromClause{U}}, \hasType{\zz}{\fromClause{U''}} \vdash \hasType{G_i'}{\boolType}$ as $\prfAb_i$, for $i \in [1,2]$.

      This allows us to construct:
      \[
        \prftree[l]
          { \ruleName{GT-Imp} }
          {\prftree[l]
            { \ruleName{GT-And} }
            { \prfAb_1 \qquad \qquad \qquad \qquad }  
            { \prfAb_2 } 
            { \Gamma_V, \hasType{\yy}{\fromClause{U}}, \hasType{\zz}{\fromClause{U''}} \vdash \hasType{G_1' \land G_2'}{\boolType} }
          }
          { \Gamma_V, \hasType{\yy}{\fromClause{U}} \vdash \hasType{ \forall \zz.\, U'' \Rightarrow G_1' \land G_2'}{\boolType} }
      \] 
    
     
      
    
    \item[\ruleName{Scope}]
      Suppose $V, \yy \vdash (\forall \zz.\, U_1 \Rightarrow G_2) \rew G_2$ and $\Gamma_V, \hasType{\yy}{\fromClause{U}} \vdash \hasType{G_2}{\boolType}$ with $\zz \cap \fv(G_2) = \emptyset$.
      Thanks to Weakening Lemma~\ref{lem:weakening}, $\Gamma_V, \hasType{\yy}{\fromClause{U}}, \hasType{\zz}{\fromClause{U_1}} \vdash \hasType{G_2}{\boolType}$ also holds.
      Finally, \ruleName{GT-Imp} proves this case.
    
    \item[\ruleName{Imp}]
      This is the first inductive case.
      Suppose $V, \yy \vdash (\forall \zz.\, U_1 \Rightarrow G_1') \rew (\forall \zz.\, U_1 \Rightarrow G_2')$ and $\Gamma_V, \hasType{\yy}{\fromClause{U}} \vdash \hasType{(\forall \zz.\, U_1 \Rightarrow G_2')}{\boolType}$.
      The latter gives us $\Gamma_V, \hasType{\yy}{\fromClause{U}}, \hasType{\zz}{\fromClause{U_1}} \vdash \hasType{G_2'}{\boolType}$ by \ruleName{GT-Imp}, which is equal to $\Gamma_{V \land U_1}, \hasType{\yy}{\fromClause{U}} \vdash \hasType{G_2'}{\boolType}$. 
      
      We use the IH on premise $V \land U_1, \yy \vdash G_1' \rew G_2'$ of \ruleName{Imp} for $\Gamma_{V \land U_1}, \hasType{\yy}{\fromClause{U}} \vdash \hasType{G_1'}{\boolType}$, which is equal to $\Gamma_V, \hasType{\yy}{\fromClause{U}}, \hasType{\zz}{\fromClause{U_1}} \vdash \hasType{G_1'}{\boolType}$.
      Finally, \ruleName{GT-Imp} gives us $\Gamma_V, \hasType{\yy}{\fromClause{U}} \vdash \hasType{(\forall \zz.\, U_1 \Rightarrow G_1')}{\boolType}$, as required. 
          
    \item[\ruleName{AndL}]
      Suppose $V, \yy \vdash G_{1} \land G_{2} \rew G_{1}' \land G_{2}$ and $\Gamma_V, \hasType{\yy}{\fromClause{U}} \vdash \hasType{G_{1}' \land G_{2}}{\boolType}$.
      The proof is straightforward: 
      use (backwards) \ruleName{GT-And} on $\Gamma_V, \hasType{\yy}{\fromClause{U}} \vdash \hasType{G_{1}' \land G_{2}}{\boolType}$, invoke the IH on the parts, and use \ruleName{GT-And} for $\Gamma_V, \hasType{\yy}{\fromClause{U}} \vdash \hasType{G_{1} \land G_{2}}{\boolType}$.

  \end{description}
  
  For direction $\Leftarrow$, suppose $\Gamma_V, \hasType{\yy}{\fromClause{U}} \vdash \hasType{G}{\boolType}$.
  We use structural induction on $G$ to prove $V, \yy \vdash G \rew^* U_0$ for some $U_0$ such that $U \leqa U_0$.
  \begin{description}
  
    \item[\ruleName{GT-Atom}]
      Suppose $\Gamma_V, \hasType{\yy}{\fromClause{U}} \vdash \hasType{P\,s}{\boolType}$ with premise $\Gamma_V, \hasType{\yy}{\fromClause{U}} \vdash \hasType{s}{q_P}$.
      We use structural induction on term $s$. 
      \begin{description}
      
        \item[\ruleName{T-Con}]
          Suppose $s = f:\iota \in \Sigma$.
          By the premise, $q_P \in \fromClause{\res{V}{f}}$ and, thus, $P\,f \in \res{V}{f}$.
          We use \ruleName{Refl} with this side condition for:
          \[
            V, \yy \vdash P\,f \rew \truetm
          \]
          Since $U \leq \truetm$, this proves the case.
          
        \item[\ruleName{T-Var}]
          If $s = z \not\in \yy$, then we proceed as in the previous case, so suppose $s = y_j \in \yy$.
          Clearly, $P\,y_j$ is automaton wrt $\yy$.
          Thus, $V, \yy \vdash P\,y_j \rew^* P\,y_j$.
          
          It remains to show $U \leqa P\,y_j$.
          Premise $\Gamma_V, \hasType{\yy}{\fromClause{U}} \vdash \hasType{y_j}{q_P}$ guarantees $q_P \in \fromClause{\res{U}{y_j}}$, by \ruleName{T-Var}, which implies $U \leqa \res{U}{y_j} \leqa P\,y_j$, as required. 
          
        \item[\ruleName{T-App}]
          We must distinguish the case where $s = y_j\,\vv{t}$ with $y_j \in \yy$ from $s = f\,\vv{t}$ with $f \in \Sigma$, but first assume $s = f\,t_1 \dots t_n = f\,\vv{t}$ for some $f \in \Sigma$. 

          Premise $\Gamma_V, \hasType{\yy}{\fromClause{U}} \vdash \hasType{f\,\vv{t}}{q_P}$ guarantees $\sigmaInt_1 \to \dots \to \sigmaInt_n \to q_P \in \fromClause{\res{V}{f}}$, by \ruleName{T-App} and \ruleName{T-Con}, with $\Gamma_V, \hasType{\yy}{\fromClause{U}} \vdash \hasType{t_i}{\sigmaInt_i'}$ such that $\sigmaInt'_i \leq \sigmaInt_i$, for all $i \in [1..n]$.
          Thus, there exists an automaton clause $(\forall \xx.\, \toClause{\sigmaInt_1}{x_1} \land \dots \land \toClause{\sigmaInt_n}{x_n} \Rightarrow P\,(f\,\xx)) \in \res{V}{f}$ we can use as a side condition for \ruleName{Step} in:
          \[
            V, \yy \vdash P\,(f\,\vv{t}) \rew \toClause{\sigmaInt_1}{t_1} \land \dots \land \toClause{\sigmaInt_n}{t_n}
          \]
          We aim to show that each $\toClause{\tau}{t_i} \in \toClause{\sigmaInt_i}{t_i}$ rewrites to an automaton clause larger than $U$, for all $i \in [1..n]$.
          Let us consider such a $\toClause{\tau}{t_i}$. 
          \begin{itemize}
          
            \item Suppose $t_i : \iota$.
              Let $\tau = q_{P'}$, so $\toClause{\tau}{t_i} = P'\,t_i$.
              Recall we have $\Gamma_V, \hasType{\yy}{\fromClause{U}} \vdash \hasType{t_i}{\sigmaInt_i'}$. 
              Since $\sigmaInt_i' \leq \sigmaInt_i$ implies $\sigmaInt'_i \supseteq \sigmaInt_i$ for type $\iota$, we also have $\Gamma_V, \hasType{\yy}{\fromClause{U}} \vdash \hasType{t_i}{q_{P'}}$.
              
              Typing rule \ruleName{GT-Atom} gives us $\Gamma_V, \hasType{\yy}{\fromClause{U}} \vdash \hasType{P'\,t_i}{\boolType}$, to which we can apply the IH for $V, \yy \vdash P'\,t_i \rew^* U_{\tau}$ such that $U \leqa U_{\tau}$.
          
            \item Suppose $t_i : \sigma_{1} \to \dots \to \sigma_{k} \to \iota$ for $k > 0$.
              Thanks to Lemma~\ref{lem:subtyping}, the inequality $\sigmaInt_i' \leq \sigmaInt_i$ with $\tau \in \sigmaInt_i$ guarantees the existence of $\tau' \in \sigmaInt'$ such that $\tau' \leq \tau$.
              
              Recall we have $\Gamma_V, \hasType{\yy}{\fromClause{U}} \vdash \hasType{t_i}{\sigmaInt_i'}$ and, thus, $\Gamma_V, \hasType{\yy}{\fromClause{U}} \vdash \hasType{t_i}{\tau'}$.
              Let $\toClause{\tau}{t_i} = (\forall \zz.\, U_0 \Rightarrow P'\,(t_i\,\zz))$ and $\toClause{\tau'}{t_i} = (\forall \zz.\, U_0' \Rightarrow P'\,(t_i\,\zz))$.
              Since $\tau' \leq \tau$, it holds that $U_0 \leqa U_0'$.
              This gives us the following:
              \begin{itemize}
      
                \item $\Gamma_V, \hasType{\yy}{\fromClause{U}}, \hasType{\zz}{\fromClause{U_0}} \vdash \hasType{t_i}{\fromClause{\forall \zz.\, U_0' \Rightarrow P'\,(t_i\,\zz)}}$ by Weakening Lemma~\ref{lem:weakening}
        
                \item for all $z_\ell \in \zz$: $\Gamma_V, \hasType{\yy}{\fromClause{U}}, \hasType{\zz}{\fromClause{U_0}} \vdash \hasType{z_\ell}{\fromClause{\res{{U_0}}{z_\ell}}}$, by \ruleName{T-Var}
        
                \item for all $z_\ell \in \zz$: $\fromClause{\res{{U_0}}{z_\ell}} \leq \fromClause{\res{{U_0'}}{z_\ell}}$
        
              \end{itemize}
              We apply \ruleName{T-Var} to these premises for $\Gamma_V, \hasType{\yy}{\fromClause{U}}, \hasType{\zz}{\fromClause{U_0}} \vdash \hasType{t_i\,\zz}{q_{P'}}$.
              By \ruleName{GT-Atom}, $\Gamma_V, \hasType{\yy}{\fromClause{U}}, \hasType{\zz}{\fromClause{U_0}} \vdash \hasType{P'\,(t_i\,\zz)}{\boolType}$, which is equal to $\Gamma_{V \land U_0}, \hasType{\yy}{\fromClause{U}} \vdash \hasType{P'\,(t_i\,\zz)}{\boolType}$.
              
              Finally, by the IH, there exists $U_{\tau}$ such that $V \land U_0, \yy \vdash P'\,(t_i\,\zz) \rew^* U_{\tau}$ and $U \leqa U_{\tau}$.
              Typing rule \ruleName{GT-Imp} gives us $V, \yy \vdash (\forall \zz.\, U_0 \Rightarrow P'\,(t_i\,\zz)) \rew^* (\forall \zz.\, U_0 \Rightarrow U_{\tau})$.
              Observe that $U_{\tau}$ is automaton wrt $\yy$. 
              Therefore, it does not contain $\zz$, and we can apply \ruleName{Scope} for:
              \[
                V, \yy \vdash 
                  (\forall \zz.\, U_0 \Rightarrow P'\,(t_i\,\zz))
                   \rew^* (\forall \zz.\, U_0\Rightarrow U_{\tau})
                   \rew U_{\tau}
              \]
          
          \end{itemize}

          We have proved, for all $i \in [1..n]$, that every $\toClause{\tau}{t_i} \in \toClause{\sigmaInt_i}{t_i}$ rewrites to some automaton clause greater than $U$.
          Therefore, so does $\toClause{\sigmaInt_i}{t_i}$, and so does $\toClause{\sigmaInt_1}{t_1} \land \dots \land \toClause{\sigmaInt_n}{t_n}$.
          We conclude the case with:
          \[
            V, \yy \vdash P\,(f\,\vv{t}) \rew \toClause{\sigmaInt_1}{t_1} \land \dots \land \toClause{\sigmaInt_n}{t_n}
          \]
          
          Now, suppose $s = y_j\,t_1 \dots t_n$ with $y_j \in \yy$ instead.
          Similar to before, premise $\Gamma_V, \hasType{\yy}{\fromClause{U}} \vdash \hasType{y_j\,\vv{t}}{q_P}$ guarantees $\sigmaInt_1 \to \dots \to \sigmaInt_n \to q_P \in \fromClause{\res{U}{y_j}}$, by \ruleName{T-App} and \ruleName{T-Var}, with $\Gamma_V, \hasType{\yy}{\fromClause{U}} \vdash \hasType{t_i}{\sigmaInt_i'}$ such that $\sigmaInt'_i \leq \sigmaInt_i$, for all $i \in [1..n]$.
          Thus, there exists an automaton clause $T = (\forall \xx.\, \toClause{\sigmaInt_1}{x_1} \land \dots \land \toClause{\sigmaInt_n}{x_n} \Rightarrow P\,(y_j\,\xx)) \in \res{U}{y_j}$ we can use as assumption in \ruleName{Assm}:
          \[
            V, \yy \vdash P\,(f\,\vv{t}) \rew \toClause{\sigmaInt_1}{t_1} \land \dots \land \toClause{\sigmaInt_n}{t_n} \land T
          \]
          Conjunct $T$ is in $U$ by definition.
          For the other conjuncts, we proceed as before.
          
      \end{description}
    
    \item[\ruleName{GT-And}]
      Suppose $\Gamma_V, \hasType{\yy}{\fromClause{U}} \vdash \hasType{G_1 \land G_2}{\boolType}$ with premises $\Gamma_V, \hasType{\yy}{\fromClause{U}} \vdash \hasType{G_1}{\boolType}$ and $\Gamma_V, \hasType{\yy}{\fromClause{U}} \vdash \hasType{G_2}{\boolType}$.
      By the IH, there exist $U_1,U_2$ such that $V, \yy \vdash G_1 \rew^* U_1$ and $U \leqa U_1$ and $V, \yy \vdash G_2 \rew^* U_2$ and $U \leqa U_2$.
      Clearly, $U \leqa U_1 \land U_2$, so the following rewrite proves the case, using \ruleName{AndL} and \ruleName{AndR}:
      \[
        V, \yy \vdash G_1 \land G_2 \rew^* U_1 \land G_2 \rew^* U_1 \land U_2
      \]
    
    \item[\ruleName{GT-Imp}]
      Suppose $\Gamma_V, \hasType{\yy}{\fromClause{U}} \vdash \hasType{\forall \zz.\, U' \Rightarrow G''}{\boolType}$ with premise $\Gamma_V, \hasType{\yy}{\fromClause{U}}, \hasType{\zz}{\fromClause{U'}} \vdash \hasType{G''}{\boolType}$, which is equal to $\Gamma_{V \land U'}, \hasType{\yy}{\fromClause{U}} \vdash \hasType{G''}{\boolType}$.
      By the IH, there exists $U''$ such that $V \land U', \yy \vdash G'' \rew^* U''$ and $U \leqa U''$.
      Rewrite rule \ruleName{Imp} gives us $V, \yy \vdash (\forall \zz.\, U' \Rightarrow G'') \rew^* (\forall \zz.\, U' \Rightarrow U'')$.
      
      It remains to show this rewrites further to something larger than $U$.
      Observe that $U''$ is automaton wrt $\yy$. 
      Therefore, it does not contain $\zz$, and we can apply \ruleName{Scope} for:
      \[
        V, \yy \vdash (\forall \zz.\, U' \Rightarrow G'') \rew^* (\forall \zz.\, U' \Rightarrow U'') \rew U''
      \]
      
%
%
%
%
%

  \end{description}
\end{proof}

\typeClause*
\begin{proof}
  We aim to show $\Gamma \vdash \hasType{t}{\sigmaInt}$ if, and only if, $\Gamma \vdash \hasType{\toClause{\sigmaInt}{t}}{\boolType}$;
  this proves \ref{enum:typeClause_without_goal_typing}, since Lemma~\ref{lem:type-clause} already gives us $\exists U'.\, U \leqa U' \land V, \yy \vdash \toClause{\sigmaInt}{t} \rew^* U'$ if, and only if $\Gamma_V, \hasType{\yy}{\fromClause{U}} \vdash \hasType{\toClause{\sigmaInt}{t}}{\boolType}$, for all $U$.
  Furthermore, if $\size{\yy} = 0$, then $U = \truetm$ and $U' = \truetm$, so \ref{enum:typeClause_true} is an instance of \ref{enum:typeClause_without_goal_typing}.
  
  It remains to prove $\Gamma \vdash \hasType{t}{\sigmaInt}$ if, and only if, $\Gamma \vdash \hasType{\toClause{\sigmaInt}{t}}{\boolType}$.
  
  Suppose $t:\iota$.
  This yields the following sequence of equivalences:
  \begin{align*}
    \Gamma \vdash \hasType{t}{\sigmaInt}
      &\iff \forall q_P \in \sigmaInt.\, \Gamma \vdash \hasType{t}{q_P} \\
      &\iff \forall q_P \in \sigmaInt.\, \Gamma \vdash \hasType{P\,t}{\boolType} \\
      &\iff \forall q_P \in \sigmaInt.\, \Gamma \vdash \hasType{\toClause{q_P}{t}}{\boolType} \\
      &\iff \Gamma \vdash \hasType{\toClause{\sigmaInt}{t}}{\boolType}
  \end{align*}
  where we rely on notation, \ruleName{GT-Atom}, $\toClause{q_P}{t} = P\,t$, and \ruleName{GT-And}, resp.
  
  Suppose $t: \sigma_1 \to \dots \to \sigma_m \to \iota$ with $m>0$.
  It suffices to prove the claim for a strict type $\tau$ of this type, thanks to \ruleName{GT-And}.
  Let $\tau = \sigmaInt_1 \to \dots \to \sigmaInt_m \to q_P$.
  We obtain the following equivalences:
  \begin{align*}
    \Gamma \vdash \hasType{t}{\tau}
      &\iff \Gamma, \hasType{\zz}{\vv{\sigmaInt}} \vdash \hasType{t}{\tau} \\
      &\iff \Gamma, \hasType{\zz}{\vv{\sigmaInt}} \vdash \hasType{t\,\zz}{q_P} \\
      &\iff \Gamma, \hasType{\zz}{\vv{\sigmaInt}} \vdash \hasType{P\,(t\,\zz)}{\boolType} \\
      &\iff \Gamma \vdash \hasType{\forall \zz.\, \toClause{\vv{\sigmaInt}}{\zz} \Rightarrow P\,(t\,\zz)}{\boolType} \\
      &\iff \Gamma \vdash \hasType{\toClause{\tau}{t}}{\boolType}
  \end{align*}
  where we respectively rely on Weakening Lemma~\ref{lem:weakening}, \ruleName{T-App}, \ruleName{GT-Atom}, \ruleName{GT-Imp}, and $\toClause{\tau}{t} = (\forall \zz.\, \toClause{\vv{\sigmaInt}}{\zz} \Rightarrow P\,(t\,\zz))$.

\end{proof}

\subsection{Reducing HORS intersection typing to \texorpdfstring{\MSL($\omega$)}{\MSL(omega)}}
\label{appx:reducing_types}

\begin{lemma}
  \label{lem:P-coconsistency}
  A non-empty $\Gamma$ is $\calG$-coconsistent if, and only if, there exist $(\forall \yy.\, U \Rightarrow P\,(f\,\yy)) \in V_\Gamma$ and $(\forall \yy.\, P\,t \Rightarrow P\,(f\,\yy)) \in D_\calG$ such that:
  \[
    \Gamma \backslash \set{ \hasType{f}{\tau} } \text{ is } \calG \text{-coconsistent and }
    \exists U'.\, V_{\Gamma \backslash \set{ \hasType{f}{\tau} }}, \yy \vdash P\,t \rew^* U' \land U \leqa U' 
  \] 
  where $\tau = \fromClause{\forall \yy.\, U \Rightarrow P\,(f\,\yy)}$. 
\end{lemma}
\begin{proof}
  Take the definition of $\calG$-coconsistency for non-empty $\Gamma$, and replace parts of the statement by equivalent ones, relying on the isomorphism between types and clauses, the construction of $D_\calG$, and Proposition~\ref{cor:type-clause}, resp.:
  \begin{itemize}
    \item $\hasType{f}{\tau} \in \Gamma$ if, and only if, $(\forall \yy.\, U \Rightarrow P\,(f\,\yy)) \in V_\Gamma$
    \item $\exists (f\,\yy = t) \in \calG$ if, and only if, $\exists (\forall \yy.\, P\,t \Rightarrow P\,(f\,\yy)) \in D_\calG$
    \item $\Gamma \backslash \set{ \hasType{f}{\tau}}, \hasType{\yy}{\fromClause{U}} \vdash \hasType{t}{q_P}$ if, and only if, $\exists U'.\, V_{\Gamma \backslash \set{ \hasType{f}{\tau} }}, \yy \vdash P\,t \rew^* U' \land U \leqa U' $
  \end{itemize}
\end{proof}

\begin{lemma}
  \label{lem:aux_red_types}
  \begin{enumerate}
    \item $\typeAppr{D_\calG} = \Gamma_{\apprx(D_\calG)}$ is $\calG$-coconsistent  \label{lem:types_typeAppr_cocons}
    \item If $\upclos{V}, \yy \vdash G \rew^* U$, then there exists $U'$ such that $V, \yy \vdash G \rew^* U'$ and $U \leqa U'$.  \label{lem:types_upwards_closed_inclusion}
    \item If $V_1 \leqa V_2$, then $V_2 \subseteq \upclos{V_1}$. \label{lem:types_automaton_order_to_inclusion}
    \item If $\Gamma$ is $\calG$-coconsistent, then $\apprx(D_\calG) \leqa V_\Gamma$. \label{lem:types_coconsistent_vs_automaton}
    \item If $\Gamma \vdash \hasType{t}{\tau}$ and $\Gamma$ is $\calG$-coconsistent, then there exists $\tau' \leq \tau$ such that $\Gamma_{\apprx(D_\calG)} \vdash \hasType{t}{\tau'}$. \label{lem:P-coconsistent_underapprox}
  \end{enumerate}
\end{lemma}
\begin{proof}
  We prove each part separately.
  \begin{enumerate}
  
    \item 
  
      Type environment $\typeAppr{D_\calG} = \Gamma_{\apprx(D_\calG)}$ is constructed inductively in parallel to an automaton approximation, which means there exists a linear order of environments $\Gamma_0,\Gamma_1,\dots, \Gamma_n$ where $\Gamma_0 = \emptyset$ and $\Gamma_n = \Gamma_{\apprx(D_\calG)}$ and $\Gamma_i = \Gamma_{i-1} \uplus \set{\hasType{f_i}{\fromClause{\forall \yy.\, U_i \Rightarrow P_i\,(f_i\,\yy)}}}$, for all $1 \leq i \leq n$, such that there exists $(\forall \yy.\, G_i \Rightarrow P_i\,(f_i\,\yy)) \in D_\calG$ and $V_{\Gamma_i} \vdash G_i \rew^* U_i$.
  
      By a straightforward induction with Lemma~\ref{lem:P-coconsistency}, each $\Gamma_i$ is $\calG$-coconsistent. 
      Therefore, so is $\typeAppr{D_\calG} = \Gamma_{\apprx(D_\calG)}$.
      
%
  
    \item
  
      Suppose $\upclos{V}, \yy \vdash G \rew^* U$.
      Clearly, $V \leqa \upclos{V}$.
      Starting from the end of the rewrite, we inductively replace each application of \ruleName{Step} with a $T' \in \upclos{V}$ by \ruleName{Step} with a smaller $T \in V$, which exists by Subtyping Lemma~\ref{lem:subtyping}.  
      (We do not need to consider \ruleName{Refl}, because $P\,x \in \upclos{V}$ if and only if $P\,x \in V$.)
  
      Suppose $\upclos{V}, \yy \vdash G \rew G' \rew^n U$ where $\upclos{V}, \yy \vdash G \rew G'$ is the only step to use an automaton clause from $\upclos{V} \backslash V$.
      Let $T' = (\forall \xx.\, U_0' \Rightarrow P\,(f\,\xx)) \in \upclos{V}$ be the side condition, so $G = P\,(f\,\ss)$ and $G' = U_0'[\ss/\xx]$.
      Let $T = (\forall \xx.\, U_0 \Rightarrow P\,(f\,\xx)) \in V$ such that $T \leqa T'$.
      This means $U_0' \leqa U_0$.
      
       By monotonicity of rewrite (Lemma~\ref{lem:monotonicity_of_rewrite}), there exists $U'$ such that $U \leqa U'$ and $V, \yy \vdash G \rew U_0[\ss/\xx] \rew^n U'$.
       
    \item 
    
      Suppose $V_1 \leqa V_2$ over variables $\yy$.
      By definition, for all $y \in \yy$, $\fromClause{\res{{V_1}}{y}} \leq \fromClause{\res{{V_2}}{y}}$.
      By Subtyping Lemma~\ref{lem:subtyping}, for every $\tau_2 \in \fromClause{\res{{V_2}}{y}}$, there exists $\tau_1 \in \fromClause{\res{{V_1}}{y}}$ such that $\tau_1 \leq \tau_2$, which is $\toClause{\tau_1}{y} \leqa \toClause{\tau_2}{y}$ in the clause setting.

      Clearly, $\toClause{\tau_2}{y} \in \upclos{\toClause{\tau_1}{y}}$.
      It follows from $\toClause{\tau_1}{y} \in  \fromClause{\res{{V_1}}{y}}$ and $\toClause{\tau_2}{y} \in \fromClause{\res{{V_2}}{y}}$ that $V_2 \subseteq \upclos{V_1}$, as required.
   
    \item
  
      We use induction on the size of $\Gamma$.  
      If $\Gamma$ is empty, then $V_\Gamma = \emptyset \subseteq \apprx(D_\calG)$, so suppose $\Gamma$ is $\calG$-coconsistent and non-empty.

      As per Lemma~\ref{lem:P-coconsistency}, let $(\forall \yy.\, U \Rightarrow P\,(f\,\yy)) \in V_\Gamma$ and $(\forall \yy.\, G \Rightarrow P\,(f\,\yy)) \in D_\calG$ such that:
      \[
        \Gamma \backslash \set{ \hasType{f}{\tau} } \text{ is } \calG \text{-coconsistent and }
        \exists U'.\, V_{\Gamma \backslash \set{ \hasType{f}{\tau} }}, \yy \vdash 
G \rew^* U' \land U \leqa U' 
      \] 
      where $\tau = \fromClause{\forall \yy.\, U \Rightarrow P\,(f\,\yy)}$.
  
      By the IH, $\apprx(D_\calG) \leqa V_{\Gamma \backslash \set{ \hasType{f}{\tau} }}$.
      Now, \ref{lem:types_automaton_order_to_inclusion} gives us $V_{\Gamma \backslash \set{ \hasType{f}{\tau} }} \subseteq \upclos{\apprx(D_\calG)}$.
      Thus, rewrite $V_{\Gamma \backslash \set{ \hasType{f}{\tau} }}, \yy \vdash G \rew^* U'$ gives us $\upclos{\apprx(D_\calG)}, \yy \vdash G \rew^* U'$.  
      By~\ref{lem:types_upwards_closed_inclusion}, $\apprx(D_\calG), \yy \vdash G \rew^* U''$ for some $U''$ such that $U' \leqa U''$.
      Transitivity yields $U \leqa U''$.
  
      Since $V_\Gamma = V_{\Gamma \backslash \set{ \hasType{f}{\tau} }} \land \toClause{\tau}{f}$, it remains to show that $\apprx(D_\calG) \leqa \toClause{\tau}{f}$, which is $\apprx(D_\calG) \leqa \forall \yy.\, U \Rightarrow P\,(f\,\yy)$.  
      Now, $(\forall \yy.\, G \Rightarrow P\,(f\,\yy)) \in D_\calG$ and $\apprx(D_\calG), \yy \vdash G \rew^* U''$ satisfy the side conditions of the rewrite algorithm for including $(\forall \yy.\, U'' \Rightarrow P\,(f\,\yy))$ in $\apprx(D_\calG)$.
  
      Finally, $U \leqa U''$ implies $(\forall \yy.\, U'' \Rightarrow P\,(f\,\yy)) \leqa (\forall \yy.\, U \Rightarrow P\,(f\,\yy))$, so $(\forall \yy.\, U'' \Rightarrow P\,(f\,\yy)) \in\apprx(D_\calG)$ gives us $\apprx(D_\calG) \leqa (\forall \yy.\, U \Rightarrow P\,(f\,\yy))$, which concludes the proof.
  
    \item
  
      Suppose $\Gamma \vdash \hasType{t}{\tau}$ and $\Gamma$ is $\calG$-coconsistent.
      By~\ref{lem:types_coconsistent_vs_automaton}, $\apprx(D_\calG) \leqa V_\Gamma$, which means $\Gamma_{\apprx(D_\calG)} \leq \Gamma_{V_\Gamma} = \Gamma$.  
      By Lemma 3.2.3 from \citet{RamsayThesis}, there exists $\tau'$ such that $\Gamma_{\apprx(D_\calG)} \vdash \hasType{t}{\tau'}$ and $\tau' \leq \tau$.
      
  \end{enumerate}
  
\end{proof}

\untypeabilityToProvability*
\begin{proof}
  We prove the following: 
  \[
    \exists \calG \text{-coconsistent } \Gamma.\, \Gamma \vdash \hasType{t}{\tau}
    \text{ if, and only if, }
    \apprx(D_\calG) \vdash \toClause{\tau}{t} \rew^* \truetm
  \]

  For direction $\Leftarrow$, suppose $\apprx(D_\calG) \vdash \toClause{\tau}{t} \rew^* \truetm$.
  Since $\typeAppr{D_\calG} = \Gamma_{\apprx(D_\calG)}$, Proposition~\ref{cor:type-clause} gives us $\typeAppr{D_\calG} \vdash \hasType{t}{\tau}$.
  We know from~\ref{lem:types_typeAppr_cocons} that $\typeAppr{D_\calG}$ is $\calG$-coconsistent, which proves the claim.
  
  For direction $\Rightarrow$, suppose there exists a $\calG$-coconsistent $\Gamma$ such that $\Gamma \vdash \hasType{t}{\tau}$.
  By \ref{lem:P-coconsistent_underapprox}, there exists $\tau'$ such that $\Gamma_{\apprx(D_\calG)} \vdash \hasType{t}{\tau'}$ and $\tau' \leq \tau$.
  Proposition~\ref{cor:type-clause} gives us $\apprx(D_\calG) \vdash \toClause{\tau'}{t} \rew^* \truetm$.
  Finally, Lemma~\ref{lem:monotonicity_of_rewrite} gives us $\apprx(D_\calG) \vdash \toClause{\tau}{t} \rew^* \truetm$.
\end{proof}

\section{Supporting Material for Section~\ref{sec:application}}\label{appx:application}

Figure~\ref{fig:socket-clauses} lists the definite clauses that capture the behaviour of sockets in our automated verifier for socket-manipulating Haskell programs.
These program-independent clauses are included in the \MSL($\omega$) definite clause that captures the socket usage in the input Haskell program.

\begin{figure}[ht]

  \begin{minipage}{0.49\textwidth}
    \begin{align*}
      \textsf{Untracked}\ (\mathtt{socket}\ k)                 & \Leftarrow \textsf{Ready}\ (k\ \underline{s})     \\
      \textsf{Untracked}\ (\mathtt{socket}\ k)                 & \Leftarrow \textsf{Untracked}\ (k\ \underline{u}) \\
      \textsf{Untracked}\ (\mathtt{bind}\ k\ \underline{u})    & \Leftarrow \textsf{Untracked}\ k                  \\
      \textsf{Untracked}\ (\mathtt{connect}\ k\ \underline{u}) & \Leftarrow \textsf{Untracked}\ k                  \\
      \textsf{Untracked}\ (\mathtt{listen}\ k\ \underline{u})  & \Leftarrow \textsf{Untracked}\ k                  \\
      \textsf{Untracked}\ (\mathtt{accept}\ k\ \underline{u})  & \Leftarrow \textsf{Open}\ (k\ \underline{s})      \\
      \textsf{Untracked}\ (\mathtt{accept}\ k\ \underline{u})  & \Leftarrow \textsf{Untracked}\ (k\ \underline{u}) \\
      \textsf{Untracked}\ (\mathtt{close}\ k\ \underline{u})   & \Leftarrow \textsf{Untracked}\ k                  \\
      \textsf{Untracked}\ (\mathtt{send}\ k\ \underline{u})    & \Leftarrow \textsf{Untracked}\ k                  \\
      \textsf{Untracked}\ (\mathtt{receive}\ k\ \underline{u}) & \Leftarrow \textsf{Untracked}\ k                  \\
      \\
      \textsf{Ready}\ (\mathtt{bind}\ k\ \underline{s})        & \Leftarrow \textsf{Bound}\ (k)                    \\
      \textsf{Ready}\ (\mathtt{connect}\ k\ \underline{s})     & \Leftarrow \textsf{Open}\ (k)                     \\
      \textsf{Ready}\ (\mathtt{listen}\ k\ \underline{s})      & \Leftarrow                                        \\
      \textsf{Ready}\ (\mathtt{accept}\ k\ \underline{s})      & \Leftarrow                                        \\
      \textsf{Ready}\ (\mathtt{close}\ k\ \underline{s})       & \Leftarrow                                        \\
      \textsf{Ready}\ (\mathtt{send}\ k\ \underline{s})        & \Leftarrow                                        \\
      \textsf{Ready}\ (\mathtt{receive}\ k\ \underline{s})     & \Leftarrow                                        \\
      \\
      \textsf{Bound}\ (\mathtt{bind}\ k\ \underline{s})        & \Leftarrow                                        \\
      \textsf{Bound}\ (\mathtt{connect}\ k\ \underline{s})     & \Leftarrow                                        \\
      \textsf{Bound}\ (\mathtt{listen}\ k\ \underline{s})      & \Leftarrow   \textsf{Listening}\ (k)              \\
      \textsf{Bound}\ (\mathtt{accept}\ k\ \underline{s})      & \Leftarrow                                        \\
      \textsf{Bound}\ (\mathtt{close}\ k\ \underline{s})       & \Leftarrow                                        \\
      \textsf{Bound}\ (\mathtt{send}\ k\ \underline{s})        & \Leftarrow                                        \\
      \textsf{Bound}\ (\mathtt{receive}\ k\ \underline{s})     & \Leftarrow                                        \\
      \\
      \textsf{Listening}\ (\mathtt{bind}\ k\ \underline{s})    & \Leftarrow                                        \\
      \textsf{Listening}\ (\mathtt{connect}\ k\ \underline{s}) & \Leftarrow                                        \\
      \textsf{Listening}\ (\mathtt{listen}\ k\ \underline{s})  & \Leftarrow                                        \\
      \textsf{Listening}\ (\mathtt{accept}\ k\ \underline{s})  & \Leftarrow   \textsf{Listening}\ (k\ \underline{u})  \\
      \textsf{Listening}\ (\mathtt{close}\ k\ \underline{s})   & \Leftarrow                                        \\
      \textsf{Listening}\ (\mathtt{send}\ k\ \underline{s})    & \Leftarrow                                        \\
      \textsf{Listening}\ (\mathtt{receive}\ k\ \underline{s}) & \Leftarrow                                        \\
    \end{align*}
  \end{minipage}
  \begin{minipage}{0.49\textwidth}
    \begin{align*}
      \textsf{Open}\ (\mathtt{bind}\ k\ \underline{s})                        & \Leftarrow                                    \\
      \textsf{Open}\ (\mathtt{connect}\ k\ \underline{s})                     & \Leftarrow                                    \\
      \textsf{Open}\ (\mathtt{listen}\ k\ \underline{s})                      & \Leftarrow                                    \\
      \textsf{Open}\ (\mathtt{accept}\ k\ \underline{s})                      & \Leftarrow                                    \\
      \textsf{Open}\ (\mathtt{close}\ k\ \underline{s})                       & \Leftarrow    \textsf{Close}\ (k)             \\
      \textsf{Open}\ (\mathtt{send}\ k\ \underline{s})                        & \Leftarrow    \textsf{Open}\ (k)              \\
      \textsf{Open}\ (\mathtt{receive}\ k\ \underline{s})                     & \Leftarrow                 \textsf{Open}\ (k) \\
      \\
      \textsf{Close}\ (\mathtt{bind}\ k\ \underline{s})                       & \Leftarrow                                    \\
      \textsf{Close}\ (\mathtt{connect}\ k\ \underline{s})                    & \Leftarrow                                    \\
      \textsf{Close}\ (\mathtt{listen}\ k\ \underline{s})                     & \Leftarrow                                    \\
      \textsf{Close}\ (\mathtt{accept}\ k\ \underline{s})                     & \Leftarrow                                    \\
      \textsf{Close}\ (\mathtt{close}\ k\ \underline{s})                      & \Leftarrow                                    \\
      \textsf{Close}\ (\mathtt{send}\ k\ \underline{s})                       & \Leftarrow                                    \\
      \textsf{Close}\ (\mathtt{receive}\ k\ \underline{s})                    & \Leftarrow                                    \\
      \\
      \forall\ q \in Q \setminus \{ \textsf{Untracked} \}. \hspace{40pt} &                                               \\
      q\ (\mathtt{socket}\ k)                                                 & \Leftarrow q\ (k\ \underline{u})              \\
      q\ (\mathtt{bind}\ k\ \underline{u})                                    & \Leftarrow q\ (k)                             \\
      q\ (\mathtt{connect}\ k\ \underline{u})                                 & \Leftarrow q\ (k)                             \\
      q\ (\mathtt{listen}\ k\ \underline{u})                                  & \Leftarrow q\ (k)                             \\
      q\ (\mathtt{accept}\ k\ \underline{u})                                  & \Leftarrow q\ (k)                             \\
      q\ (\mathtt{close}\ k\ \underline{u})                                   & \Leftarrow q\ (k)                             \\
      q\ (\mathtt{send}\ k\ \underline{u})                                    & \Leftarrow q\ (k)                             \\
      q\ (\mathtt{receive}\ k\ \underline{u})                                 & \Leftarrow q\ (k)                             \\
    \end{align*}
  \end{minipage}
  \caption{Clauses defining the behaviour of socket programs}\label{fig:socket-clauses}
\end{figure}

\end{document}